\documentclass[letter,12pt]{article}%
\usepackage{booktabs}
\usepackage{amsmath,bm}
\usepackage{amsfonts}
\usepackage{amssymb}
\usepackage{graphicx}
\usepackage{rotating}
\usepackage{multirow}
\usepackage{dsfont}
\usepackage{listings}
\usepackage{xcolor}  
\usepackage{comment}
\usepackage{bibunits}
\usepackage{algorithm}
\usepackage{algpseudocode}

\usepackage[height=9in,left=1in,right=1in,bottom=1in]{geometry}
\usepackage[breaklinks=true,bookmarksopen=true,colorlinks=true,citecolor=blue]%
{hyperref}
\usepackage{setspace}
\usepackage{natbib}
\usepackage{color}
\usepackage{colortbl}
\usepackage{paralist}
\usepackage{amsmath}
\usepackage{tikz}
\usepackage{lscape}%
\usepackage{subcaption}
\usepackage{threeparttablex}
\usepackage{float}
\usepackage[labelfont=bf]{caption}
\providecommand{\U}[1]{\protect\rule{.1in}{.1in}}
\usetikzlibrary{positioning}
\RequirePackage{hypernat}
\newtheorem{theorem}{Theorem}

\newtheorem{assumption}{Assumption}

\newtheorem{corollary}{Corollary}[theorem]

\newenvironment{example}[1][]{%
  \par\medskip\noindent\textbf{Example%
  \ifx\relax#1\relax\else\ (#1)\fi.}\ \itshape%
}{%
  \par\medskip
}

\newtheorem{lemma}{Lemma}

\newtheorem{proposition}{Proposition}
\newtheorem{remark}{Remark}[section]

\newenvironment{proof}[1][Proof]{\noindent \textbf{#1:} }{\  \rule{0.5em}{0.5em}}

\catcode`~=11
\newcommand{\urltilde}{\kern -.15em\lower .7ex\hbox{~}\kern .04em}
\catcode`~=13
\def \@seccntformat#1{\csname the#1\endcsname.\quad}
\makeatother
\numberwithin{equation}{section}
\hypersetup{pdftitle={EscancianoTerschuur}, pdfsubject={Debiased Machine Learning U-Statistics}, pdfauthor={Juan Carlos Escanciano}, pdfkeywords={Debiased Machine Learning; Double Machine Learning; Inequality of Opportunity; Machine Learning; U-statistics; AUC; pairwise difference estimators} }
\begin{document}
\begin{bibunit}[ectabib]   
\title{Debiased Machine Learning U-Statistics\thanks{\textbf{Escanciano}
(corresponding author): Universidad Carlos III de Madrid
(jescanci@eco.uc3m.es). \textbf{Terschuur}: Technical University Munich
(joel.terschuur@tum.de). Research supported by MICIN/AEI/10.13039/501100011033, grant
CEX2021-001181-M, Comunidad de Madrid, grants EPUC3M11 (V PRICIT) and
H2019/HUM-5891, and grant PID2021-127794NB-I00 (MCI/AEI/FEDER, UE). We thank conference participants at numerous institutions, the Editor and three anonymous referees
for useful comments. This paper is a revised version of the arXiv preprint ``Machine learning inference on inequality of opportunity'', arXiv:2206.05235.}}
\author{Juan Carlos Escanciano
\and Jo\"{e}l R. Terschuur}
\date{\today}
\maketitle
\begin{abstract}
We propose a method to debias estimators based on U-statistics with machine-learning (ML) first steps. Standard plug-in estimators often suffer from regularization and model-selection biases, leading to invalid inference. We characterize orthogonal adjustment terms for two-step U-statistics, construct Debiased Machine Learning (DML) U-estimators, develop a cross-fitting algorithm, and establish a general asymptotic theory for inference with ML first steps. We illustrate the methodology with applications to Inequality of Opportunity (IOp), the Area Under the ROC Curve, and conditional moment restrictions. For the latter setting, we introduce Kernel-DML, a class of orthogonal estimators based on identification-preserving kernel distance criteria, and apply it to semiparametric production function estimation. Using European survey data, we provide debiased ML-based estimates of income IOp and show that the widely used Conditional Inference Forest plug-in approach can substantially underestimate IOp. Moreover, plug-in results vary substantially across ML methods, whereas the debiased estimates are considerably more stable.



\textbf{R package:} \url{https://joelters.github.io/home/code/}
\end{abstract}

\section{Introduction}

A fast-growing literature in economics and other sciences uses Machine Learning (ML) to estimate high-dimensional first-step nuisance functions in two-step procedures. While ML improves flexibility, regularization and model selection can introduce first-step bias that contaminates downstream inference.


This concern has motivated a large literature on Debiased Machine Learning (DML); see, among many others, \cite{chernozhukov2015valid, chernozhukov2018double, chernozhukov2022locally}. DML methods reduce second-step bias by constructing Neyman-orthogonal moments, whose derivatives with respect to the first-step limit vanish. Combined with cross-fitting, these methods are effective when the target parameter is a population average of a function of the observed data and the ML first-step limit, as in average treatment effects and related policy parameters.


Many parameters of economic interest, however, are more complex than population averages over a single observation. This paper takes the next natural step by considering targets involving pairs of independent observations. These ``double population averages'' give rise to U-statistics (see \cite{lee2019u}) and cover relevant problems such as inequality of opportunity, ML-based classification evaluation metrics, and conditional moment restrictions with ML first-steps such as those used in production function estimation.


To introduce the main problem, let $W_i$ denote an observation and let $\gamma$ denote a high-dimensional or ML population first-step. Standard DML settings consider targets such as
\[
\theta(\gamma)=E[m(W_i,\gamma)],
\]
for a known function $m$. In contrast, this paper considers pairwise targets of the form
\[
\theta(\gamma)=E[g(W_i,W_j,\gamma)],
\]
for a symmetric and known moment function $g$ involving pairs of
independent observations.\footnote{Note that our setting nests the previous setting as a special case since we can always
write $\theta(\gamma)=\mathbb{E}[m(W_{i},\gamma)]=\mathbb{E}[g(W_{i}%
,W_{j},\gamma)],$ for $g(W_{i},W_{j},\gamma)=\left[  m(W_{i},\gamma
)+m(W_{j},\gamma)\right]  /2,$ provided $W_{i}$ and $W_{j}$ have the same
distribution.} 


In both cases, the main econometric challenge is that the asymptotic behavior of two-step estimators $\hat{\theta}(\hat{\gamma})$ depends critically on
\begin{equation}
\frac{\partial\theta(\gamma_{0})}{\partial\gamma}\left[  \hat{\gamma}%
-\gamma_{0}\right]  ,
\label{main}%
\end{equation}
where $\hat{\gamma}$ is the first-step estimator and $\gamma_{0}$ its probabilistic limit.\footnote{Heuristically, $\hat{\theta}(\hat{\gamma})\approx\hat{\theta}(\gamma_{0})+\frac{\partial\theta(\gamma_{0})}{\partial\gamma}\left[  \hat{\gamma}%
-\gamma_{0}\right]$, see Section 6 in \cite{newey1994large}.} The derivative is defined formally in Section 2.1. Because ML first-steps often have non-negligible bias from regularization and model selection, this term can invalidate inference unless it is zero or sufficiently small and lead to potentially misleading empirical conclusions.


Moment functions $m$ satisfying this zero-derivative property are called Neyman-orthogonal.\footnote{We use the terms Neyman-orthogonal, orthogonal, locally robust or debiased interchangeably. } \cite{chernozhukov2022locally} show how to construct such locally robust moment functions for parameters that solve Generalized Method of Moments (GMM) criteria. We adapt this construction to quadratic, or U-statistic, functionals of the form $\theta(\gamma)=\mathbb{E}[g(W_i,W_j,\gamma)]$. We refer to the resulting cross-fitted estimators as DML U-estimators.



Our main contribution is to extend DML from population-average moments to U-statistics. We construct debiased moments for two-step U-estimators, introduce a cross-fitting procedure adapted to pairwise estimators, and establish asymptotic normality and valid inference under high-level conditions allowing for high-dimensional ML first steps. The framework applies to a broad class of U-statistic problems arising in economics, including inequality of opportunity, ML-based classification metrics, and conditional moment restrictions. 


Our first application is to Inequality of Opportunity (IOp): the share of inequality explained by circumstances beyond individual control, such as parental income, education, sex, or social origin. Although intergenerational mobility measures remain more common because of their simpler econometric structure, IOp measures allow for richer circumstance sets and are increasingly estimated with flexible ML methods; see, for example, \cite{brunori2019inequality, brunori2021evolution, brunori2021roots, carranza2022, hufe2022lower, salas2022inheritances, bernardo2025model}.

We focus on one of the most widely used IOp measures: the Gini coefficient of predicted income from circumstances. We propose a simple debiased estimator that is valid for a broad class of ML first-steps and establish its asymptotic normality. The main technical challenge is the non-smoothness of the Gini functional, which complicates the control of nonlinear remainder terms. To address this challenge, we develop a novel analysis that reveals a trade-off between the smoothness of the target functional and the convergence rates required of the ML first-steps.


Our second application provides inference on the Area Under the Receiver Operating Characteristic Curve (AUC), arguably the leading scalar performance measure for binary classifiers in economics and, more broadly, in empirical work using machine learning.\footnote{See, among many others, \cite{einav2018predictive, kleinberg2018human, herrera2020political, baron2021banking, jorda2021bank, aydin2022consumption, bazzi2022promise, broner2022fiscal, chan2022selection, dellavigna2024bottlenecks, ludwig2024machine, muller2024credit}.} Despite the ubiquity of AUC in ML classification, existing inference theory largely treats the classifier score as fixed or observed, or accounts for score estimation only in parametric and low-dimensional settings (see \cite{hanley1982meaning,newson2006confidence,lahiri2018confidence,hsu2021inference,feng2026maximum}). Thus, existing results do not cover the canonical modern use case. We show that the AUC admits a locally robust representation as a functional of our debiased Gini estimand and the share of observations in the positive class. This yields, to the best of our knowledge, the first general inference theory for AUC with generic nonparametric ML classifiers.


Our third application considers parameters identified by conditional moment restrictions with ML first-steps. Classical inference for such models has largely relied on sieve, kernel smoothing, minimum-distance, and two-step GMM methods; see, among others, \citet{ai2003efficient}, \citet{ackerberg2014asymptotic}, \citet{escanciano2014uniform}  and references therein. We introduce Kernel-DML, a class of orthogonal estimators based on identification-preserving kernel distance criteria. We illustrate the method in semiparametric production function estimation with proxy variables, as in \citet{olley1996dynamics}, \citet{levinsohn2003estimating}, and \citet{ackerberg2015identification}, but allowing flexible ML methods to estimate high-dimensional nuisance components. The kernel criterion preserves the identifying content of the conditional moment restriction, avoiding the potential identification loss that arises when it is reduced to finitely many instruments \citep{dominguez2004consistent}. Our estimator complements the finite-instrument debiased production function approach of \cite{arganaraz2025automatic}.

The simulations and empirical application demonstrate that the gains from debiasing can be substantial in practice. Monte Carlo simulations for the IOp application show that conventional ML-based plug-in estimators suffer from substantial regularization and model-selection biases, whereas our debiased IOp estimator substantially reduces these distortions and delivers valid inference. Using the 2019 wave of the EU-SILC survey, we provide the first debiased ML-based estimates of IOp across 29 European countries and find that commonly employed plug-in estimators systematically underestimate IOp, with differences reaching up to 10 percentage points in some cases. Our debiased estimates are also substantially more stable across ML methods.

Beyond the applications studied here, U-statistics arise naturally in a variety of econometric problems, including pairwise difference estimators, fairness measures, causal inference, dyadic data, and policy learning. We expect the orthogonalization and cross-fitting ideas developed in this paper to be useful in some of these settings, as well as in future work on higher-order debiasing.


The remainder of the paper is organized as follows. Section \ref{sec_methodology} presents the methodology and cross-fitting procedure. Section \ref{sec_ufsif_construction} introduces the construction of orthogonal adjustment terms. Section \ref{sec_AN} establishes the asymptotic theory. Section \ref{sec_App} presents applications to inequality measurement, predictive accuracy evaluation, and conditional moment restrictions. Sections \ref{sec_simulations} and \ref{sec_empapp} report simulation results and the IOp empirical application, respectively. Proofs and additional results are collected in the appendices and Online Appendix.
\section{Methodology}

\label{sec_methodology}

We now introduce the methodology in a more general setting. Let $W_{i}%
=(Y_{i},X_{i})$, for $i=1,...,n$, be independent and identically distributed (i.i.d.) data with unknown distribution
$F_{0}$. Let $\gamma_{0}$ be an unknown first-step function and $\theta_0 \in \Theta \subseteq \mathbb{R}^d$ a finite-dimensional
parameter of interest. We form pairs $(W_{i},W_{j}),$ where $W_{j}$ is an independent copy of
$W_{i}$, with realizations $(w_{i},w_{j})$. Let $\mathbb{E}[\cdot]$ denote the
expectation under $F_{0}$. We assume there is
a vector $g(w_{i},w_{j},\gamma,\theta)$ of $p$ known identifying moment
functions such that
\[
\mathbb{E}[g(W_{i},W_{j},\gamma_{0},\theta)]=0\iff\theta=\theta_{0}\in\Theta.
\]
Without loss of generality, we can assume that $g$ is a symmetric function in
$w_{i}$ and $w_{j}$.\footnote{Otherwise, replace $g(w_{i},w_{j},\gamma
_{0},\theta)$ by $g^{*}(w_{i},w_{j},\gamma_{0}%
,\theta)=(1/2)[g(w_{i},w_{j},\gamma_{0},\theta)+g(w_{j},w_{i},\gamma
_{0},\theta)]$.} By independence between $W_{i}$
and $W_{j}$,
\begin{equation}
\label{quadratic}
\mathbb{E}[g(W_{i},W_{j},\gamma_{0},\theta)]=\int\int g(w_{i},w_{j},\gamma
_{0},\theta)F_{0}(dw_{i})F_{0}(dw_{j}).
\end{equation}
A key feature of (\ref{quadratic}) is that it is a \emph{quadratic functional} of $F_{0}$.
This is what differentiates our analysis from the standard debiasing
literature (e.g., \cite{chernozhukov2018double}). It naturally gives rise 
to the U-statistic\footnote{This is called a second-order U-statistic. A discussion on the extension of our results to $r-$order U-statistics, $r\geq2$, is provided in the Online Appendix \ref{app_r_order}.}
\[
U_{n}g(\cdot,\gamma_{0},\theta)=\frac{2}{n(n-1)}\sum_{i<j}g(W_{i},W_{j}%
,\gamma_{0},\theta),
\]
where, henceforth, we use the short notation $\sum_{i<j}=\sum_{i=1}^{n-1}%
\sum_{j=i+1}^{n}$. Plugging in a first-step estimator $\hat{\gamma}$ for $\gamma_{0}$ leads to the plug-in estimator $\hat{\theta}^{P}$ solving
\begin{equation}
\label{U-estimator}
U_{n}g(\cdot,\hat{\gamma},\hat{\theta}^{P})=0.
\end{equation}
Plug-in estimators are generally biased by model selection and/or regularization in the first-step. To explain the source of this bias heuristically, consider a first-order expansion of the estimating equation (\ref{U-estimator}) around $(\gamma_{0},\theta_{0})$:
\[
0\approx U_{n}g(\cdot,\gamma_{0},\theta_{0})
+\frac{\partial U_{n}g(\cdot,\gamma_{0},\theta_{0})}{\partial\theta}
(\hat{\theta}^{P}-\theta_{0})
+\frac{\partial U_{n}g(\cdot,\gamma_{0},\theta_{0})}{\partial\gamma}
[\hat{\gamma}-\gamma_{0}].
\]
Rearranging,
\[
\hat{\theta}^{P}-\theta_{0}
\approx
-\left(
\frac{\partial U_{n}g(\cdot,\gamma_{0},\theta_{0})}{\partial\theta}
\right)^{-1}
\left[
U_{n}g(\cdot,\gamma_{0},\theta_{0})
+
\frac{\partial U_{n}g(\cdot,\gamma_{0},\theta_{0})}{\partial\gamma}
[\hat{\gamma}-\gamma_{0}]
\right].
\]
The first term in brackets is the infeasible estimator using the true first-step $\gamma_{0}$. The second term captures the propagation of first-step estimation error into the second-step estimator. Due to regularization bias, this second term can dominate the asymptotic behavior of the plug-in estimator unless the derivative with respect to $\gamma$ is zero or sufficiently small.\footnote{Intuitively, ML algorithms trade off bias and variance, allowing some first-step bias if it improves prediction accuracy. This first-step bias can then propagate into the second-step estimator, distorting the estimation of the target parameter.
} In the population, the same bias propagation can be seen by differentiating the identifying restriction $\mathbb{E}[g(W_{i},W_{j},\gamma,\theta(\gamma))]=0$ with respect to $\gamma$ at $\gamma_{0}$, yielding
\begin{equation}
\label{deriv}
\frac{\partial\theta(\gamma_{0})}{\partial\gamma}=-\left[  \frac{\partial
\mathbb{E}[g(W_{i},W_{j},\gamma_{0},\theta_0)]}{\partial\theta}\right]
^{-1} \frac{\partial}{\partial\gamma}%
\mathbb{E}[g(W_{i},W_{j},\gamma_{0},\theta_0)].
\end{equation}
The first factor in (\ref{deriv}) is a standard derivative, whereas the second factor is a functional derivative whose formal definition requires some care. The next section formalizes these heuristics and shows how to construct identifying moments with a zero functional derivative.

\subsection{Construction of debiased moments}
\label{Construction}

To reduce the first-step bias, we aim to find an adjustment term $\phi$ for the original moment function $g$ such that $\psi=g+\phi$ satisfies $\mathbb{E}[\psi(W_{i},W_{j},\gamma_{0},\theta)]=0 \Longleftrightarrow \theta=\theta_0$, and
\begin{equation}
\label{deriv0}
\frac{\partial}{\partial\gamma}\mathbb{E}[\psi(W_{i},W_{j},\gamma_{0},\theta_0)]=0,
\end{equation}
so the derivative in (\ref{deriv}) with $g$ replaced by $\psi$ is zero. We now formalize this statement and show how such $\phi$ can be found. Let $F_0$ again
be the cumulative distribution function of $W_{i}$ and let $\gamma(F)$ denote the probabilistic limit of the first-step $\hat{\gamma}$ when $W_{i}$ is distributed as $F$, as in
\cite{newey1994asymptotic}, with $F$ unrestricted except for regularity conditions like existence of certain moments. For
example, for $\gamma(F)=\mathbb{E}_{F}[Y|X]$ we require $\mathbb{E}%
_{F}[|Y|]<\infty$, where $\mathbb{E}_{F}$ denotes expectation under $F.$  Let $F_{\tau}=F_{0}+\tau(H-F_{0})$ be a deviation from $F_{0}$ along some
alternative distribution $H$, with $\tau\in\lbrack0,1]$. $H$ is unrestricted except that $\gamma_{\tau}=\gamma(F_{\tau})$ exists
for $\tau$ small enough and other regularity conditions are satisfied, as in \cite{ichimura2022influence}. We aim to find a function $\phi$ such that
\begin{equation}
\mathbb{E}_{F_{\tau}}[\phi(W_{i},W_{j},\gamma(F_{\tau}),\alpha(F_{\tau
},\theta),\theta)]=0\text{ for all }\tau\in\lbrack0,\bar{\tau}),\text{ }\bar{\tau}\in[0,1],
\label{1}%
\end{equation}
and
\begin{equation}
\frac{d}{d\tau}\mathbb{E}[g(W_{i},W_{j},\gamma(F_{\tau}),\theta)]
=
\int\int
\phi(w_{i},w_{j},\gamma_{0},\alpha_{0}(\theta),\theta)
K_{H}\left(dw_{i},dw_{j}\right),
\label{2}%
\end{equation}
for all $\theta$ and all $H$, with $K_{H}\left(  dw_{i},dw_{j}\right)  =F_{0}(dw_{i})H(dw_{j})+H(dw_{i}%
)F_{0}(dw_{j}).$ Here, $\alpha_{0}$ is an additional nuisance parameter which
satisfies (\ref{2}) and might depend on $\theta$, and henceforth, $d/d\tau$ denotes the derivative from the right with respect to
$\tau$ evaluated at $\tau=0$. We assume without loss of generality that $\phi$ is
symmetric in $w_{i},w_{j}$.\footnote{As with $g,$ setting $\phi^{\ast}%
(w_{i},w_{j},\gamma_{0},\alpha_{0}(\theta),\theta)=(1/2)[\phi(w_{i},w_{j},\gamma
_{0},\alpha_{0}(\theta),\theta)+\phi(w_{j},w_{i},\gamma_{0},\alpha_{0}(\theta),\theta)]$, we
have $\int\int\phi(w_{i},w_{j},\gamma_{0},\alpha_{0}(\theta),\theta)K_{H}%
(dw_{i},dw_{j})=\int\int\phi^{\ast}(w_{i},w_{j},\gamma_{0},\alpha_{0}(\theta)%
,\theta)K_{H}(dw_{i},dw_{j})$ and $\phi^{\ast}$ is symmetric.} 

Equation (\ref{1}) is a zero mean condition on $\phi$. The pathwise derivative in
(\ref{2}) formalizes the heuristic derivative
$\partial \mathbb{E}[g(W_i,W_j,\gamma_0,\theta_0)]/\partial\gamma$
of the previous section. A function $\phi$ can be found by solving the functional equation (\ref{2}). Such a function is generally not unique, since the same pathwise derivative may admit multiple orthogonal representations.\footnote{An extended discussion of this non-uniqueness and alternative representations can be found in a previous version of this paper \citep{escanciano_terschuur_2023}.} Any solution gives the same first-order asymptotic effect of the first-step $\gamma(F)$ on the functional
$\mu(F)=\mathbb{E}[g(W_i,W_j,\gamma(F),\theta)]$
as $F$ varies away from $F_0$ in any direction $H$. In Section
\ref{sec_ufsif_construction} we present new results identifying $\phi$ for
first-steps satisfying orthogonality conditions, including high-dimensional regressions.

To motivate our definition of the adjustment term $\phi$, note that equation
(\ref{1}) is a zero mean condition and it implies by the chain rule
\begin{align}
\frac{d}{d\tau}\mathbb{E}[\phi(W_{i},W_{j},\gamma(F_{\tau}),\alpha(F_{\tau
},\theta),\theta)]  &  =-\frac{d}{d\tau}\mathbb{E}_{F_{\tau}}[\phi(W_{i},W_{j}%
,\gamma(F_{0}),\alpha(F_{0},\theta),\theta)]\nonumber\label{3}\\
&  =-\int\int\phi(w_{i},w_{j},\gamma_{0},\alpha_{0}(\theta),\theta)K_{H}(dw_{i}%
,dw_{j}),
\end{align}
for all $\theta$ and all $H$. Equation (\ref{3}) shows that the effect the first-steps have on $\phi$ \textquotedblleft cancels out\textquotedblright\ with
the effect they have on the original identifying moment in (\ref{2}). Thus, letting $\psi(w_{i},w_{j},\gamma,\alpha(\theta),\theta)=g(w_{i},w_{j}%
,\gamma,\theta)+\phi(w_{i},w_{j},\gamma,\alpha(\theta),\theta)$, the total first-step effect on $\psi$ is zero, i.e., we have local robustness (by (\ref{2}) and (\ref{3}))
\[
\frac{d}{d\tau}\mathbb{E}[\psi(W_{i},W_{j},\gamma(F_{\tau}),\alpha(F_{\tau
},\theta),\theta)]=0.
\]
This local robustness motivates estimators based on the orthogonal moments
$\psi$. 

\subsection{DML U-estimators}
\label{DMLsec}
\subsubsection{Cross-fitting U-estimators}
\label{DMLsec_crossfitting}


Estimation of nuisance parameters and moment conditions with the same
observations can induce an \textquotedblleft overfitting\textquotedblright%
\ bias. Also, ML first-steps usually do not satisfy Donsker
conditions (see \cite{chernozhukov2018double} for discussion). We use cross-fitting to overcome these issues (see
\cite{bickel1982adaptive}, \cite{schick1986asymptotically},
\cite{klaassen1987consistent}, \cite{chernozhukov2018double} and
\cite{chernozhukov2022locally}), but adapted to the U-statistics setting.

In standard cross-fitting, the sample is divided into $K$ folds, each using $n(K-1)/K$ observations for nuisance estimation and $n/K$ observations for moment evaluation. This ensures equal nuisance-sample sizes across folds and that every observation contributes once to the second step. Extending these features to U-statistics is nontrivial because the second step is indexed by pairs rather than observations, as the following example illustrates.

\begin{example}[Failure of standard cross-fitting for U-statistics]
    Let $n = 6$ and split the sample in two folds $G_1 = \{1,2,3\}$ and $G_2 = \{4,5,6\}$. Estimate the moment in $G_1$ plugging in $(\hat{\gamma}_1, \hat{\alpha}_1)$ trained with observations in $G_2$. Then, the second step in this fold uses pairs $(1,2)$, $(1,3)$ and $(2,3)$. Estimating the moment in $G_2$ using $G_1$ for nuisance estimation, we use pairs $(4,5)$, $(4,6)$, and $(5,6)$ in the second step. Pairs $(1,4)$, $(1,5)$, $(1,6)$, $(2,4)$, $(2,5)$, $(2,6)$, $(3,4)$, $(3,5)$, $(3,6)$ are not used to compute the sample moment. Hence, standard cross-fitting leaves many pairs unused and fails to exploit the U-statistic structure.
\end{example}
\noindent Instead of partitioning the set of observations $\mathcal{N}=\{1,\ldots,n\}$, we partition the set of pairs
\[
\mathcal{P}=\{(i,j)\in\mathcal{N}^{2}:i<j\}.
\]
Our goal is to construct a cross-fitting scheme satisfying three properties:
(i) each pair in $\mathcal{P}$ is used exactly once in the second step;
(ii) the nuisance estimators used for a given block are independent of all pairs in that block; and
(iii) all blocks leave the same number of observations available for nuisance estimation.

We partition $\mathcal N$ into $T\le n/2$ equally sized \emph{groups} $G_1,\ldots,G_T$ and then partition the pair set $\mathcal P$ into \emph{blocks}.\footnote{For simplicity, we assume that $n$ is divisible by $2T$. Online Appendix \ref{app_CF} provides an algorithm for arbitrary $n$.} The cardinality of a set $A$ is denoted by $|A|$. There are two types of blocks: for $u,v=1,\ldots,T$,
\[
\text{(Triangles):} \qquad \{(i,j)\in\mathcal P:i<j,\ i,j\in G_u\},
\]
and
\[
\text{(Rectangles):} \qquad \{(i,j)\in\mathcal P:i\in G_u,\ j\in G_v,\ u\neq v\}.
\]
For each block $I_l$, let
\[
S_l=\{i\in\mathcal N:i \text{ appears in some pair in } I_l\},
\]
and estimate $(\hat\gamma_l,\hat\alpha_l(\theta))$ using observations in $S_l^c=\mathcal N\setminus S_l$. Triangles use only observations of one group $G_u$, hence they pair observations within group $G_u$. Rectangles form pairs with observations from two different groups $G_u$ and $G_v$ with $u \neq v$, i.e. rectangles form pairs across groups. Hence, rectangles use twice as many observations as triangles, i.e. $|G_u| + |G_v| = 2|G_u|$. Therefore, to equalize $|S_l|$ across all $l = 1,...,L$, we split the groups used by the rectangle in half so that the rectangle is subdivided into four equal sub-rectangles. Now each sub-rectangle uses $|G_u|/2 + |G_v|/2 = |G_u|$ observations. This yields
\[
L=T(2T-1)
\]
blocks. Although block sizes differ in the number of pairs, each block involves exactly $|S_l| = n/T$ observations and therefore leaves $|S_l^c|=n(T-1)/T$ observations for nuisance estimation. Thus, $T$ plays the role of the number of folds in standard cross-fitting. Algorithm \ref{alg:cf_u} summarizes the procedure.
\begin{algorithm}[H]
\footnotesize
\caption{Cross-Fitting for U-Statistics (Divisible Case)}
\label{alg:cf_u}
\begin{algorithmic}[1]

\Require Sample $\{W_i\}_{i=1}^n$, integer $T \leq n/2$ with $n$ divisible by $2T$
\Ensure Cross-fitted estimator $\hat\theta$

\State Partition $\{1,\dots,n\}$ into $T$ parts $G_1,\dots,G_T$ with $|G_t| = n/T$

\State Initialize block index $l \gets 0$

\For{$t=1$ to $T$}  \Comment{Within-group blocks}
    \State $l \gets l+1$
    \State $I_l \gets \{(i,j): i<j,\ i,j \in G_t\}$
\EndFor

\For{$1 \le u < v \le T$}  \Comment{Between-group blocks}
    \State Partition $G_u$ and $G_v$ into two equal subsets
    \Statex \hspace{1.5em}
    $G_u = A_{uv}^{(1)} \cup A_{uv}^{(2)}$,
    \quad
    $G_v = B_{uv}^{(1)} \cup B_{uv}^{(2)}$

    \For{each combination $(A,B)\in
    \{A_{uv}^{(1)},A_{uv}^{(2)}\}
    \times
    \{B_{uv}^{(1)},B_{uv}^{(2)}\}$}
        \State $l \gets l+1$
        \State $I_l \gets \{(i,j): i\in A,\ j\in B\}$
    \EndFor
\EndFor

\State Let $L \gets l$

\For{$l=1$ to $L$}
    \State Let $S_l$ be the set of unique indices appearing in $I_l$
    \State Estimate $(\hat\gamma_l,\hat\alpha_l(\theta))$ using $\{W_i:i\notin S_l\}$
\EndFor

\State Return $\hat\theta$ solving
\[
\sum_{l=1}^{L}
\sum_{(i,j)\in I_l}
\Bigl[
g(W_i,W_j,\hat\gamma_l,\hat\theta)
+
\phi(W_i,W_j,\hat\gamma_l,\hat\alpha_l(\hat\theta),\hat\theta)
\Bigr]
=0.
\]

\end{algorithmic}
\end{algorithm}
\noindent The next result formalizes the key properties required for valid cross-fitting in the U-statistic setting. Its proof can be found in the Online Appendix \ref{app_CF}.

\begin{proposition}
\label{prop_UCF}
Algorithm \ref{alg:cf_u} satisfies:
(i) Each pair in $\mathcal{P}$ is used exactly once;
(ii) $(\hat\gamma_l,\hat\alpha_l)$ is independent of all pairs in $I_l$ for $l = 1,...,L$;
(iii) each block excludes the same number of observations.
\end{proposition}
While our proposed construction is not unique in satisfying (i)--(iii) in Proposition \ref{prop_UCF}, it provides a transparent algorithm where the choice of $T$ allows for exact control of the number of observations used for nuisance estimation. For an illustration with $n = 12$ and $T = 2$, see Figure \ref{fig_CF}.
\begin{figure}[h]
\centering
\includegraphics[scale = 0.4]{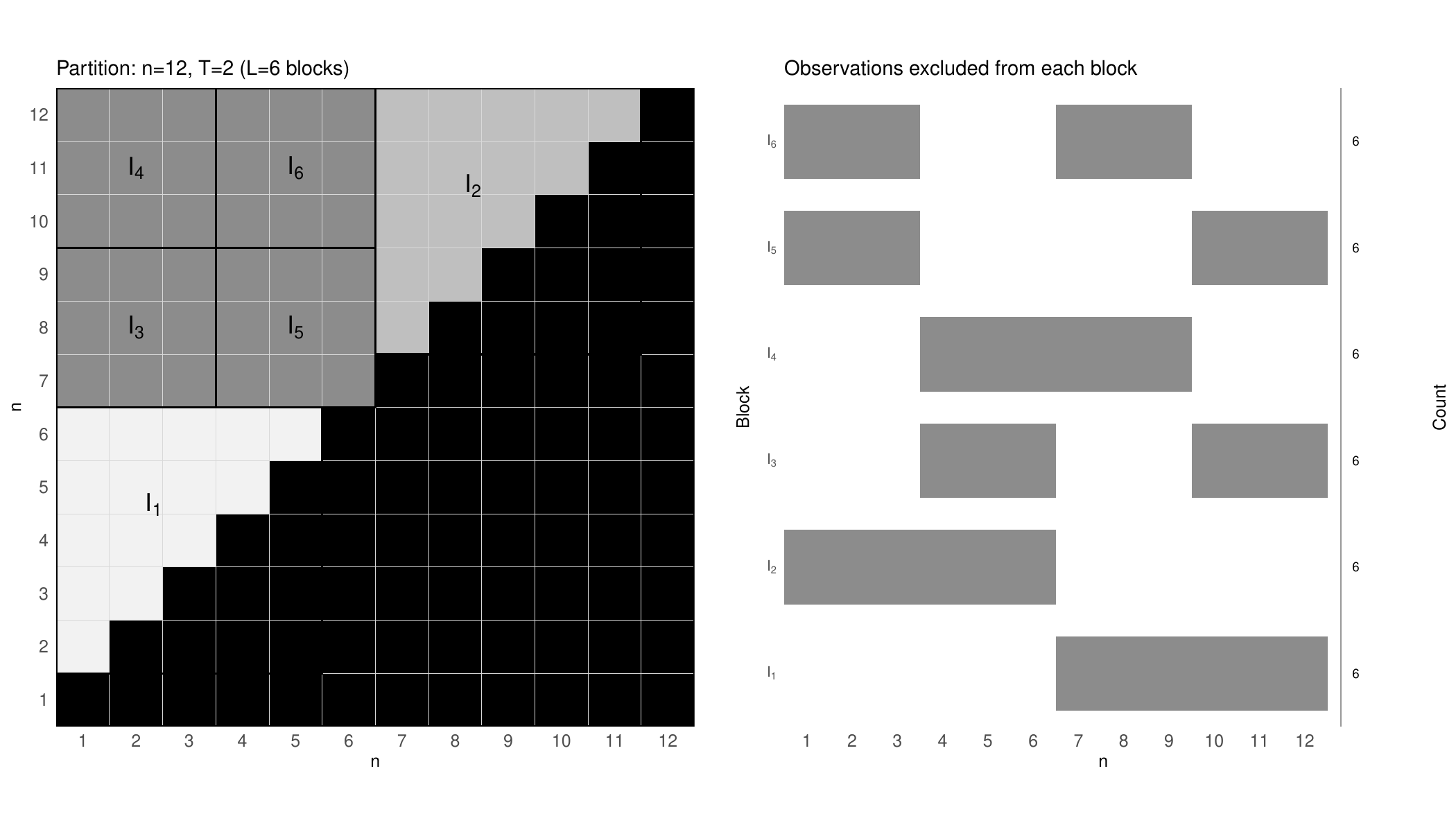}
\caption{\scriptsize Illustration of the proposed cross-fitting partition for $n=12$ and $T=2$. Each pair is used exactly once, while nuisance estimators for a given block are estimated using observations outside that block (shadowed region in the right panel). See Online Appendix \ref{app_CF} for the general construction.}
\label{fig_CF}%
\end{figure}
\subsubsection{DML U-estimators}
\label{DMLsec_DMLU-estimators}
\noindent The debiased sample moment is $\hat{\psi}(\theta)=\hat{g}(\theta)+\hat{\phi}(\theta)$, where 
\[
\hat{g}(\theta)
=\binom{n}{2}^{-1}\sum_{l=1}^{L}\sum_{(i,j)\in I_{l}} g(W_{i},W_{j},\hat{\gamma}_{l},\theta),
\quad
\hat{\phi}(\theta)
=\binom{n}{2}^{-1}\sum_{l=1}^{L}\sum_{(i,j)\in I_{l}} \phi(W_{i},W_{j},\hat{\gamma}_{l},\hat{\alpha}_{l}(\theta),\theta).
\]
Any estimator $\hat\theta$ solving $\hat{\psi}(\hat{\theta})=0$ is called a DML U-estimator.\footnote{\cite{chernozhukov2018double} recommend repeatedly estimating with different random cross-fitting splits and taking the median to improve robustness to the partition, see also \cite{ritzwoller2023reproducible}. These ideas are possible here but may be computationally expensive.} 
\begin{remark}
We could also implement the debiased sample moment as
\[
\hat{\psi}(\theta)=\hat{g}(\theta)+\tilde{\phi},
\]
where $\tilde{\phi}$ replaces $\phi(W_{i},W_{j},\hat{\gamma}_{l},\hat{\alpha}_{l}(\theta),\theta)$ by $\phi(W_{i},W_{j},\hat{\gamma}_{l},\hat{\alpha}_{l}(\tilde{\theta}_{l}),\tilde{\theta}_{l})$ in $\hat{\phi}(\theta)$, with $\tilde{\theta}_{l}$ a preliminary estimator of $\theta_0$ not using observations in $I_{l}$, for example, a cross-fitted plug-in consistent estimator. The asymptotic analysis for this estimator is carried out in Corollary \ref{corollary_preliminary_estimator} in the Online Appendix \ref{Alter_implem}. This is the implementation proposed in \cite{chernozhukov2022locally} for GMM. The alternative implementation proposed in this paper is more convenient in many applications, such as our IOp application, the partially linear model, etc.
\end{remark}
We give sufficient conditions in Section \ref{sec_AN} so that $\sqrt{n}(\hat{\theta}-\theta_{0})$ is asymptotically normally distributed with mean zero and an asymptotic variance $V=B^{-1}\Sigma B^{\prime-1},$ where, henceforth, $B'$ denotes the transpose of a vector or matrix $B$, and
\[
B=\frac{\partial\mathbb{E}[g(W_{i},W_{j},\gamma_{0},\theta_0
)]}{\partial\theta},\quad\Sigma=4\mathbb{V}ar\biggl(\mathbb{E}[\psi
(W_{i},W_{j},\gamma_{0},\alpha_{0},\theta_{0})|W_{i}]\biggr),
\]
and $\alpha_0 \equiv \alpha_0(\theta_0)$. The asymptotic variance $V$ can be estimated by $\hat{V}=\hat{B}^{-1}%
\hat{\Sigma}\hat{B}^{\prime-1},$ where%
\begin{align*}
\hat{B}  &  =\binom{n}{2}^{-1}\sum_{i<j}\frac{\partial}{\partial\theta}%
g(W_{i},W_{j},\hat{\gamma},\hat{\theta}),\\
\hat{\Sigma}  &  =\frac{4}{n(n-1)^{2}}\sum_{i=1}^{n}\biggl[\sum_{j\neq i}%
\psi(W_{i},W_{j},\hat{\gamma},\hat{\alpha}(\hat{\theta}),\hat{\theta})\biggr]\biggl[\sum
_{j\neq i}\psi(W_{i},W_{j},\hat{\gamma},\hat{\alpha}(\hat{\theta)},\hat{\theta
})\biggr]^{\prime}.
\end{align*}
For variance estimation, we can use the whole sample for $\hat{\gamma}$ and $\hat{\alpha}$. It is also possible to cross-fit, but it is computationally more expensive.

The construction of the orthogonal moment requires deriving an adjustment term $\phi$ for the
given identifying moment $g$ satisfying (\ref{1}) and (\ref{2}). The next section develops a general construction of orthogonal adjustment terms $\phi$ for a broad class of two-step U-statistics.

\section{Construction of Adjustment Terms}

\label{sec_ufsif_construction}

The construction of the adjustment term $\phi$ proceeds in two steps. First, we linearize the effect of $\gamma(F)$ on the identifying moment function $g$. Second, we project this linearized effect onto a suitable function space associated with the orthogonality condition defining $\gamma(F)$. This projection removes the first-order effect of estimating $\gamma$.

Henceforth, let $L_2(\mu)$ denote the set of square-integrable functions with respect to a measure $\mu$. Let $P_X$ denote the marginal distribution of $X_i$ under $F_0$. We shall apply the notation $L_2(\mu)$ with $\mu=P_X$ or $\mu=P_X\times P_X$, but when no confusion arises, we simply write $L_2$. We consider functions $\gamma(F)\in\Gamma\subseteq L_2(P_X)$ satisfying
\begin{equation}
\mathbb{E}_{F}\left[\nu(X_i)\left(Y_i-\gamma(F)(X_i)\right)\right]=0
\qquad\text{for all }\nu\in\Gamma.
\label{orth1}%
\end{equation}
This setting is quite general and covers linear, nonparametric, additive, sieve, and high-dimensional regressions, among others. For example, if $\Gamma$ is the mean-square closure of finite linear combinations of a dictionary $(b_k)_{k=1}^\infty$ of square-integrable functions, then $\gamma(F)(X_i)$ corresponds to a high-dimensional regression fit.\footnote{Following the DML literature, see, e.g., \citet{chernozhukov2022locally}, we keep the probability limit of the first-step learner fixed for simplicity. Allowing this limit to vary with the sample size would accommodate settings where the number of active coefficients in the true regression grows with $n$. High-dimensional regressions are nevertheless accommodated through approximate sparsity in a fixed dictionary.}

We impose the following assumption on $\Gamma$.

\begin{assumption}
\label{Ass_S}
$\Gamma\subseteq L_2(P_X)$ is a closed linear subspace that contains the constant functions.
\end{assumption}

For simplicity, we take $\gamma(F)$ to be real-valued. Extensions to settings with multiple nuisance functions,
\[
g(W_i,W_j,\gamma_1(F),\ldots,\gamma_K(F),\theta),
\]
follow by linearity of the pathwise derivative, since the derivative decomposes additively across nuisance components.

We next study the linearization and projection steps separately; see \citet{newey1994asymptotic} for the analogous construction in the GMM setting.

\subsection{Linearization}
\label{linearization}

The following is a standard pathwise smoothness condition. It requires that the first-order effect of perturbing the nuisance function can be represented through finitely many linear functionals of the nuisance path.

\begin{assumption}
\label{Ass_linearization}
There exist functions $\delta_m\in L_2(P_X\times P_X)$, constants $(c_{1m},c_{2m})$, $m=1,\ldots,M$, and an integer $M\geq1$, such that for all $\theta$ and $H$,
\begin{equation}
\frac{d}{d\tau}\mathbb{E}[g(W_i,W_j,\gamma(F_\tau),\theta)]
=
\sum_{m=1}^M
\frac{d}{d\tau}
\mathbb{E}[
\delta_m(X_i,X_j,\gamma_0,\theta)
(c_{1m}\gamma_\tau(X_i)+c_{2m}\gamma_\tau(X_j))
]
\label{linearization}
\end{equation}
where $\gamma_\tau(x)=\gamma(F_\tau)(x)$ and $F_\tau=F_0+\tau(H-F_0)$.
\end{assumption}
For simplicity of notation, we suppress the possible dependence of $c_{1m}$ and $c_{2m}$ on $(\gamma_0,\theta)$. Henceforth, we also use the shorthand notation
\[
\delta_{ij,m}(\gamma,\theta)\equiv\delta_m(X_i,X_j,\gamma,\theta),
\qquad
\delta_{ij,m}\equiv\delta_{ij,m}(\gamma_0,\theta_0).
\]
\subsection{Projection}
\label{RR}

After the linearization step, we construct the adjustment term $\phi$ through a projection argument that extends \citet[p.~1360--1361]{newey1994asymptotic} to the present U-statistic framework. The key idea is to express the linearized effect of the first step through a projection onto a suitable function space associated with the orthogonality condition defining the first-step. This projection delivers an adjustment term that removes the first-order effect of estimating $\gamma$.

We consider two cases. In case (i), the researcher adopts a fully nonparametric specification and sets $\Gamma=L_2(P_X)$ and $\mathcal{S}=L_2(P_X\times P_X)$. In case (ii), $\Gamma\subseteq L_2(P_X)$ satisfies Assumption \ref{Ass_S} and
\[
\mathcal{S}
=
\Gamma+\Gamma
=
\Big\{
(x_i,x_j)\mapsto \gamma_1(x_i)+\gamma_2(x_j):
\gamma_1,\gamma_2\in\Gamma
\Big\}
\subseteq L_2(P_X\times P_X).
\]
Thus, $\mathcal S$ is the class of additive bivariate functions whose components belong to $\Gamma$. Case (ii) is useful when fully nonparametric estimation is infeasible, for example, because $X$ is high-dimensional. Each case gives rise to a nuisance parameter $\alpha_0(\theta) \in\mathcal S$. Henceforth, $\Pi_V(\cdot)$ denotes the orthogonal projection operator onto a closed linear set $V$. We note that $\mathcal{S}$ is a closed linear set, so $\Pi_\mathcal S(\cdot)$ is well-defined.

\begin{lemma}
    \label{lemma_S_closed_linear_set}
    Under Assumption \ref{Ass_S}, $\mathcal{S}$ is a closed linear subspace of $L_2(P_X \times P_X)$.
\end{lemma}
Define the nuisance parameter $\alpha_{0m}$ as the orthogonal projection of $\delta_{ij,m}$ onto $\mathcal S$:
\[
\alpha_{0m}
=
\Pi_{\mathcal S}(\delta_{ij,m}),
\qquad
m=1,\ldots,M.
\]
The vector of nuisance parameters is $\alpha_0=(\alpha_{01},\ldots,\alpha_{0M})'$ and the corresponding adjustment term are characterized in the following lemma. For simplicity of notation, we sometimes suppress the dependence of $\alpha_0(\theta)$ on $\gamma_0$.

\begin{lemma}
\label{qorth2}
The first-step is identified by (\ref{orth1}). If Assumptions \ref{Ass_S}-\ref{Ass_linearization} hold, then
\begin{equation}
\phi(W_i,W_j,\gamma_0,\alpha_0,\theta)
=
\sum_{m=1}^M
\alpha_{0m}(X_i,X_j,\theta)
\left(
c_{1m}Y_i+c_{2m}Y_j
-c_{1m}\gamma_0(X_i)
-c_{2m}\gamma_0(X_j)
\right)
\label{phi}%
\end{equation}
satisfies (\ref{1})--(\ref{2}), where:
\begin{description}
\item[(i)] if $\mathcal S=L_2(P_X\times P_X)$,
\[
\alpha_{0m}(\theta)
=
\delta_{ij,m}(\gamma_0,\theta),
\]

\item[(ii)] if $\mathcal S=\Gamma+\Gamma$,
\begin{equation}
\alpha_{0m}(X_i,X_j,\theta)
=
\Pi_\Gamma \mu_{1m}(X_i,\theta)
+
\Pi_\Gamma \mu_{2m}(X_j,\theta)
-
\mathbb E[\delta_{ij,m}(\gamma_0,\theta)],
\label{alpha}%
\end{equation}
where
\[
\mu_{1m}(x,\theta)
=
\mathbb E[\delta_{ij,m}(\gamma_0,\theta)\mid X_i=x], \quad \mu_{2m}(x,\theta)
=
\mathbb E[\delta_{ij,m}(\gamma_0,\theta)\mid X_j=x].
\]
In many cases, see Applications in Section \ref{sec_App}, $\delta_{ij,m}$ is symmetric or antisymmetric in $(i,j)$ and hence, either $\mu_{1m} = \mu_{2m}$ or $\mu_{1m} = -\mu_{2m}$.
\end{description}
\end{lemma}
The nuisance $\alpha_0$ and its corresponding cross-fitted estimator $\hat\alpha_l$ depend on the specification of $\mathcal S$. When $\alpha_0$ is known up to $\gamma_0$, we can simply set
\[
\hat\alpha_l(\theta)
=
\delta(\hat\gamma_l,\theta).
\]

More generally, $\alpha_0$ can be estimated using any ML method for projection or conditional expectation estimation, including Lasso, neural networks, or sieve methods. For example, using (\ref{alpha}), one may first estimate $(\mu_{1m}(x,\theta), \mu_{2m}(x,\theta))$ for $m=1,...,M$ by
\[
\hat{\mu}_{1ml}(x,\theta)
=
\frac{1}{|S_l^c|}
\sum_{j\in S_l^c}
\delta_{m}(x,X_j,\hat{\gamma}_l,\theta), \qquad 
\hat{\mu}_{2ml}(x,\theta)
=
\frac{1}{|S_l^c|}
\sum_{i\in S_l^c}
\delta_{m}(X_i,x,\hat{\gamma}_l,\theta),
\]
using observations outside block $I_l$, and then project the resulting estimator onto $\Gamma$. Finally, if $\alpha_{0m}=0$ for all $m=1,\ldots,M$, then the original identifying moment is already locally robust and no adjustment term is required ($\phi\equiv0$).

\section{Asymptotic Theory}

\label{sec_AN}
Let $\left\vert \cdot\right\vert $ and $||\cdot||_q$ be the Euclidean and $L_{q}$ norms, respectively. Specifically, for a measurable function $g$ of $W_i$, $||g||_q =(\mathbb{E}[|g(W_i)|^q])^{1/q}$, for $q\geq1$. For simplicity of notation, we drop the subindex for $q=2$, and simply write $||g||$. Also define the supremum norm, $||g||_\infty =\sup_{x\in\mathcal{X}}|g(x)|$, where $\mathcal{X}$ denotes the support of $X_i$. The terms $\rightarrow_{p}$ and $\rightarrow_{d}$ denote convergence in probability and distribution, respectively. Henceforth, $C$ and $a$ are generic positive constants that may change
from expression to expression.

We provide general conditions for the asymptotic normality of DML U-estimators, and verify these conditions for the applications in the next section. To streamline notation, we sometimes use $\phi_{ij}(\gamma,\alpha(\theta),\theta)\equiv\phi(w_{i},w_{j},\gamma,\alpha(\theta),\theta)$. Define the following interaction term
\[
\hat{\xi}_{l}(w_{i},w_{j},\theta)=\phi_{ij}(\hat{\gamma}_{l},\hat{\alpha}%
_{l}(\theta),\theta)-\phi_{ij}(\gamma_{0},\hat{\alpha}_{l}(\theta),\theta
)-\phi_{ij}(\hat{\gamma}_{l},\alpha_{0}(\theta),\theta)+\phi_{ij}(\gamma_{0},\alpha_{0}(\theta),\theta),
\]
and the following norm
\[
\|\alpha - \alpha_0\|^2_{\mathcal{A}} = \sup_{\theta \in \Theta} \int \int |\alpha(x_i,x_j,\theta) - \alpha_0(x_i,x_j,\theta)|^2 F_0(dw_i)F_0(dw_j).
\]
The next theorem provides sufficient conditions for consistency.
\begin{assumption}
\label{Ass_AT_Consist}
(i) For all $\theta\in\Theta$ and $l=1,\ldots,L$,
\[
\int\!\!\int |g(w_i,w_j,\hat\gamma_l,\theta)-g(w_i,w_j,\gamma_0,\theta)|
\,F_0(dw_i)F_0(dw_j)\to_p0,
\]
\[
\int\!\!\int |\phi(w_i,w_j,\hat\gamma_l,\alpha_0(\theta),\theta)
-\phi(w_i,w_j,\gamma_0,\alpha_0(\theta),\theta)|
\,F_0(dw_i)F_0(dw_j)\to_p0,
\]
and
\[
\int\!\!\int |\hat\xi_l(w_i,w_j,\theta)|
\,F_0(dw_i)F_0(dw_j)\to_p0;
\]
(ii) $\Theta$ is compact;
(iii) $\|\hat\gamma_l-\gamma_0\|\to_p0$ and
$\|\hat\alpha_l-\alpha_0\|_{\mathcal A}\to_p0$;
(iv) there exist $a>0$ and $C<\infty$ such that for all $(\gamma,\alpha)$ sufficiently close to $(\gamma_0,\alpha_0)$ and all $\theta,\tilde\theta\in\Theta$,
\begin{align*}
|g(W_i,W_j,\gamma,\tilde\theta)-g(W_i,W_j,\gamma,\theta)|
&\le d_g(W_i,W_j,\gamma)|\tilde\theta-\theta|^a,\\
|\phi(W_i,W_j,\gamma,\alpha(\tilde\theta),\tilde\theta)
-\phi(W_i,W_j,\gamma,\alpha(\theta),\theta)|
&\le d_\phi(W_i,W_j,\gamma,\alpha)|\tilde\theta-\theta|^a,
\end{align*}
with
\[
\mathbb E[d_g(W_i,W_j,\gamma)]<C,
\qquad
\mathbb E[d_\phi(W_i,W_j,\gamma,\alpha)]<C.
\]
\end{assumption}

\begin{theorem}
\label{Thm_AT_Consist}
Under Assumption \ref{Ass_AT_Consist},
\[
\hat\theta\to_p\theta_0.
\]
\end{theorem}
Theorem \ref{Thm_AT_Consist} imposes mild consistency conditions on the first-step estimators and Lipschitz conditions that, together with compactness of $\Theta$,  allow for the uniform convergence needed for consistency of Z-estimators. In many cases, $\phi$ depends on $\theta$ only through an $\alpha_0$ of the form $\alpha_0(x_i,x_j,\theta) = h(\theta)\eta_0(x_i,x_j)$, so that the Lipschitz condition on $\phi$ in Assumption \ref{Ass_AT_Consist}(iv) reduces to a Lipschitz condition on $h(\theta)$ and the consistency condition on $\hat{\alpha}_l$ reduces to continuity of $h(\theta)$, compactness of $\Theta$ and mild consistency of some ML estimator $\hat \eta_l$ of $\eta_0$. This is the case for all our applications. Now we proceed to asymptotic normality.

\begin{assumption}
\label{Ass_AT_equiv} $\mathbb{E}[|\psi(W_{i},W_{j},\gamma_{0},\alpha
_{0},\theta_{0})|^{2}]<\infty$ and for each $l=1,...,L,$

\begin{enumerate}
[(i)]

\item $\int\int|g(w_{i},w_{j},\hat{\gamma}_{l},\theta_0)-g(w_{i}%
,w_{j},\gamma_{0},\theta_0)|^{2}F_{0}(dw_{i})F_{0}(dw_{j})\rightarrow_{p}0$;

\item $\int\int|\phi(w_{i},w_{j},\hat{\gamma}_{l},\alpha_{0},\theta_0%
)-\phi(w_{i},w_{j},\gamma_{0},\alpha_{0},\theta_0)|^{2}F_{0}(dw_{i}%
)F_{0}(dw_{j})\rightarrow_{p}0$;

\item $\int\int|\phi(w_{i},w_{j},\gamma_{0},\hat{\alpha}_{l}(\theta_0),\theta_0%
)-\phi(w_{i},w_{j},\gamma_{0},\alpha_{0},\theta_0)|^{2}F_{0}(dw_{i}%
)F_{0}(dw_{j})\rightarrow_{p}0$.
\end{enumerate}
\end{assumption}
These are mild mean-square consistency conditions for $\hat{\gamma}_{l}$ and
$\hat{\alpha}_{l}$ separately. Unlike for consistency, here it is enough that they hold at $\theta_0$. A linearization argument like equation (3.2)
often implies that the left hand sides of Assumption \ref{Ass_AT_equiv}(i)-(ii)
are bounded above by a constant times $\left\Vert \hat{\gamma}_{l}-\gamma
_{0}\right\Vert,$ so $L_{2}$ consistency suffices. Assumption
\ref{Ass_AT_equiv} (iii) typically follows from $L_{2}$ consistency of
$\hat{\alpha}_{l}$. 
Let $\hat{\xi}_{l}(w_{i},w_{j})\equiv\hat{\xi}_{l}(w_{i},w_{j},\theta_{0})$.
\begin{assumption}
\label{Ass_AT_interaction} For each $l=1,...,L$, either i)
\[
\sqrt{n}\int\int\hat{\xi}_{l}(w_{i},w_{j})F_{0}(dw_{i})F_{0}(dw_{j}%
)\rightarrow_{p}0,\int\int|\hat{\xi}_{l}(w_{i},w_{j})|^{2}F_{0}(dw_{i}%
)F_{0}(dw_{j})\rightarrow_{p}0,
\]
or (ii) $\sqrt{n}\binom{n}{2}^{-1}%
\sum_{(i,j)\in I_{l}}\hat{\xi}_{l}(W_{i},W_{j})\rightarrow_{p}0$.
\end{assumption}
These are rate conditions on the remainder term $\hat{\xi}_{l}(w_{i},w_{j})$.
For $\phi$ in (\ref{phi}), the interaction term is a sum over $m$ of terms of the
form
\[
\hat{\xi}_{lm}(w_{i},w_{j})=\left(  \hat{\alpha}_{lm}(x_{i},x_{j})-\alpha
_{0m}(x_{i},x_{j})\right)  \left(  c_{1m}\hat{\gamma}_{l}(X_{i})+c_{2m}%
\hat{\gamma}_l(X_{j})-c_{1m}\gamma_{0}(X_{i})-c_{2m}\gamma_{0}(X_{j})\right)  .
\]
Therefore, Assumption \ref{Ass_AT_interaction} follows for first-steps
satisfying orthogonality restrictions if $\sqrt{n}||\hat{\alpha}_{lm}%
-\alpha_{0m}||||\hat{\gamma}_{l}-\gamma_{0}||=o_{p}(1).$ This is a product rate
condition which allows for nuisance estimators to converge at slower rates as
long as the product converges faster to zero than the $n^{-1/2}$-rate. Define now $\bar{\psi
}(\gamma,\alpha,\theta) \equiv\mathbb{E}[\psi(W_{i},W_{j},\gamma,\alpha(\theta)
,\theta)]$, i.e. the expectation is only over $(W_i,W_j)$. 
\begin{assumption}
\label{Ass_AT_GRsmallbias} For each $l=1,...,L$ and $\theta$, i) $\int\int%
\phi(w_{i},w_{j},\gamma_{0},\hat{\alpha}_{l}(\theta),\theta)F_{0}(dw_{i})F_{0}%
(dw_{j})=0$ with probability approaching one; and ii) $\sqrt{n}\bar{\psi}(\hat{\gamma
}_{l},\alpha_{0},\theta_{0})\rightarrow_{p}0$.
\end{assumption}
Assumption \ref{Ass_AT_GRsmallbias} (i) incorporates the global robustness
property of $\alpha$ and is in most cases easy to check by inspection of
$\phi$. For example, it holds for $\phi$ in (\ref{phi}).
Assumption \ref{Ass_AT_GRsmallbias} (ii) is a small bias condition. A sufficient condition for Assumption \ref{Ass_AT_GRsmallbias} (ii) is that $\bar{\psi}(\gamma
,\alpha_{0},\theta_{0}) \leq C || \gamma - \gamma_0||_q^\rho$ for all $\gamma$ with $||\gamma - \gamma_0||$ small enough and that $|| \hat{\gamma}_l - \gamma_0||_q = o_p(n^{-1/2\rho})$, for $\rho>0$. 

\begin{remark}
There is a large literature establishing $L_q$-convergence rates for machine-learning estimators under low-level sparsity or smoothness conditions, including kernel and series methods, penalized estimators such as Lasso, boosting, random forests, and deep neural networks; see, e.g., \citet{tsybakov2009introduction,chen2007large,belloni2011,belloni2013least,wager2015adaptive,athey2019generalized,farrell2021deep}. For further general results on $L_q$ rates, see \citet{el_hanchi2023optimal} and \citet{peng2025adversarial}.
\end{remark}

We also need conditions on the Jacobian of the moment condition.

\begin{assumption}
\label{Ass_AT_Jacobianhat}(i) $B = \partial \mathbb{E}[g(W_i,W_j,\gamma_0,\theta_0)]/ \partial \theta$ exists and is invertible; ii) $\hat \psi(\theta)$ is differentiable in a neighborhood $\mathcal{N}$ of $\theta_0$ with probability approaching one; iii) $||\hat{\gamma}_l-\gamma
_{0}||\rightarrow_{p}0$ and $\|\hat{\alpha}_l - \alpha_0\|_\mathcal{A} \to_p 0$; iv) there exists $a > 0$, $C < \infty$ such that for all $(\gamma,\alpha)$ with $\|\gamma - \gamma_0\|$ and $\|\alpha - \alpha_0\|_{\mathcal{A}}$ small enough and $\theta$ in a neighborhood $\mathcal{N}$ of $\theta_0$ 
    \begin{align*}
        \left| \frac{\partial \psi(W_i,W_j, \gamma, \alpha(\theta), \theta)}{\partial \theta} - \frac{\partial \psi(W_i,W_j, \gamma, \alpha(\theta_0), \theta_0)}{\partial \theta_0} \right| \leq d(W_i,W_j,\gamma,\alpha)|\theta - \theta_0|^a, 
    \end{align*}
    with $d(W_i,W_j,\gamma,\alpha)$ not depending on $(\theta, \theta_0)$ and $\mathbb{E}[d(W_i,W_j,\gamma,\alpha(\cdot,\cdot))] < C$; \\
    v) For all $l = 1,...,L$: $\int \int |\partial g(w_i,w_j,\hat \gamma_l, \theta_0)/\partial \theta - \partial g(w_i,w_j, \gamma_0, \theta_0)/\partial \theta| F_0(dw_i) F_0(dw_j) \to_p 0$, $\int \int |\partial \phi(w_i,w_j,\hat \gamma_l, \alpha_0(\theta_0), \theta_0)/\partial \theta - \partial \phi(w_i,w_j, \gamma_0, \alpha_0(\theta_0), \theta_0)/\partial \theta| F_0(dw_i) F_0(dw_j) \to_p 0$, \\ $\int \int |\partial \hat{\xi}_l(w_i,w_j,\theta_0)/\partial \theta| F_0(dw_i) F_0(dw_j) \to_p 0$; (vi) for all $\alpha$ with $\|\alpha - \alpha_0\|_{\mathcal{A}}$ small enough, $\mathbb{E}[|\partial \psi(W_{i},W_{j},\gamma_0,\alpha(\theta_0),\theta_0)/\partial \theta |] < C$.
\end{assumption}

\noindent Now we are ready to give the main asymptotic result of the paper,

\begin{theorem}
\label{Thm_AT_AN} If Assumptions \ref{Ass_AT_equiv}-\ref{Ass_AT_Jacobianhat} hold, $\hat{\theta
}\rightarrow_{p}\theta_{0}$ and $V=B^{-1}\Sigma B^{\prime-1}$ is non-singular,
then
\[
\sqrt{n}(\hat{\theta}-\theta_{0})\rightarrow_{d}\mathcal{N}(0,V).
\]
\end{theorem}
For valid inference, we also need convergence of the asymptotic variance
estimators. To simplify the computation, we implement the variance estimator
without cross-fitting. Define $\hat{g}_{ij}=g(W_{i},W_{j},\hat{\gamma},\hat{\theta})$ and $g_{ij}=g(W_{i},W_{j},\gamma_{0},\theta_{0}).$
Then define the following leave-one out average $\hat{g}_{-i} = (n-1)^{-1}%
\sum_{j\neq i}\hat{g}_{ij}$ and $g_{-i} = (n-1)^{-1}\sum_{j\neq i}g_{ij}$, let
$\hat{\phi}_{-i}$ and $\phi_{-i}$ be defined in the same way.

\begin{assumption}
\label{ass_AT_Sigmahat} (i) $\mathbb{E}[|\psi(W_{i},W_{j},\gamma_{0}%
,\alpha_{0},\theta_{0})|^{2}]<\infty$; (ii) $n^{-1}\sum_{i=1}^{n}|\hat{g}_{-i} -
g_{-i}|^{2} \to_{p} 0$ and $n^{-1}\sum_{i=1}^{n}|\hat{\phi}_{-i} - \phi
_{-i}|^{2} \to_{p} 0$; iii) $\|\hat{\gamma} - \gamma_0\| \to_p 0$.
\end{assumption}

\begin{proposition}
    \label{prop_Sigmahat_consistency}
    Under Assumption \ref{ass_AT_Sigmahat} and (i) $B = \partial \mathbb{E}[g(W_i,W_j,\gamma_0,\theta_0)]/ \partial \theta$ exists and is invertible; (ii) $\binom{n}{2}^{-1} \sum_{i<j } g(W_i,W_j,\hat \gamma, \theta)$ is differentiable in a neighborhood $\mathcal{N}$ of $\theta_0$ with probability approaching one; (iii) there exists $a > 0$, $C < \infty$ such that for all $\gamma$ with $\|\gamma - \gamma_0\|$ small enough and $\theta$ in a neighborhood $\mathcal{N}$ of $\theta_0$ 
    \begin{align*}
        \left| \frac{\partial g(W_i,W_j, \gamma,\theta)}{\partial \theta} - \frac{\partial g(W_i,W_j, \gamma, \theta_0)}{\partial \theta_0} \right| \leq d(W_i,W_j,\gamma)|\theta - \theta_0|^a, \quad \mathbb{E}[d(W_i,W_j,\gamma)] < C;
    \end{align*}
    (iv) $\mathbb{E} [|\partial g(w_i,w_j,\hat \gamma, \theta_0)/\partial \theta - \partial g(w_i,w_j, \gamma_0, \theta_0)/\partial \theta|] \to_p 0$; (v) $\mathbb{E}[|\partial g(W_{i},W_{j},\gamma_0,\theta_0)/\partial \theta |] < C$; then $\hat V \to_p V$.
\end{proposition}
In many cases, such as the IOp and AUC examples below, $\alpha_0$ is known up to $\theta_0$ and $\gamma_0$. In such cases, a much simpler asymptotic theory follows from dropping the dependence on $\alpha_0$ of $\psi(W_i,W_j,\gamma_0,\alpha_0,\theta)$ and treating $\psi(W_i,W_j,\gamma_0,\theta)$ as an identifying function that is already orthogonal without the need of extra nuisance parameters.

\begin{corollary}
    \label{corollary_AT_alphaknownuptogamma}
    If $\phi = 0$, then the results of Theorems \ref{Thm_AT_Consist}, \ref{Thm_AT_AN} and Proposition \ref{prop_Sigmahat_consistency} continue to hold after replacing $g$ by $\psi$ throughout and dropping all conditions involving $\phi$, $\alpha$, and $\xi_l$.
\end{corollary}

\begin{remark}[On degeneracy of orthogonal moments]. When
\[
\mathbb{E}[\psi(W_{i},W_{j},\gamma_{0},\alpha_{0},\theta_{0})|W_{i}]=0 \text{ a.s},
\]
then $\Sigma=0$ and $V=0$. In this case, we say that the kernel $\psi$ or the U-statistic is degenerate. In this situation, our results show the estimator converges to $\theta_0$
faster than $n^{-1/2}$. Our asymptotic distribution theory provides distributional results only for
non-degenerate orthogonal moments. We focus on this case because it is the
most common one in the applications we are interested in. A complete asymptotic distributional analysis of the
degenerate case is beyond the scope of this paper and is deferred to future
research.
\end{remark}

\section{Applications}
\label{sec_App}
\subsection{Inequality of Opportunity}
\label{sec_IOp}

The leading measure of IOp is given by the Gini coefficient of income's predictions from a set of circumstances,
which can be expressed in the population as\footnote{The primitive definition of the Gini is twice the area between the Lorenz curve and the 45-degree line. This definition works for both discrete and continuous variables. The U-statistic expression we use in this paper is equivalent to the primitive Lorenz-based expression even with discrete variables. The Online Appendix \ref{app_Othermeasures} extends the analysis to other popular inequality measures, including the Atkinson and Generalized Entropy indices.}
\begin{equation}
\theta_{0}=\frac{\mathbb{E}[|\gamma_{0}(X_{i})-\gamma_{0}(X_{j})|]}%
{\mathbb{E}[\gamma_{0}(X_{i})+\gamma_{0}(X_{j})]}, \label{Gini}%
\end{equation}
where $\gamma_0(x)=\mathbb{E}[Y_{i}|X_{i}=x]$, $x\in\mathcal{X}$, $Y_{i}$ is income, and $X_{i}$ is a vector of circumstances individual $i$ did not choose, such
as parental wealth/income, parental education, sex, color of the skin or
social origin. Rearranging (\ref{Gini}), we obtain an
identifying moment function
\begin{equation}
g(w_{i},w_{j},\gamma,\theta)=(\gamma(x_{i})+\gamma(x_{j}))\theta-|\gamma
(x_{i})-\gamma(x_{j})|. \label{eq_g_IOp}%
\end{equation}
The plug-in sample Gini coefficient $\hat{\theta}^{P}$ solves 
$U_{n}g(\cdot,\hat{\gamma},\hat{\theta}^{P})=0$, i.e.
\[
\hat{\theta}^{P}=\frac{\sum_{i<j}|\hat{\gamma}(X_{i})-\hat{\gamma}(X_{j}%
)|}{\sum_{i<j}\left(  \hat{\gamma}(X_{i})+\hat{\gamma}(X_{j})\right)  },
\]
for a first-step $\hat{\gamma}$ such as the Conditional Inference Forests (CIF),
which is a popular ML method in the IOp literature, see
\cite{brunori2021roots}, \cite{brunori2019inequality} or
\cite{brunori2021evolution}. Monte Carlo simulations and an empirical application show that
$\hat{\theta}^{P}$ is highly biased when ML first-steps, such as the CIF, are used.

A key technical challenge in this example is that the identifying moment involves the non-smooth map $a\mapsto |a|$, whose derivative is discontinuous through the sign function. As a result, standard smooth orthogonality arguments do not apply directly and controlling the nonlinear remainder terms requires additional analysis. Despite this non-smoothness, the orthogonal adjustment admits a simple closed-form representation.

\begin{proposition}
\label{prop_UFSIF_IOp} Under $0<\mathbb{E}[Y]<\infty$, the following function
$\phi$ satisfies (\ref{1}) and (\ref{2}) for the Gini of fitted values, 
\begin{equation}
\phi(w_{i},w_{j},\gamma_{0},\alpha_{0}(\theta),\theta)=\alpha_{01}(\theta)(y_{i}+y_{j}-\gamma
_{0}(x_{i})-\gamma_{0}(x_{j}))+\alpha_{02}(x_{i},x_{j})(y_{i}-y_{j}-\gamma
_{0}(x_{i})+\gamma_{0}(x_{j})), \label{eq_iop_ufsif}%
\end{equation}
where $\alpha_{0}=(\alpha_{01},\alpha_{02})$, $\alpha_{01}(\theta)=\theta$ and $\alpha_{02}(x_{i},x_{j})=-sgn(\gamma_{0}(x_{i})-\gamma_{0}(x_{j}))$.
\end{proposition}
The proof of Proposition \ref{prop_UFSIF_IOp} is given in Online Appendix \ref{app_pfsIOp}. Adding $\phi$ from Proposition \ref{prop_UFSIF_IOp} to $g$ in (\ref{eq_g_IOp}) yields, after cancellation of terms and dropping the dependence of $\psi$ on $\alpha_0$ (see discussion before Corollary \ref{corollary_AT_alphaknownuptogamma})
\[
\psi(W_i,W_j,\gamma_0,\theta)
=
\theta(Y_i+Y_j)
-
sgn(\gamma_0(X_i)-\gamma_0(X_j))(Y_i-Y_j).
\]
The resulting orthogonal moment function admits a particularly simple representation. Let $\hat{\gamma}_{l}(x)$ denote cross-fitted ML predictions of
income given circumstances $x$. Solving the debiased orthogonal sample moment we get the DML IOp U-estimator
\[
\hat{\theta}=\frac{\sum_{l=1}^{L}\sum_{(i,j)\in I_{l}}sgn(\hat{\gamma}_{l}(X_{i})-\hat{\gamma
}_{l}(X_{j}))(Y_{i}-Y_{j})}{\sum_{i<j}(Y_{i}+Y_{j})}.
\]
The debiased estimator $\hat{\theta}$ resembles the Gini coefficient for income,
but weights by the sign of the difference in predictions rather than by the sign of the differences in income. Whenever two individuals have the same fitted values their
difference in incomes cannot be attributed to inequality of opportunity.\footnote{Our estimator can be interpreted as an ML extension of the Lorenz regression approach independently proposed by \cite{heuchenne2022inference}.} The asymptotic variance of the debiased IOp estimator is
\[
V=\frac{\Sigma}{B^2},
\qquad
\Sigma=4\mathbb{V}ar(h(W_i)),
\qquad
B=2\mathbb{E}[Y_i],
\]
where
\[
h(w)
=
\mathbb{E}[\psi(W_i,W_j,\gamma_0,\alpha_0,\theta)\mid W_i=w]
=
\mathbb{E}\!\left[
\theta(y+Y_j)
-
sgn(\gamma_0(x)-\gamma_0(X_j))(y-Y_j)
\right].
\]
A consistent estimator is $\hat V=\hat\Sigma/(4\bar Y^2)$, where
$\bar Y=n^{-1}\sum_{i=1}^n Y_i$ and
\begin{equation}
\hat{\Sigma}
=
\frac{4}{n}
\sum_{i=1}^n
\left(
\frac{1}{n-1}
\sum_{j\neq i}
\hat{\theta}(Y_i+Y_j)
-
sgn(\hat{\gamma}(X_i)-\hat{\gamma}(X_j))(Y_i-Y_j)
\right)^2.
\label{VhatIOp}%
\end{equation}

A practitioner-oriented guide to implementing the proposed estimator, including the construction of the debiased estimator, variance estimation, confidence intervals, and practical recommendations for selecting ML first steps, is provided in Online Appendix \ref{app_practitioner_iop}.
\subsubsection{Asymptotic properties of the debiased IOp}

We give lower-level sufficient conditions for the asymptotic normality of the debiased IOp estimator based on the general assumptions previously introduced. To simplify the notation, define $\Delta_{\gamma}(X_{i},X_{j})=\gamma(X_{i})-\gamma(X_{j})$ and $\Delta_{0}\equiv\Delta_{\gamma_{0}}(X_{i},X_{j})$. 

\begin{assumption}
(i) $||\hat{\gamma}_{l} - \gamma_{0}||_1 = o_p(1)$; (ii) $ P(\Delta_0=0| X_i \neq X_j) = 0$; (iii) $P(0 < |\Delta_0| \leq t ) \leq C t^\beta$   \text{ for } $\beta > 0$, and for all $t>0$ sufficiently small.
\label{Ass_AT_iop}
\end{assumption}
Assumption \ref{Ass_AT_iop} (i) requires $\hat{\gamma}_{l}$ to be $L_1$ consistent. Assumption \ref{Ass_AT_iop} (ii) requires that, given $X_i \neq X_j$, the mass at zero of $|\Delta_0|$ is zero. This holds if $\Delta_0$ has an absolutely continuous distribution or in the finitely discrete case if $\gamma_0$ is injective. Assumption \ref{Ass_AT_iop} (iii) requires that the mass of $|\Delta_0|$ strictly above zero disappears at a given rate. If $\Delta_0$ has a Lebesgue density that is bounded at $0$ then it holds with $\beta = 1$. For $\beta = \infty$, the mass around zero disappears. This is the case when $\Delta_0$ is discrete with finitely many mass points, since there exists $\tilde t>0$ such that
\[
P(0<|\Delta_0|\le t)=0
\]
for all $t\le \tilde t$. Hence, the condition holds for every $\beta>0$, which we summarize by writing $\beta=\infty$. The next result uses (i) and (ii) to show mean-square consistency of the sign difference and (iii) to control the smoothness of the functional $\gamma \rightarrow\mu(\Delta_{\gamma})=\mathbb{E}[(sgn(\Delta_{\gamma}) - sgn(\Delta_0))\Delta_0]$.
\begin{proposition}
\label{prop_sgn_consistency_smoothness}
Under Assumption \ref{Ass_AT_iop} (i) and (ii),
\[
\int\int|sgn(\Delta_{\hat{\gamma}_l}) - sgn(\Delta_0)|^{2} F_{0}(dw_{i})F_{0}(dw_{j}) \to_{p} 0,
\]
and under Assumption \ref{Ass_AT_iop} (iii),
\[
\mu(\Delta_{\gamma})  \leq 
\begin{cases}
    C ||\gamma - \gamma_0||_q^{\frac{q(1+\beta)}{q+ \beta}} &\text{ if } q \in [1,\infty),\\
    C ||\gamma - \gamma_0||_\infty^{1+\beta} &\text{ if } q = \infty.
\end{cases}
\]
\end{proposition}
In $L_2$ norm (i.e., $q=2$), we have that $\mu(\Delta_{\gamma}) \leq C ||\gamma - \gamma_0||^{2 \frac{1+\beta}{2+\beta}}$, so in the continuous case with bounded density at 0 (i.e., $\beta=1$) the smoothness exponent is $4/3$ while in the finitely discrete case ($\beta = \infty$) we achieve a quadratic bound. This proposition is an improvement upon the results in \cite{clemenccon2011minimax}.\footnote{They assume $P(|\gamma_0(X_i) - \gamma_0(x)| \leq t) \leq C t^\beta$ for all $x$. This implies Assumption \ref{Ass_AT_iop}(iii) for $\beta \leq 1$. In contrast, we allow for $\beta > 1$, which is key for the rates in Assumption \ref{ass_rates_iop}. As a result, we recover Massart’s margin condition from classification (\cite{massart2000some}; \cite{audibert2007fast}) in a ranking setting.} The next assumption relates to the trade-off between smoothness properties and the required rates of the ML estimators.
\begin{assumption}
    \label{ass_rates_iop}
    $||\hat{\gamma}_{l} - \gamma_{0}||_q = o_p(n^{-\rho_{q\beta}})$ where $\rho_{q\beta} \geq \frac{1}{2q} \cdot \frac{q+\beta}{1+\beta} \text{ if } q \in [1,\infty)$, and $\rho_{q\beta} \geq \frac{1}{2(1+\beta)} \text{ if } q = \infty$. 
\end{assumption}
There are several ways in which one can achieve the asymptotic normality results, depending on the margin parameter $\beta$ and the rate of the ML estimator under different norms. Table \ref{tab_ML_rates_margin} shows the most common configurations. In the finitely discrete case ($\beta = \infty$) and under the supremum norm, we allow for slow convergence rates, even slower than $n^{-1/4}$. 

\begin{table}[h!]
\centering
\begin{tabular}{@{}ccc@{}}
\toprule
\textbf{\(q\)} & \textbf{\(\beta = 1\)} & \textbf{\(\beta = \infty\)} \\
\midrule
\(2\) & \(\displaystyle 3/8 \) & \(\displaystyle 1/4\) \\
\(\infty\) & \(\displaystyle 1/4 \) & \(\displaystyle  0\) \\
\bottomrule
\end{tabular}
\caption{\footnotesize ML rates required for different margin parameters in $L_2$ and $L_\infty$ norms.}
\label{tab_ML_rates_margin}
\end{table}
\noindent In our empirical application $\beta = \infty$ so the $L_2$ nonparametric rates of $n^{-1/4}$ typically imposed in the DML literature suffice. For consistent estimation of the variance, we assume the following.
\begin{assumption}
\label{ass_Vhatcons_IOp} $\mathbb{E}[(\hat{\gamma}(X_{i}) - \gamma
_{0}(X_{i}))^{2}] = o(1)$.
\end{assumption}
This assumption strengthens Assumption \ref{Ass_AT_iop} (i) to mean-square convergence (unconditionally). It can be discarded at the cost of cross-fitting the variance estimator.
\begin{proposition}
\label{Prop_AT_IOp} Let Assumptions \ref{Ass_AT_iop}, \ref{ass_rates_iop} and
\ref{ass_Vhatcons_IOp} hold. Assuming further that either $\mathbb{E}[Y_i^2 |X_i] < C < \infty$ a.s. or that $\mathbb{E}[|Y_i|^{2+\delta}] < \infty$ for some $\delta > 0$, we have $\sqrt{n}(\hat{\theta}-\theta_{0})\rightarrow_{d}\mathcal{N}(0,V)$ and
$\hat{V}\rightarrow_{p}V$, where $\hat{V}$ is given in (\ref{VhatIOp}).
\end{proposition}

\subsection{Inference on ML Performance through the AUC}
\label{sec_AUC}

The Area Under the ROC Curve (AUC) is one of the most widely used measures
of predictive accuracy for binary classifiers; see \cite{bradley1997use}.
Classical asymptotic theory for the AUC is well understood under fixed or low-dimensional scores; see, e.g., \citet{hanley1982meaning} and \citet{newson2006confidence}. However, modern empirical applications increasingly rely on flexible ML classifiers, for which standard plug-in inference procedures may fail because of regularization and overfitting biases in the estimated score. Let
$Y_i\in\{0,1\}$, $p_0=P(Y_i=1)$, and let $\gamma_0(x)=\mathbb{E}[Y_i|X_i=x]$
be the population ML score. The AUC is the probability that the score ranks a
randomly selected positive observation ($Y_i=1$) above a randomly selected null observation ($Y_i=0$), with ties counted as one half:
\[
\theta_0
=
\mathbb{E}[1(\gamma_0(X_i)>\gamma_0(X_j))|Y_i=1,Y_j=0]
+
\frac12
\mathbb{E}[1(\gamma_0(X_i)=\gamma_0(X_j))|Y_i=1,Y_j=0].
\]
The next result shows that the AUC is a known functional of the same pairwise
sign functional underlying our debiased Gini estimator.

\begin{proposition}
\label{AUCrep}
Let $0<p_0<1$. The AUC can be written as
\[
\theta_0
=
\frac12
+
\frac{
A_0
}
{4p_0(1-p_0)},
\qquad
A_0
=
\mathbb{E}\left[
sgn(\gamma_0(X_i)-\gamma_0(X_j))(Y_i-Y_j)
\right].
\]
Hence, the AUC is locally robust with respect to the ML first-step $\gamma_0$.
\end{proposition}
\begin{remark}
First-step estimation can affect the limiting distribution of the ROC process; see, for example, \citet{hsu2021inference}. Our result targets the scalar AUC directly.
\end{remark}
\noindent Given cross-fitted predictions $\hat\gamma_l$, define
\[
\hat A
=
\binom{n}{2}^{-1}
\sum_{l=1}^{L}
\sum_{(i,j)\in I_l}
sgn(\hat\gamma_l(X_i)-\hat\gamma_l(X_j))(Y_i-Y_j),
\qquad
\hat p=\frac1n\sum_{i=1}^{n}Y_i.
\]
The cross-fitted AUC estimator is
\[
\hat\theta_{AUC}
=
\frac12
+
\frac{\hat A}{4\hat p(1-\hat p)}.
\]

\begin{corollary}[Asymptotic normality of the cross-fitted AUC estimator]
\label{cor_AUC_AN}
Suppose $0<p_0<1$, $Y_i\in\{0,1\}$, and Assumptions
\ref{Ass_AT_iop} and \ref{ass_rates_iop} hold. Then
\[
\sqrt n(\hat\theta_{AUC}-\theta_0)
\rightarrow_d
N(0,V_{AUC}),
\]
where
\[
V_{AUC}
=
\mathbb{V}ar
\left(
\frac{
\mathbb{E}\left[
sgn(\gamma_0(X_i)-\gamma_0(X_j))(Y_i-Y_j)
\mid W_i
\right]
-A_0
}
{2p_0(1-p_0)}
-
\frac{A_0(1-2p_0)}{4p_0^2(1-p_0)^2}
(Y_i-p_0)
\right).
\]
\end{corollary}
The corollary follows from the asymptotic linearity established for the debiased IOp estimator and the delta method applied to $(\hat A,\hat p)$. The key insight is that the AUC inherits local robustness from the debiased pairwise sign functional $A_0$. This yields a feasible and asymptotically valid cross-fitted AUC estimator under ML first-step estimation.
\subsection{Kernel-DML Estimators for Conditional Moment Restrictions}
\label{DBSE}

Many econometric models can be formulated through conditional moment restrictions (CMR):
\begin{equation}
m(Z_i,\gamma_0,\theta)
:=
\mathbb E[\varepsilon(W_i,\gamma_0,\theta)\mid Z_i]
=
0
\ \text{a.s.}
\iff
\theta=\theta_0,
\label{CMR}
\end{equation}
where $\varepsilon:\mathcal W\times\Gamma\times\Theta\to\mathbb R^q$ is known, $\gamma_0$ is an unknown first-step, and $\theta\in\Theta\subseteq\mathbb R^p$ is finite dimensional.\footnote{Henceforth, we focus on $q=1$, since we can apply this case to each component of $\varepsilon$.} Classical GMM estimators based on finitely many unconditional moments may lose identification because they only exploit necessary but not sufficient conditions of (\ref{CMR}); see \citet{dominguez2004consistent}. Integrated Conditional Moment methods avoid this issue by using a continuum of moments; see, among others, \citet{dominguez2004consistent, escanciano2006consistent, shin2008semiparametric, antoine2014conditional, escanciano2018simple,antoine2022partially}. We extend the literature by combining identification-preserving kernel criteria with DML methods. Define
\begin{equation}
Q(\gamma_0,\theta)
=
\mathbb E[
\varepsilon(W_i,\gamma_0,\theta)
\varepsilon(W_j,\gamma_0,\theta)
K(Z_i,Z_j)
],
\label{Q}
\end{equation}
where $K$ is a symmetric kernel. Assume that $K$ is integrally strictly positive definite (ISPD):
\[
\iint K(z_i,z_j)\,d\mu(z_i)\,d\mu(z_j)=0
\quad\Longrightarrow\quad
\mu=0,
\]
for any finite signed measure $\mu\ll P_Z$, where $P_Z$ is the probability measure of $Z_i$. Sufficient conditions for ISPD have been well studied in the literature, e.g., Gaussian and many translation-invariant kernels are known to satisfy it (\citet{sriperumbudur2011universality}). 

\begin{proposition}[No identification loss]
\label{prop:CMR_kernel_id}
Suppose $m(\cdot,\gamma_0,\theta)\in L_2(P_Z)$ for all $\theta\in\Theta$. If (\ref{CMR}) holds and $K$ is ISPD, then
\[
Q(\gamma_0,\theta)\ge0,
\qquad
Q(\gamma_0,\theta)=0
\iff
\theta=\theta_0.
\]
\end{proposition}
Proposition \ref{prop:CMR_kernel_id} establishes that the kernel criterion
preserves identification under ISPD kernels. The estimator studied below is
based on the first-order conditions of this criterion. Sufficient conditions linking the identification of the criterion
to the identification of the corresponding estimating equations are discussed in
Online Appendix \ref{app_ident_prod_fcns}.

Let
\[
\varepsilon_\theta(W_i,\gamma,\theta)
=
\frac{\partial}{\partial\theta}
\varepsilon(W_i,\gamma,\theta),
\]
and
\[
\varepsilon_\gamma(W_i,\gamma_0,\theta)
=
\frac{\partial}{\partial t}\varepsilon(W_i,t,\theta)\big|_{t=\gamma_0(X_i)},
\qquad
\varepsilon_{\theta\gamma}(W_i,\gamma_0,\theta)
=
\frac{\partial^2}{\partial t\,\partial\theta'}
\varepsilon(W_i,t,\theta)\big|_{t=\gamma_0(X_i)}.
\]
We apply our methodology to the first-order conditions of $Q(\gamma_0,\theta)$:
\begin{equation}
g(W_i,W_j,\gamma,\theta)
=
K(Z_i,Z_j)
\Big(
\varepsilon(W_i,\gamma,\theta)\varepsilon_\theta(W_j,\gamma,\theta)
+
\varepsilon_\theta(W_i,\gamma,\theta)\varepsilon(W_j,\gamma,\theta)
\Big).
\label{FOC_CMR}
\end{equation}

\begin{assumption}
\label{ass_KCMR}
(i) $\varepsilon(\cdot,\theta)$ and $\varepsilon_\gamma(\cdot,\theta)$ are differentiable in $\theta$ a.s.; 
(ii) $K\in L_2$; 
(iii) $\varepsilon$, $\varepsilon_\theta$, $\varepsilon_\gamma$, and $\varepsilon_{\theta\gamma}$ admit square-integrable envelopes uniformly in $\theta$.
\end{assumption}

\begin{proposition}
\label{prop_phi_KCMR}
Suppose Assumption \ref{ass_KCMR} holds and $\gamma_0$ satisfies (\ref{orth1}) with $\Gamma=L_2(P_X)$. Then an orthogonal correction term for (\ref{FOC_CMR}) is
\begin{align*}
&\phi(W_i,W_j,\gamma_0,\alpha_0(\theta),\theta)
=
\alpha_0(X_i,\theta)(Y_i-\gamma_0(X_i))
+
\alpha_0(X_j, \theta)(Y_j-\gamma_0(X_j)), \\
&\alpha_0(x,\theta)
=
\mathbb E\!\left[
\Big(
\varepsilon_\gamma(W_i,\gamma_0,\theta)\varepsilon_\theta(W_j,\gamma_0,\theta)
+
\varepsilon_{\theta\gamma}(W_i,\gamma_0,\theta)\varepsilon(W_j,\gamma_0,\theta)
\Big)
K(Z_i,Z_j)
\mid X_i=x
\right].
\end{align*}
\end{proposition}
\subsubsection{Semiparametric Production Functions with a Proxy Variable}

Consider the two-period production model
\[
Y_{it}
=
F(K_{it},L_{it},\beta_0)+\omega_{it}+\varepsilon_{it},
\qquad t=1,2,
\]
where $Y_{it}$ is log output for firm $i$ at time $t$, $(K_{it},L_{it})$ denote capital and labor inputs, $\beta_0\in\mathcal{B}\subset\mathbb{R}^{d_\beta}$ is a finite-dimensional parameter, $\omega_{it}$ is firm productivity, and $\varepsilon_{it}$ is a mean-zero innovation. Productivity is known to the firms but not to the econometrician and correlates with observable variables. The proxy-variable approach addresses this endogeneity by using intermediate inputs or investment decisions to control for unobserved productivity shocks; see \citet{olley1996dynamics}, \citet{levinsohn2003estimating}, and \citet{ackerberg2015identification}.

\begin{assumption}[Proxy-variable structure]
\label{ass_proxy_model}
\leavevmode
\vspace{-0.4\baselineskip}
\begin{enumerate}[i)]
\item \textbf{Proxy variable:} There exists an observable variable $I_{i1}=\varphi_0(\omega_{i1},K_{i1},L_{i1})$.

\item \textbf{Monotonicity:} For fixed $(K_{i1},L_{i1})$, $\varphi_0(\cdot,K_{i1},L_{i1})$ is strictly monotone, so that $\omega_{i1}=\omega_0(X_i)$ with $X_i=(I_{i1},K_{i1},L_{i1})$.

\item \textbf{Productivity law of motion:} $\mathbb E[\omega_{i2}\mid X_i]=\rho_0\omega_{i1}$, for some $\rho_0\in\mathbb{R}$.

\item \textbf{Exogeneity:} $\mathbb E[\varepsilon_{it}\mid X_i]=0$ for $t=1,2$.
\end{enumerate}
\end{assumption}
Then
\[
\gamma_0(X_i)
=
\mathbb E[Y_{i1}\mid X_i]
=
F(K_{i1},L_{i1},\beta_0)+\omega_0(X_i),
\]
and the model implies $m(X_i,\gamma_0,\theta)
=
\mathbb E[\varepsilon(W_i,\gamma_0,\theta)\mid X_i]=0$ at $\theta=\theta_0$, where
\[
\varepsilon(W_i,\gamma_0,\theta)
=
Y_{i2}
-
F(K_{i2},L_{i2},\beta)
-
\rho\big(
\gamma_0(X_i)-F(K_{i1},L_{i1},\beta)
\big), \quad \theta=(\beta',\rho)'.
\]
To see how identification might fail with finitely many instruments, suppose $\beta_0$ is known and consider identification of $\rho_0$. We have that $m(X_i,\gamma_0,\beta_0,\rho) = (\rho_0-\rho)\omega_0(X_i)$. Hence, identification only requires that $\omega_0(X_i)$ is not equal to zero almost surely. However, finite-dimensional instruments $h(X_i)$ identify $\rho_0$ only through the finite moments
\[
\mathbb E[m(X_i,\rho)h(X_i)]
=
(\rho_0-\rho)\mathbb E[\omega_0(X_i)h(X_i)].
\]
Hence, identification may become weak if $\omega_0(X_i)$ is weakly correlated with the chosen instruments.\footnote{Adjustment costs in investment or intermediate materials lead to non-linearities in $\omega_0$ (see \cite{cooper2006nature}, \cite{levinsohn2003estimating} or \cite{peter2023aggregate}.} For a numeric example where using a finite set of polynomial as instruments fails, see the example in Online Appendix \ref{app_ident_prod_fcns}. The following Proposition provides the orthogonal moment function for the proxy-variable approach to production function estimation.

\begin{proposition}
\label{prop:prod_phi}
Suppose $\Theta$ is compact, $K\in L_2$, $Y_{it}$ and $\gamma_0(X_i)$ are square integrable, and
\[
F(K_{it},L_{it},\beta),
\qquad
\frac{\partial}{\partial\beta}F(K_{it},L_{it},\beta)
\]
admit square-integrable envelopes uniformly in $\beta$. Then the orthogonal score is $\psi=g+\phi$, where $g$ and $\phi$ are given by (\ref{FOC_CMR}) and Proposition \ref{prop_phi_KCMR}, respectively, with
\[
\alpha_0(X_i)
=
\left(
(\rho_0\eta_{02}(X_i)-\rho_0^2\eta_{01}(X_i))',
\,
\rho_0\eta_{03}(X_i)
\right)',
\]
where
\begin{align*}
\eta_{01}(x)
&=
\mathbb E\!\left[
\frac{\partial}{\partial\beta}F(K_{j1},L_{j1},\beta_0)K(x,X_j)
\right], \quad \eta_{02}(x)
&=
\mathbb E\!\left[
\frac{\partial}{\partial\beta}F(K_{j2},L_{j2},\beta_0)K(x,X_j)
\right],
\\
\eta_{03}(x)
&=
\mathbb E\!\left[
(Y_{j1}-F(K_{j1},L_{j1},\beta_0))K(x,X_j)
\right].
\end{align*}
\end{proposition}

\paragraph{Flexible linear specification} Let $F(Q_{it},\beta)=\beta'Q_{it}$, where $Q_{it}$ contains flexible transformations of inputs. Then
\[
\varepsilon(W_i,\gamma,\theta)
=
Y_{i2}-\beta'Q_{i2}-\rho(\gamma(X_i)-\beta'Q_{i1}),
\]
and
\[
\alpha_0(x,\theta)
=
\rho
\begin{pmatrix}
\eta_{02}(x)-\rho\eta_{01}(x)
\\
\eta_{03}(x) - \beta'\eta_{01}(x)
\end{pmatrix},
\]
with $\eta_{01}(x)=\mathbb E[Q_{j1}K(x,X_j)]$, $\eta_{02}(x)=\mathbb E[Q_{j2}K(x,X_j)]$ and $\eta_{03}(x) = \mathbb E[Y_{j1}K(x,X_j)]$. For each block $I_l$,
\begin{align*}
\hat{\eta}_{1l}(x)
&=
\frac{1}{|S_l^c|}
\sum_{j\in S_l^c}Q_{j1}K(x,X_j),
\quad
\hat{\eta}_{2l}(x)
=
\frac{1}{|S_l^c|}
\sum_{j\in S_l^c}Q_{j2}K(x,X_j), \quad \hat{\eta}_{3l}(x)
&=
\frac{1}{|S_l^c|}
\sum_{j\in S_l^c}Y_{j1}K(x,X_j),
\end{align*}
are just sample averages. The induced estimator of $\alpha_0(\theta)$ is
\[
\hat{\alpha}_l(x,\theta)
=
\rho
\begin{pmatrix}
\hat{\eta}_{02l}(x)-\rho\hat{\eta}_{01l}(x)
\\
\hat{\eta}_{03l}(x) - \beta'\hat{\eta}_{01l}(x)
\end{pmatrix},
\]
and the DML U-estimator solves
\begin{equation}
\binom{n}{2}^{-1}
\sum_{l=1}^{L}
\sum_{(i,j)\in I_l}
\psi(W_i,W_j,\hat\gamma_l,\hat\alpha_l(\theta),\theta)
=
0.
\label{flexLin-DMLU}
\end{equation}
For the asymptotic theory of the flexible linear specification, we assume the following
\begin{assumption}
\label{ass_asympt_prod_fcns}
(i) $\theta_0$ is an interior point of a compact set $\Theta$ and
\[
\mathbb E[g(W_i,W_j,\gamma_0,\theta)]=0
\iff
\theta=\theta_0;
\]
(ii) $K(X_i,X_j)$ is bounded a.s.; 
(iii) $Y_{it}$, $Q_{it}$, and $\gamma_0(X_i)$ have finite fourth moments; 
(iv) $\|\hat\gamma_l-\gamma_0\|=o_p(n^{-1/4})$ for all $l$; (v) $B
=
\frac{\partial}{\partial\theta'}
\mathbb E[g(W_i,W_j,\gamma_0,\theta_0)]$
exists and is nonsingular.
\end{assumption}
Identification in semiparametric production functions and sufficient
conditions for Assumption \ref{ass_asympt_prod_fcns}(i) are discussed in
Online Appendix \ref{app_ident_prod_fcns}.
\begin{proposition}
\label{prop_AN_prod_fcns}
Let $\hat\theta$ solve (\ref{flexLin-DMLU}). Under Assumption \ref{ass_asympt_prod_fcns},
\[
\sqrt n(\hat\theta-\theta_0)
\rightarrow_d
N(0,V),
\]
where $V = B^{-1}\Sigma B^{\prime-1}$ and $\Sigma = 4\mathbb V\!ar\left(\mathbb E[
\psi(W_i,W_j,\gamma_0,\alpha_0,\theta_0) \mid W_i]\right)$.
\end{proposition}
\section{Monte Carlo Simulations}

\label{sec_simulations}

This section compares the finite-sample performance of the debiased IOp and plug-in estimators under two DGPs: an artificial Gaussian first-step error model that isolates first-step bias, and an empirically motivated design calibrated to the Spanish data used below.

\subsection{A Gaussian First-Step Error Model}
In this DGP, $X_i \sim \mathcal{N}(0,\sigma_X^2)$ and
\[
Y_i = \gamma_0(X_i) + \varepsilon_i, \quad \gamma_0(x) = 5 + x^2,
\]
where $\varepsilon_i \sim \mathcal{N}(0, \sigma_\varepsilon^2)$, $\sigma_\varepsilon^2 = \mathbb{V}ar(\gamma_0(X_i))/StN$, and $StN$ is the Signal to Noise Ratio which is fixed by us. This DGP is called a Gaussian First-Step Error Model because we generate
\[
\hat{\gamma}(X_i)
=
\gamma_0(X_i)+b_n(X_i)+\sigma_nU_i,
\qquad
U_i\sim N(0,1),
\]
where
\[
b_n(x)
=
c_1\bigl(\mathbb E[\gamma_0(X_i)]-\gamma_0(x)\bigr)n^{-\rho_1},
\qquad
\sigma_n^2
=
\frac{c_2}{STN}n^{-\rho_2},
\]
and $c_1,c_2,\rho_1,\rho_2$ are constants. The variables $(X_i,\varepsilon_i,U_i)$ are independent. This exercise is artificial in the sense that $\hat \gamma$ is not obtained from the data. However, the upside of this experiment is that we can check asymptotic theory with full control over the rates of $\hat{\gamma}$. Conveniently, the derivative (\ref{main}) is normally distributed

\begin{equation*}
\sqrt{n}\frac{\partial\theta(\gamma_{0})}{\partial\gamma}\left[  \hat{\gamma}%
-\gamma_{0}\right]\sim \mathcal{N}\left(C_{1}n^{0.5-\rho_1}, \sigma^2_{dn}\right),
\end{equation*}
where $C_{1}=c_{1}\theta_0$, $\sigma^2_{dn}=C_{2}n^{1-\rho_2}$ and $C_{2}=c_{2}(\mathbb{E}[2Y_i])^{-2}\theta^2_{0}/StN$. Computations for the Gaussian First-Step Error Model can be found in Online Appendix \ref{app_computations_GFSM}. Therefore, the first-step contribution to the root-$n$ distribution diverges when $\rho_1 < 1/2$ or $\rho_2 < 1$, preventing valid root-$n$ inference for the plug-in estimator. This illustrates the impact of first-step estimation error and the potential failure of standard plug-in inference.

Assumption \ref{Ass_AT_iop} (iii) holds with any $\beta < 1$ (see Online Appendix \ref{app_computations_GFSM}). Hence, we require rates faster than $3/8$ in $L_2$ norm for valid inference of the debiased estimator. This implies $\rho_1 > 3/8$ and $\rho_2 > 3/4$. In the simulations, we present slower and faster rates than required to evaluate the robustness of the finite sample performance to our sufficient conditions. The lower the $\beta$, the less smooth the model and the harder it is to control the nonlinear bias terms. Hence, with this DGP, we are dealing with a difficult case ($\beta<1$).  

Table \ref{tab:tab_GEM_sims} shows bias and 95\% coverage across choices of $(c_1,c_2,\rho_1,\rho_2,\sigma_X^2,\mathrm{StN})$. For the plug-in estimator, we report naive coverage ignoring first-step estimation and corrected-standard-error coverage, using the same formula as for the debiased estimator. For the debiased estimator, we report coverage and average interval length.

The results match the theory. Slow convergence of $\hat\gamma$ produces sizeable plug-in bias and severe coverage distortions, especially with naive standard errors. The debiased estimator has small bias and close-to-nominal coverage across sample sizes and parameter configurations, including cases outside our sufficient conditions.

\begin{table}[!h]
\centering
\caption{\label{tab:tab_GEM_sims}Gaussian First-step error model simulations}
\centering
\fontsize{8}{10}\selectfont
\begin{threeparttable}
\begin{tabular}[t]{cccccccccc}
\toprule
\textbf{StN} & \textbf{$\rho_1$} & \textbf{$\rho_2$} & \textbf{$n$} & \textbf{BP} & \textbf{CP} & \textbf{CPN} & \textbf{BD} & \textbf{CD} & \textbf{CI length}\\
\midrule
0.1 & 0.26 & 0.6 & 1000 & -0.011 & 1.000 & 0.165 & -0.005 & 0.935 & 0.057\\
 &  &  & 3000 & -0.009 & 0.999 & 0.020 & -0.003 & 0.934 & 0.033\\
 &  &  & 6000 & -0.008 & 0.974 & 0.001 & -0.002 & 0.935 & 0.023\\
\midrule
 &  & 0.8 & 1000 & -0.016 & 1.000 & 0.017 & -0.002 & 0.945 & 0.057\\
 &  &  & 3000 & -0.012 & 0.971 & 0.001 & -0.001 & 0.936 & 0.033\\
 &  &  & 6000 & -0.010 & 0.751 & 0.000 & -0.001 & 0.945 & 0.023\\
\midrule
 & 0.5 & 0.6 & 1000 & 0.003 & 1.000 & 0.920 & -0.004 & 0.945 & 0.057\\
 &  &  & 3000 & 0.002 & 1.000 & 0.893 & -0.002 & 0.934 & 0.033\\
 &  &  & 6000 & 0.001 & 1.000 & 0.910 & -0.002 & 0.942 & 0.023\\
\midrule
 &  & 0.8 & 1000 & -0.002 & 1.000 & 0.923 & -0.002 & 0.948 & 0.057\\
 &  &  & 3000 & -0.001 & 1.000 & 0.915 & -0.001 & 0.936 & 0.033\\
 &  &  & 6000 & -0.001 & 1.000 & 0.911 & -0.001 & 0.947 & 0.023\\
\midrule
1 & 0.26 & 0.6 & 1000 & -0.017 & 0.115 & 0.010 & -0.001 & 0.947 & 0.024\\
 &  &  & 3000 & -0.013 & 0.009 & 0.001 & 0.000 & 0.935 & 0.014\\
 &  &  & 6000 & -0.011 & 0.001 & 0.000 & 0.000 & 0.936 & 0.010\\
\midrule
 &  & 0.8 & 1000 & -0.017 & 0.085 & 0.006 & 0.000 & 0.949 & 0.024\\
 &  &  & 3000 & -0.013 & 0.005 & 0.000 & 0.000 & 0.940 & 0.014\\
 &  &  & 6000 & -0.011 & 0.000 & 0.000 & 0.000 & 0.934 & 0.010\\
\midrule
 & 0.5 & 0.6 & 1000 & -0.003 & 0.983 & 0.889 & -0.001 & 0.947 & 0.024\\
 &  &  & 3000 & -0.001 & 0.982 & 0.893 & 0.000 & 0.938 & 0.014\\
 &  &  & 6000 & -0.001 & 0.982 & 0.893 & 0.000 & 0.935 & 0.010\\
\midrule
 &  & 0.8 & 1000 & -0.003 & 0.979 & 0.861 & 0.000 & 0.950 & 0.024\\
 &  &  & 3000 & -0.002 & 0.980 & 0.873 & 0.000 & 0.942 & 0.014\\
 &  &  & 6000 & -0.001 & 0.973 & 0.865 & 0.000 & 0.933 & 0.010\\
\bottomrule
\end{tabular}
\begin{tablenotes}[para]
\item Simulation Results 1000 iterations, 95\% confidence intervals. sd(X) = 1, $c_1 = 1$, $c_2 = 0.5$. BP: Bias Plug-in, CP: Coverage Plug-in, CPN: Coverage Plug-in Naive, BD: Bias Debiased, CD: Coverage Debiased.
\end{tablenotes}
\end{threeparttable}
\end{table}

\subsection{Empirically-Based DGP}

Next, we calibrate a DGP to the Spanish data used in the empirical application. We regress log income on mother's education, father's education, and father's occupation, and generate
\[
Y_i=\exp(\gamma_0(X_i)+\varepsilon_i),\qquad \varepsilon_i\sim N(0,\sigma_\varepsilon^2),
\]
where $\gamma_0(X_i)$ is sampled with replacement from the OLS fitted values and $\sigma_\varepsilon^2$ is the residual variance. The econometrician sees all circumstances in the data, but not the subset entering the DGP, and uses grouped Lasso for variable selection (see \cite{breheny2015group}).

The DGP has 52 groups with distinct conditional means and true IOp equal to 0.085. The econometrician does not know which groups enter the DGP and might entertain a much larger number of groups. Since the ML methods are trained on $Y$ rather than $\log Y$, detecting the signal is difficult; the resulting signal-to-noise ratio is a realistic 3\%.

We use random forests (RF) and conditional inference forests (CIF) implemented with the \texttt{ranger} and \texttt{party} packages.\footnote{Hyperparameters are selected by cross-validation. See \citet{bach2024hyperparameter} for a recent study of hyperparameter tuning in DML.} CIF is popular in the IOp literature (e.g. \cite{brunori2021roots}). As a benchmark, we also report a correctly specified parametric estimator obtained by regressing $\log Y$ on the true circumstances and using $\exp(\hat\gamma(X))$ as predictions. Since the Gini is scale-invariant, this estimator is correctly specified up to scale. As before, plug-in coverage is reported both with and without correcting for first-step estimation.

Table \ref{tab_selection_sims} shows that RF and CIF plug-in estimators remain biased and deliver poor coverage, especially with naive standard errors. Debiasing sharply reduces bias and restores close-to-nominal coverage. The correctly specified parametric estimator performs well, but only after accounting for first-step estimation in the standard errors. The debiased ML intervals shrink with sample size and are comparable in length to the parametric benchmark, illustrating the adaptive and efficient features of the proposed debiased inferences.

\begin{table}[H]
  \centering
  \footnotesize
\caption{Simulation results with 1000 iterations, 95\% condidence intervals.}

\begin{tabular}[t]{llcccccc}
\toprule
\multicolumn{2}{c}{ } & \multicolumn{3}{c}{Plug-in} & \multicolumn{3}{c}{Debiased} \\
\cmidrule(l{3pt}r{3pt}){3-5} \cmidrule(l{3pt}r{3pt}){6-8}
ML & n & Bias & Coverage & Coverage naive & Bias & Coverage & Avg. CI length\\
\midrule
\addlinespace[0.3em]
\multicolumn{8}{l}{\textbf{}}\\
\hspace{1em}\textbf{RF} & 1000 & -0.012 & 0.916 & 0.186 & 0.002 & 0.953 & 0.081\\
\hspace{1em} & 3000 & -0.012 & 0.806 & 0.140 & -0.001 & 0.937 & 0.047\\
\hspace{1em} & 6000 & -0.011 & 0.741 & 0.097 & -0.001 & 0.923 & 0.033\\
\addlinespace[0.3em]
\multicolumn{8}{l}{\textbf{}}\\
\hspace{1em}\textbf{CIF} & 1000 & -0.033 & 0.621 & 0.081 & -0.008 & 0.901 & 0.077\\
\hspace{1em} & 3000 & -0.022 & 0.550 & 0.034 & -0.003 & 0.922 & 0.046\\
\hspace{1em} & 6000 & -0.016 & 0.495 & 0.037 & -0.002 & 0.934 & 0.032\\
\addlinespace[0.3em]
\multicolumn{8}{l}{\textbf{}}\\
\hspace{1em}\textbf{Parametric} & 1000 & 0.022 & 0.737 & 0.076 & -0.006 & 0.918 & 0.077\\
\hspace{1em} & 3000 & 0.012 & 0.868 & 0.106 & -0.003 & 0.926 & 0.046\\
\hspace{1em} & 6000 & 0.008 & 0.887 & 0.103 & -0.001 & 0.956 & 0.032\\
\bottomrule
\end{tabular}

  \label{tab_selection_sims}
\end{table}

\section{Inequality of Opportunity in Europe}

\label{sec_empapp} 
We measure IOp in 29 European countries using the 2019 wave of the EU-SILC survey. Income is equivalized household income, the unit is the individual, and we restrict the sample to ages 25--59. Circumstances are measured retrospectively around age 14 and include sex, country of birth, household composition, municipality size, housing tenure, parental country of birth, parental education and occupation, father's managerial position, father's occupation, household financial situation, and access to school materials. All circumstances are discrete, and we use grouped Lasso to select among them before estimating the first step.

Figure \ref{fig:iop} reports relative IOp estimates using CIF, with 95\% confidence intervals adjusted for multiple testing (\cite{vsidak1967rectangular}) for the debiased estimator. We also report the parametric (\cite{ferreira2011measurement}) benchmark with bootstrap standard errors. For CIF, the plug-in estimator systematically underestimates IOp and often lies outside the debiased confidence intervals. Results for other ML methods are reported in Online Appendix \ref{app_emp}.

\begin{figure}[H]
    \centering
    \includegraphics[width=1\textwidth]{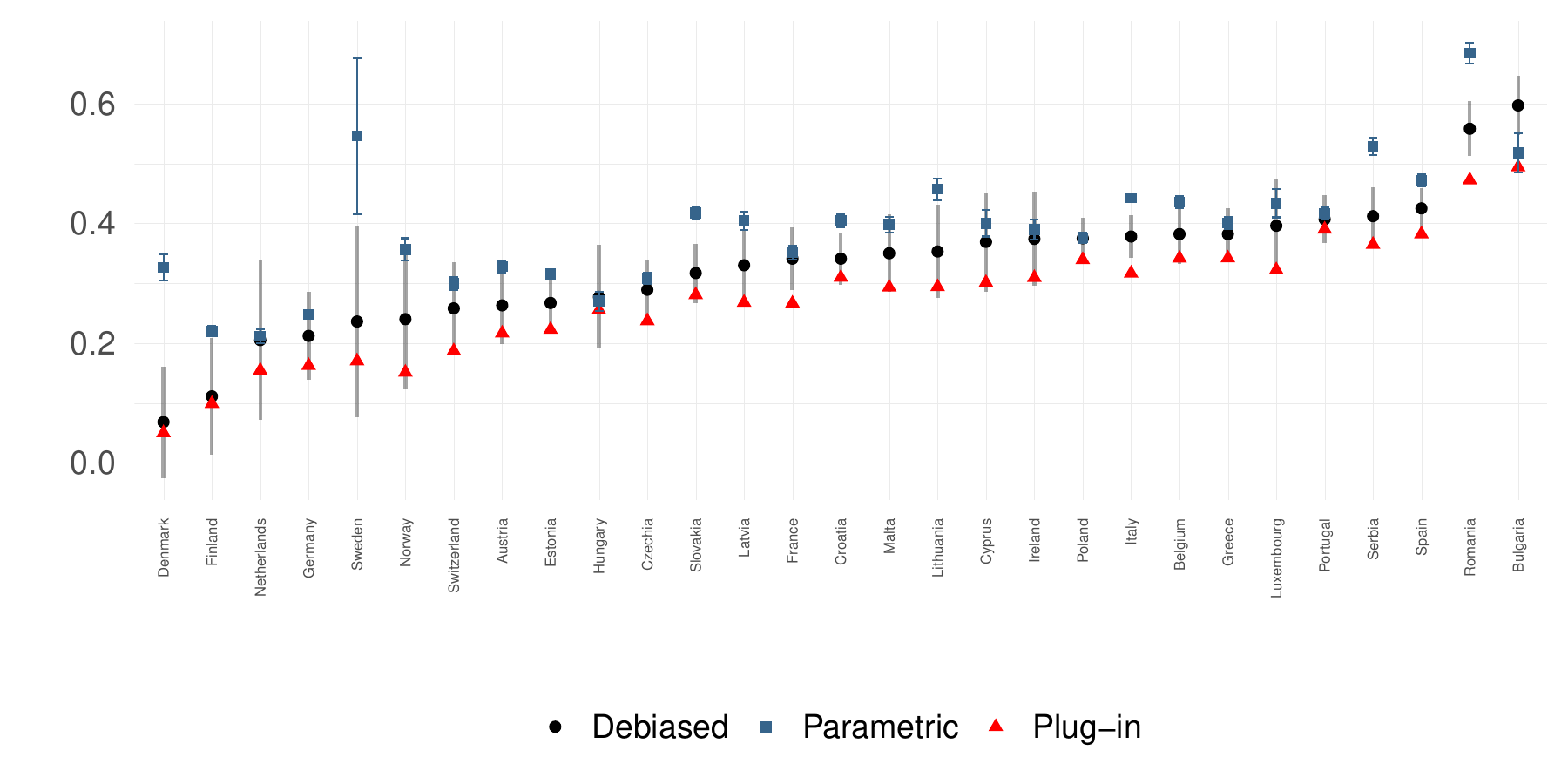}
    \caption{IOp with CIF}
    \label{fig:iop}
  \end{figure}

Relative IOp varies substantially across countries. Nordic countries, the Netherlands, and Germany are at the lower end, while Southern countries tend to display higher IOp. Eastern Europe is heterogeneous: Czechia and Hungary are relatively low, whereas Romania and Bulgaria are very high. Detailed estimates are reported in Table 4.

\begin{table}[h]
\centering
\caption{\label{tab:emp_cif}\footnotesize Results for CIF. Hyperparameter tuning grids: Number of Trees: \{300, 500, 800\}, CIF Depth: \{1, 3, 6\}, Number of variables to split (mtry): \{5, 6, 7, 8, 9\}}
\centering
\fontsize{10}{12}\selectfont
\begin{tabular}[t]{cccccc}
\toprule
Country & Mean & Gini & Plug in/Gini & Debiased/Gini & n\\
\midrule
Austria & 29783 & 0.268 & 22 \% & 26 \%  (22\%,30\%) & 5150\\
Belgium & 28423 & 0.237 & 34 \% & 38 \%  (35\%,41\%) & 6368\\
Bulgaria & 6603 & 0.415 & 49 \% & 60 \%  (57\%,63\%) & 5906\\
Switzerland & 52732 & 0.273 & 19 \% & 26 \%  (21\%,31\%) & 4715\\
Cyprus & 20281 & 0.299 & 30 \% & 37 \%  (32\%,42\%) & 4647\\
Czechia & 12200 & 0.232 & 24 \% & 29 \%  (26\%,32\%) & 6881\\
Germany & 28799 & 0.271 & 16 \% & 21 \%  (17\%,26\%) & 5981\\
Denmark & 36736 & 0.260 & 5 \% & 7 \%  (1\%,13\%) & 2076\\
Estonia & 14307 & 0.283 & 22 \% & 27 \%  (24\%,30\%) & 5669\\
Greece & 9915 & 0.310 & 34 \% & 38 \%  (35\%,41\%) & 14144\\
Spain & 17707 & 0.327 & 38 \% & 42 \%  (40\%,45\%) & 16975\\
Finland & 29297 & 0.263 & 10 \% & 11 \%  (5\%,17\%) & 4403\\
France & 27001 & 0.274 & 27 \% & 34 \%  (31\%,37\%) & 7924\\
Croatia & 8916 & 0.281 & 31 \% & 34 \%  (31\%,37\%) & 7120\\
Hungary & 6925 & 0.282 & 26 \% & 28 \%  (22\%,33\%) & 4568\\
Ireland & 32224 & 0.278 & 31 \% & 37 \%  (32\%,42\%) & 3660\\
Italy & 20006 & 0.319 & 32 \% & 38 \%  (36\%,40\%) & 16360\\
Lithuania & 10383 & 0.347 & 29 \% & 35 \%  (30\%,40\%) & 3643\\
Luxembourg & 45435 & 0.329 & 32 \% & 40 \%  (35\%,44\%) & 3533\\
Latvia & 10787 & 0.336 & 27 \% & 33 \%  (29\%,37\%) & 3344\\
Malta & 19263 & 0.266 & 29 \% & 35 \%  (31\%,39\%) & 3632\\
Netherlands & 30629 & 0.259 & 15 \% & 20 \%  (12\%,29\%) & 4446\\
Norway & 41869 & 0.245 & 15 \% & 24 \%  (17\%,31\%) & 2366\\
Poland & 8621 & 0.292 & 34 \% & 38 \%  (35\%,40\%) & 14293\\
Portugal & 12030 & 0.305 & 39 \% & 41 \%  (38\%,43\%) & 13783\\
Romania & 4826 & 0.343 & 47 \% & 56 \%  (53\%,59\%) & 5932\\
Serbia & 3984 & 0.329 & 36 \% & 41 \%  (38\%,44\%) & 5648\\
Sweden & 29282 & 0.300 & 17 \% & 24 \%  (14\%,34\%) & 2027\\
Slovakia & 9197 & 0.221 & 28 \% & 32 \%  (28\%,35\%) & 5727\\
\bottomrule
\end{tabular}
\end{table}

Interpreting the direction and size of the debiasing correction is not straightforward. Although the plug-in estimator systematically underestimates IOp with CIF, Online Appendix \ref{app_emp} shows that this depends on the ML method. A useful heuristic is to decompose the difference between the debiased and plug-in estimators into three terms.

\begin{align*}
    \hat{\theta} - \hat{\theta}^P &= \frac{\sum_{l=1}^L \sum_{(i,j) \in I_l} (sgn(\hat{\gamma}_l(X_i) - \hat{\gamma}_l(X_j)) - sgn(\hat{\gamma}(X_i) - \hat{\gamma}(X_j)))(Y_i - Y_j)}{\sum_{i<j} Y_i + Y_j} \\
    &+ \frac{\sum_{i<j} sgn(\hat{\gamma}(X_i) - \hat{\gamma}(X_j))(Y_i - \hat{\gamma}(X_i) - Y_j + \hat{\gamma}(X_j))}{\sum_{i<j} Y_i + Y_j} \\
    &+ \sum_{i<j} |\hat{\gamma}(X_i) - \hat{\gamma}(X_j)| \left(\frac{1}{\sum_{i<j} Y_i + Y_j} -  \frac{1}{\sum_{i<j} \hat{\gamma}(X_i) + \hat{\gamma}(X_j)}\right).
\end{align*}

The first term captures the overfitting removed by cross-fitting. The second corrects the numerator holding the denominator fixed: it can be read as the covariance between prediction orderings and prediction errors---positive when high/low predictions are systematically under/overestimated, as in income data---or as a regularization-bias correction, since regularization compresses prediction differences relative to income differences. Thus, it is expected to be positive, though it may turn negative if prediction orderings are systematically wrong. The third term replaces predicted incomes with actual incomes in the denominator. Figure \ref{fig:bias_decomp} reports these terms divided by the Gini of income; white circles show their sum.

\begin{figure}[H]
    \centering
    \includegraphics[width=1\textwidth]{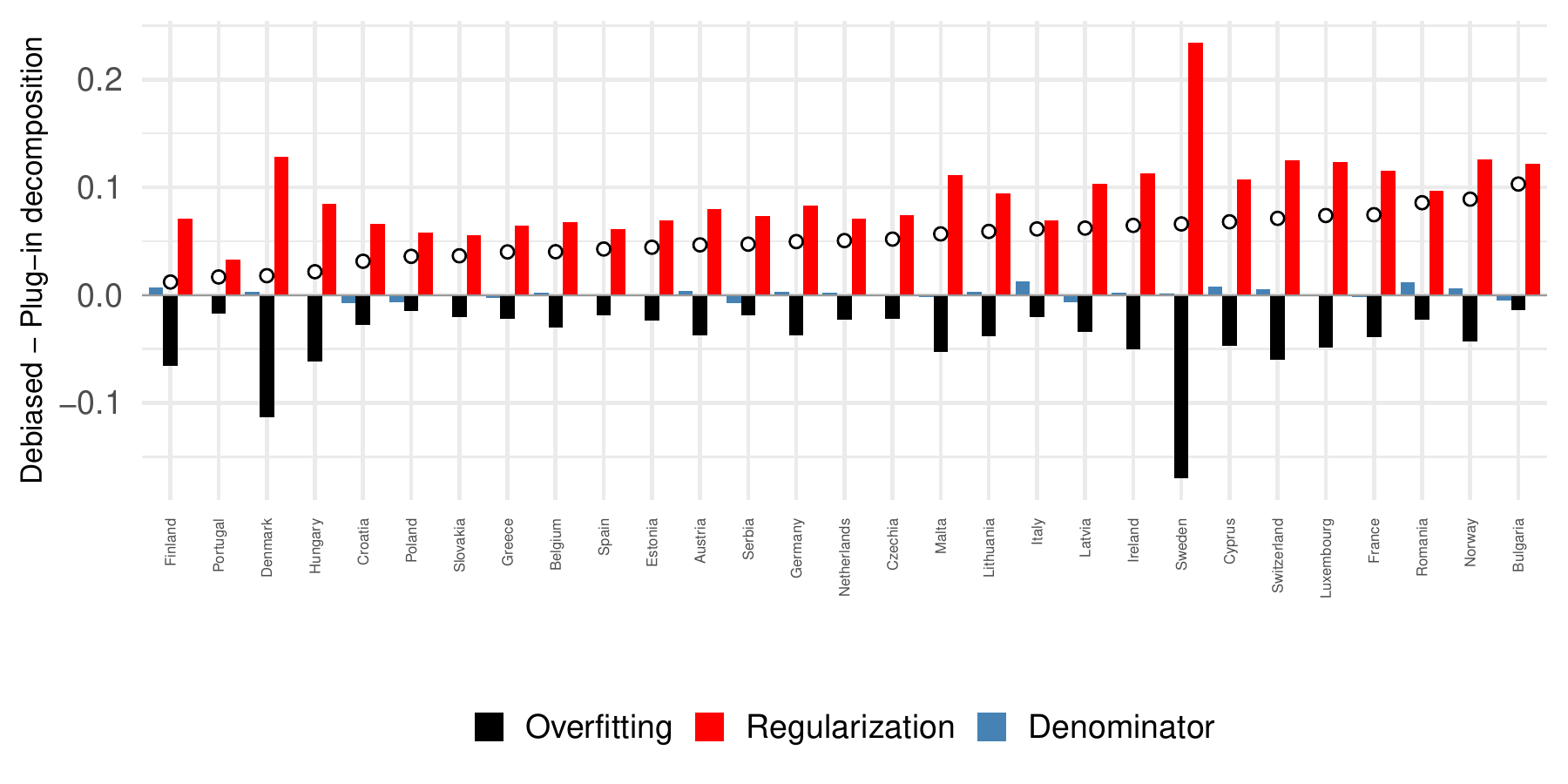}
    \caption{Debiased minus plug-in decomposition for CIF}
    \label{fig:bias_decomp}
  \end{figure}

The denominator component is generally negligible, the regularization component is positive in all countries, and the overfitting component is negative. Since the regularization component dominates, the plug-in estimator underestimates IOp throughout.

Online Appendix M shows that the regularization and overfitting components have similar signs across ML methods, although the total correction may change sign depending on their relative magnitude. It also shows that overfitting is stronger in smaller samples, which helps explain the large overfitting corrections in Scandinavian countries. Hence, the correction reflects regularization, overfitting, ML choice, and sampling noise. 

Figure \ref{fig:iop_spread} shows that debiased estimates are much less sensitive to the first-step ML choice. Across Lasso, Ridge, RF, CIF, XGBoost, and CatBoost, debiased estimates are tightly clustered, whereas plug-in estimates are substantially more dispersed, with cross-ML differences sometimes close to 40 percentage points. This sensitivity highlights the practical importance of orthogonalization for ML-based IOp estimation.

\begin{figure}[h!]
    \centering
    \includegraphics[width=1\textwidth]{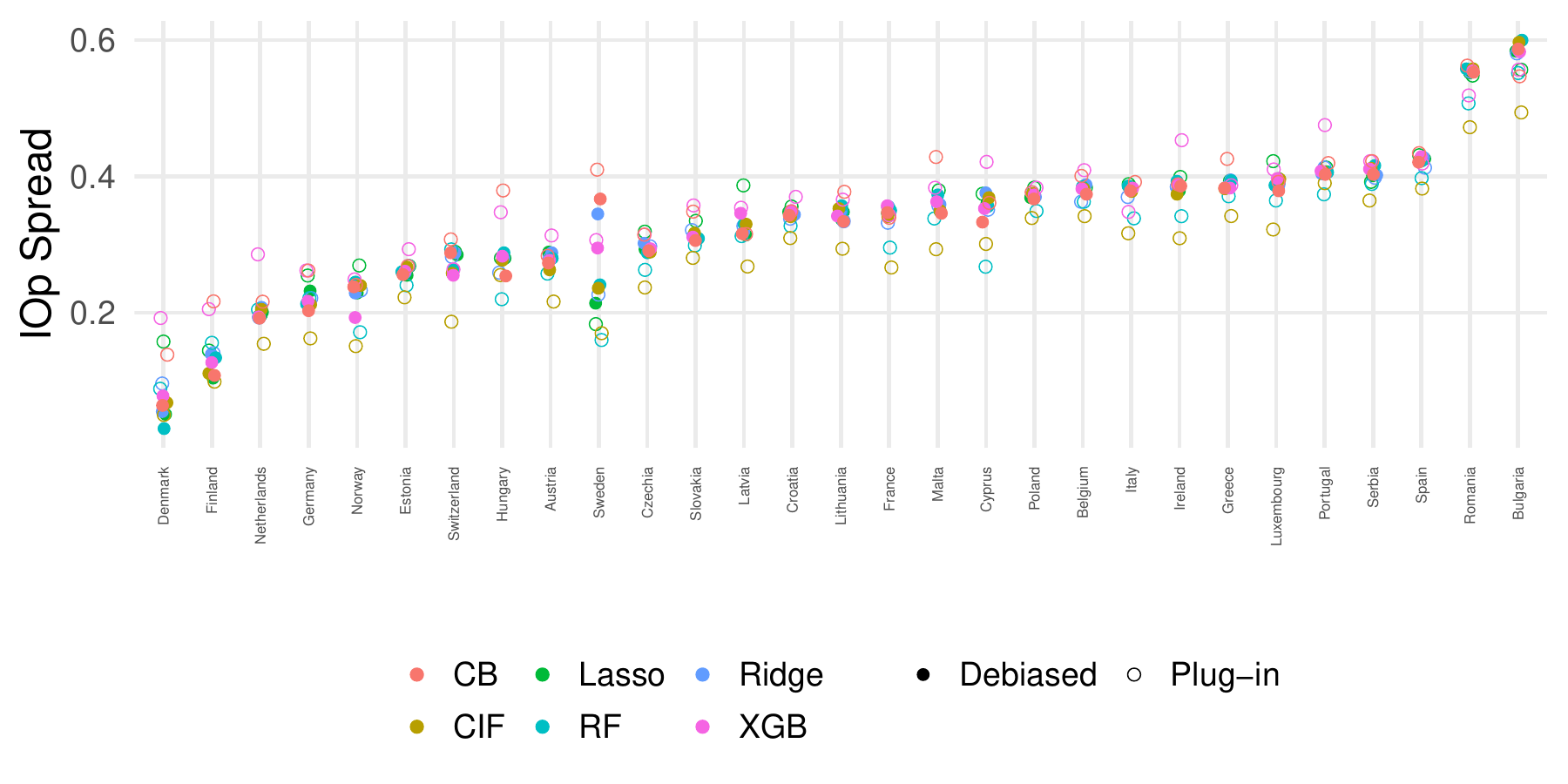}
    \caption{Sensitivity of Plug-in and Debiased to choice of ML in the first-step.}
    \label{fig:iop_spread}
  \end{figure}

\section{Appendices}
\label{Appendix}

\addcontentsline{toc}{section}{Appendices} \renewcommand{\thesubsection}{\Alph{subsection}}

\subsection{Practitioner's Guide to IOp}
\label{app_practitioner_iop}

We provide the R package \texttt{ineqopp}, available at \url{https://github.com/joelters/ineqopp}, to implement the proposed IOp estimator. The package documentation contains examples and details on the available functions.

Valid inference is provided only for \texttt{est\_method = TRUE} and \texttt{CFit = TRUE}, corresponding to the cross-fitted debiased estimator. Other options are implemented only for comparison with plug-in or non-cross-fitted alternatives. Conditional on this choice, practitioners must select the first-step ML method, tune its hyperparameters, and decide whether to report relative IOp. 

\begin{enumerate}
    \item \textbf{ML method.} We recommend a flexible estimator of the conditional mean. The package implements lasso, ridge, random forests, conditional inference forests, XGBoost, CatBoost, and neural networks.
    \item \textbf{Hyperparameter tuning.} Tuning parameters should be selected carefully, for example by cross-validation using \texttt{MLtuning} in the companion \texttt{ML} package, available at \url{https://github.com/joelters/ML}. The same procedure can be used to compare first-step ML methods.
    \item One can decide to get relative IOp as the one reported in the empirical application (the Gini of the predictions over the Gini of $Y$) instead of just the Gini of the predictions.
\end{enumerate}

\subsection{Proofs of General Results}
\label{sec_app_pfs_general_results}

\subsubsection{General representation of adjustment terms}

\label{app_pfsgeneral}

Without loss of generality, we prove all results for the case in which $M = 1$. The general case follows by applying the same arguments for each $m$ in $\{1,...,M\}$.

\begin{proof}[Proof of Lemma \ref{lemma_S_closed_linear_set}]
    Linearity follows directly from the definition of $\mathcal{S}$ and linearity of $\Gamma$. For closedness, take a sequence $s_k \to s$, where $s_k \in \mathcal{S}$ for all $k = 1,2,...$ and $s \in L_2(P_X \times P_X)$. Note that $\Gamma$ is complete since it is a closed subset of the complete space $L_2(P_X)$. Let $\Gamma^0 = \{h \in \Gamma: \mathbb{E}[h(X_i)] = 0\}$. By definition of $\mathcal{S}$, there are sequences $f_k$ and $g_k$ in $\Gamma$ such that $s_k(X_i,X_j) = f_k(X_i) + g_k(X_j)$. Let $f_k^0(x) = f_k(x) - \mathbb{E}[f_k(X_i)]$, $g_k^0(x) = g_k(x) - \mathbb{E}[g_k(X_j)]$ and $c_k =  \mathbb{E}[f_k(X_i)] + \mathbb{E}[g_k(X_j)]$. Then, $s_k(X_i,X_j) = f_k^0(X_i) + g_k^0(X_j) + c_k$ with $f_k^0$ and $g_k^0$ in $\Gamma^0$ and $c_k \in \mathbb{R}$. Note that $T(h) = \mathbb{E}[h(X_i)]$ is continuous since $|T(h) - T(m)| \leq ||h - m||_{L_2(P_X)}$ and hence $\Gamma^0 = \Gamma \cap T^{-1}(\{0\})$ where $T^{-1}(\{0\})$ is closed since $\{0\}$ is closed in the usual topology and $T$ is continuous. Hence, $\Gamma^0$ is a closed subset of $\Gamma$ and hence complete. By independence of $X_i$ and $X_j$ and by the elements of $\Gamma^0$ centered, for $k,m = 1,2,...$  
    \[
    \|s_k - s_m\|_{L_2(P_X \times P_X)}^2 = \|f_k^0 - f_m^0\|^2_{L_2(P_X)} + \|g_k^0 - g_m^0\|^2_{L_2(P_X)} + |c_k - c_m|^2.
    \]
    Since convergent sequences are Cauchy sequences, the left hand side goes to $0$ as $k,m \to \infty$. Since all terms in the right hand side are non-negative, they all have to go to zero. Hence, $\{f_k^0\}_{k=1}^\infty$, $\{g_k^0\}_{k=1}^\infty$ and $\{c_k\}_{k=1}^\infty$ are Cauchy sequences. Since $\Gamma^0$ and $\mathbb{R}$ are complete spaces, $s_k(X_i,X_j) \to f^0(X_i) + g^0(X_j) + c$ for some $f^0$ and $g^0$ in $\Gamma^0$ and $c \in \mathbb{R}$. In fact, by uniqueness of limits in metric/normed spaces, $s(X_i,X_j) = (f^0(X_i) + c) + g^0(X_j) \in \mathcal{S}$ since $\Gamma$ is a linear space including the constants.
\end{proof}

\begin{proof}
[Proof of Lemma \ref{qorth2}] Let $\mathcal{S}$ denote $L_2(P_X\times P_X)$ in case (i) and $\mathcal{S}=\Gamma+\Gamma$ in case (ii). Set
$\alpha_{0}(X_{i},X_{j})=\Pi_{\mathcal{S}}(\delta(X_{i},X_{j},\gamma_{0}))$. Since $\alpha_{0}(\cdot,x)\in\Gamma$ and
$\alpha_{0}(x,\cdot)\in\Gamma$ for all $x\in\mathcal{X}$ (by the last display in this proof and Assumption \ref{Ass_S}) it holds by iterated expectations and Assumption \ref{Ass_linearization},%
\[
\mathbb{E}_{F_{\tau}}[\alpha_{0}(X_{i},X_{j})(c_{1}(Y_{i}-\gamma_{\tau}%
(X_{i}))+c_{2}(Y_{j}-\gamma_{\tau}(X_{j})))]=0.
\]
Thus, by Assumption \ref{Ass_linearization}, the chain rule and $\alpha_{0}\in\mathcal{S}$,%
\begin{align*}
\frac{d}{d\tau}\mathbb{E}[g(W_{i},W_{j},\gamma(F_{\tau}),\theta)]  &
=\frac{d}{d\tau}\mathbb{E}[\delta(X_{i},X_{j},\gamma_{0})(c_{1}\gamma_{\tau
}(X_{i})+c_{2}\gamma_{\tau}(X_{j}))]\\
&  =\frac{d}{d\tau}\mathbb{E}[\alpha_{0}(X_{i},X_{j})(c_{1}\gamma_{\tau}%
(X_{i})+c_{2}\gamma_{\tau}(X_{j}))]\\
&  =\frac{d}{d\tau}\mathbb{E}_{F_{\tau}}[\alpha_{0}(X_{i},X_{j})(c_{1}%
Y_{i}+c_{2}Y_{j}-c_{1}\gamma_{0}(X_{i})-c_{2}\gamma_{0}(X_{j}))]\\
&  =\int\int\phi(w_{i},w_{j},\gamma_{0},\alpha_{0},\theta)K_{H}\left(
dw_{i},dw_{j}\right)
\end{align*}
with $\phi(w_{i},w_{j},\gamma_{0},\alpha_{0},\theta)=\alpha_{0}(x_{i}%
,x_{j})(c_{1}y_{i}+c_{2}y_{j}-c_{1}\gamma_{0}(x_{i})-c_{2}\gamma_{0}(x_{j}))$. To complete the proof for case (ii), using the short notation $\delta_{ij}(\gamma,\theta)\equiv\delta(X_{i}%
,X_{j},\gamma,\theta),$ and noting that by iterated expectations and independence, we
can write
\begin{align*}
\delta_{ij}(\gamma,\theta)  &  =\mathbb{E}[\delta_{ij}(\gamma,\theta)|X_{i}]+\mathbb{E}%
[\delta_{ij}(\gamma,\theta)|X_{j}]-\mathbb{E}[\delta_{ij}(\gamma,\theta)]+U_{ij},\\
&  \equiv\tilde{\delta}_{ij}(\gamma,\theta)+U_{ij},
\end{align*}
where $\mathbb{E}[U_{ij}|X_{i}]=\mathbb{E}[U_{ij}|X_{j}]=0$. This expansion
implies that only $\tilde{\delta}_{ij}$ matters for the derivative in
(\ref{linearization}), so $\delta_{ij}$ could be replaced everywhere by
$\tilde{\delta}_{ij}$. Thus, substituting $\delta_{ij}$ from the last display
in the expression for $\alpha_{0}$ we obtain
\begin{align*}
\alpha_{0}(X_{i},X_{j},\theta)  &  =\Pi_{\mathcal{S}}(\delta(X_{i},X_{j},\gamma
_{0},\theta))\\
&  =\Pi_{\mathcal{S}}\mathbb{E}[\delta_{ij}(\gamma_{0},\theta)|X_{i}]+\Pi
_{\mathcal{S}}\mathbb{E}[\delta_{ij}(\gamma_{0},\theta)|X_{j}]-\mathbb{E}[\delta
_{ij}(\gamma_0,\theta)]\\
&  =\Pi_{\Gamma}\mathbb{E}[\delta_{ij}(\gamma_{0},\theta)|X_{i}]+\Pi_{\Gamma
}\mathbb{E}[\delta_{ij}(\gamma_{0},\theta)|X_{j}]-\mathbb{E}[\delta_{ij}(\gamma_0,\theta)],
\end{align*}
where last expression uses that $\mathcal{S}=\Gamma+\Gamma.$
\end{proof}

\subsubsection{Asymptotic theory proofs}
\label{app_asymptotic_theory_proofs}

    We use the following short-hand notation: $\bar g(\theta) = \mathbb{E}[g(W_i,W_j,\gamma_0,\theta)]$ and 
\begin{align*}
\hat g(\theta) &= \binom{n}{2}^{-1} \sum_{l=1}^L \sum_{(i,j) \in I_l} g(W_i,W_j,\hat{\gamma}_l, \theta), \quad \hat \phi(\theta) = \binom{n}{2}^{-1} \sum_{l=1}^L \sum_{(i,j) \in I_l} \phi(W_i,W_j,\hat{\gamma}_l, \hat{\alpha}_l(\cdot, \theta),\theta), \\
\tilde g(\theta) &= \binom{n}{2}^{-1} \sum_{l=1}^L \sum_{(i,j) \in I_l} g(W_i,W_j,\gamma_0, \theta), \quad \tilde \phi(\theta) = \binom{n}{2}^{-1} \sum_{l=1}^L \sum_{(i,j) \in I_l} \phi(W_i,W_j,\gamma_0, \alpha_0(\cdot, \theta),\theta),
\end{align*}
so that $\hat{\psi}(\theta) = \hat g(\theta) + \hat \phi(\theta)$ and $\tilde{\psi}(\theta) = \tilde g(\theta) + \tilde \phi(\theta)$. DCT and LLN refer to the Dominated Convergence Theorem and the Law of Large Numbers, respectively.

  \begin{proof}[Proof of Theorem \ref{Thm_AT_Consist}] 
    Condition (ii) is Assumption 1 in \cite{newey1991uniform}. Condition (i) implies pointwise convergence of $\hat \psi(\theta)$ to $\bar{g}(\theta)$ (Assumption 2 in \cite{newey1991uniform}) by noting
    \begin{align*}
    \phi(W_i,W_j,\hat{\gamma}_l,\hat{\alpha}_l(\cdot, \theta), \theta) &= \phi(W_i,W_j,\hat{\gamma}_l,\alpha_0(\cdot, \theta),\theta) - \phi(W_i,W_j,\gamma_0,\alpha_0(\cdot, \theta),\theta) \\
    &+ \phi(W_i,W_j,\gamma_0,\hat \alpha_l(\cdot, \theta),\theta) + \hat{\xi}_l(w_i,w_j,\theta),
    \end{align*}
    global robustness of $\alpha$, triangle inequality, conditional Markov and DCT. Conditions (iii) and (iv) imply Assumption 3A in \cite{newey1991uniform}. Then, since conditions (iii) and (iv) imply uniform continuity of $\bar g$ in $\theta$ and $\Theta$ is compact, by the Heine-Cantor Theorem and $\bar g$ not depending on $n$, by Corollary 2.2 in \cite{newey1991uniform}, $\sup_{\theta \in \Theta} |\hat{\psi}(\theta) - \bar{g}(\theta)| = o_p(1)$. By $\bar g$  continuous, \(\Theta\) compact and $\bar g(\theta)=0$ if and only if $\theta=\theta_0$, for every $\varepsilon>0$,
    \[
    \inf_{\theta\in\Theta:\|\theta-\theta_0\|\ge \varepsilon}
    |\bar g(\theta)|>0.
    \]
    Since \(\hat\theta\) satisfies \(\hat\psi(\hat\theta)=o_p(1)\), Theorem 5.9 in \cite{van2000asymptotic} yields \(\hat\theta\to_p\theta_0\).
    \end{proof}

\begin{lemma}
\label{Lemma_AT_Jacobianhat} If Assumption \ref{Ass_AT_Jacobianhat} is
satisfied and if $\bar{\theta} \rightarrow_{p} \theta_{0}$ then $\partial
\hat{\psi}(\bar{\theta})/\partial\theta\to_{p} B$.
\end{lemma}

    \begin{proof}
    [Proof of Lemma \ref{Lemma_AT_Jacobianhat}]
    We start with the following decomposition
    \begin{align*}
        \left| \frac{\partial \hat{\psi}(\bar \theta)}{\partial \theta} - B \right| \leq \left| \frac{\partial \hat{\psi}(\bar \theta)}{\partial \theta} - 
        \frac{\partial \hat{\psi}(\theta_0)}{\partial \theta} \right| + \left| \frac{\partial \hat{\psi}(\theta_0)}{\partial \theta} - 
        \frac{\partial \tilde{\psi}(\theta_0)}{\partial \theta} \right| + 
        \left| \frac{\partial \tilde{\psi}(\theta_0)}{\partial \theta} - B \right|.
    \end{align*}
    For the first term, by Assumption \ref{Ass_AT_Jacobianhat} (iii)-(iv)
    \begin{align*}
        \left| \frac{\partial \hat{\psi}(\bar \theta)}{\partial \theta} - 
        \frac{\partial \hat{\psi}(\theta_0)}{\partial \theta} \right| \leq 
        |\bar \theta - \theta_0|^a \binom{n}{2}^{-1} \sum_{l=1}^L \sum_{(i,j) \in I_l} d(W_i,W_j,\hat \gamma_l, \hat \alpha_l(\cdot, \cdot)) = o_p(1) O_p(1) = o_p(1).
    \end{align*}
    For the second term, 
    \begin{align*}
        &\left| \frac{\partial \hat{\psi}(\theta_0)}{\partial \theta} - 
        \frac{\partial \tilde{\psi}(\theta_0)}{\partial \theta} \right| \leq \left| \frac{\partial \hat{g}(\theta_0)}{\partial \theta} - 
        \frac{\partial \tilde{g}(\theta_0)}{\partial \theta} \right| + \left| \frac{\partial \binom{n}{2}^{-1} \sum_{l=1}^L \sum_{(i,j) \in I_l}\phi(W_i,W_j,\hat{\gamma}_l, \alpha_0, \theta_0)}{\partial \theta} - 
        \frac{\partial \tilde{\phi}(\theta_0)}{\partial \theta} \right| \\
        &+ \left| \frac{\partial \binom{n}{2}^{-1} \sum_{l=1}^L \sum_{(i,j) \in I_l} \phi(W_i,W_j,\gamma_0,\hat{\alpha}_l(\cdot,\theta_0),\theta_0)}{\partial \theta} - 
        \frac{\partial \tilde{\phi}(\theta_0)}{\partial \theta} \right| + \left| \frac{\partial \binom{n}{2}^{-1} \sum_{l=1}^L \hat{\xi}_l(W_i,W_j,\theta_0)}{\partial \theta} \right|,
    \end{align*}
   the first two terms and the last are handled by Assumption \ref{Ass_AT_Jacobianhat} (v) and conditional Markov. The remaining term by global robustness of $\alpha$ together with DCT (by Assumption \ref{Ass_AT_Jacobianhat} (vi)). The last term of the first display goes to zero in probability by the  LLN for U-statistics and the DCT which holds thanks to Assumption \ref{Ass_AT_Jacobianhat} (vi).
    \end{proof}
\bigbreak
\noindent Next, we will show that under our assumptions, 
\begin{equation}
\sqrt{n}\binom{n}{2}^{-1}\sum_{l=1}^{L}\sum_{(i,j)\in I_{l}}\psi(W_{i}%
,W_{j},\hat{\gamma}_{l},\hat{\alpha}_{l},\theta_{0})=\sqrt{n}\binom{n}{2}%
^{-1}\sum_{i<j}\psi(W_{i},W_{j},\gamma_{0},\alpha_{0},\theta_{0})+o_{p}(1).
\label{Eq_AT_equivasymp}%
\end{equation}
This means that the first-order asymptotic properties of the debiased
$\hat{\theta}$ are not affected by the estimation of first-steps. 

\begin{lemma}
\label{Lemma_AT_equiv} If Assumptions \ref{Ass_AT_equiv}%
-\ref{Ass_AT_GRsmallbias} are satisfied then equation (\ref{Eq_AT_equivasymp}) holds.
\end{lemma}

\begin{proof}
    [Proof of Lemma \ref{Lemma_AT_equiv}]Define%
    \begin{align*}
    \hat{R}_{1,ij,l}  &  =g(W_{i},W_{j},\hat{\gamma}_{l},\theta_{0})-g(W_{i}%
    ,W_{j},\gamma_{0},\theta_{0}),\quad\hat{R}_{2,ij,l}=\phi(W_{i},W_{j}%
    ,\hat{\gamma}_{l},\alpha_{0},\theta_{0})-\phi(W_{i},W_{j},\gamma_{0}%
    ,\alpha_{0},\theta_{0}),\\
    \hat{R}_{3,ij,l}  &  =\phi(W_{i},W_{j},\gamma_{0},\hat{\alpha}_{l},\theta
    _{0})-\phi(W_{i},W_{j},\gamma_{0},\alpha_{0},\theta_{0}),\quad(i,j)\in I_{l}.
    \end{align*}
    Then%
    \begin{equation}
    g(W_{i},W_{j},\hat{\gamma}_{l},\theta_{0})+\phi(W_{i},W_{j},\hat{\gamma}%
    _{l},\hat{\alpha}_{l},\theta_{0})-\psi(W_{i},W_{j},\gamma_{0},\alpha
    _{0},\theta_{0})=\hat{R}_{1,ij,l}+\hat{R}_{2,ij,l}+\hat{R}_{3,ij,l}+\hat{\xi
    }_{l}(W_{i},W_{j}). \label{eq1lemma10}%
    \end{equation}
    Let $\hat{A}_{ij,l} = \hat{R}_{1,ij,l}+\hat{R}_{2,ij,l}+\hat{R}_{3,ij,l}$. We can rewrite
    \[
     \hat{A}_{ij,l} + \hat{\xi
    }_{l}(W_{i},W_{j}) = \left(\hat{A}_{ij,l} - \mathbb{E}[\hat{A}_{ij,l} | S_l^c]\right) + \mathbb{E}[\hat{A}_{ij,l} | S_l^c] + \hat{\xi
    }_{l}(W_{i},W_{j}).
    \]
    Since pairs in $I_{l}$ are dependent only when
    one or two of the members of the pair coincide (also we omit the fact that we
    are dealing with vectors since the convergence of the vector is the
    convergence of its elements) we have that for $b = 1,2,3$ and $|A|$ the cardinality of set $A$
    \begin{align*}
    &  \mathbb{E}\biggl[\biggl(\sqrt{n}\binom{n}{2}^{-1}\sum_{(i,j)\in I_{l}}%
    (\hat{R}_{b,ij,l}-\mathbb{E}(\hat{R}_{b,ij,l}|S_l^c))\biggr)^{2}%
    \biggr|S_l^c\biggr]=\\
    &  n\binom{n}{2}^{-2}\biggl[|I_l|\mathbb{V}ar(\hat{R}_{b,ij,l}%
    |S_l^c)+|I_l|\cdot 2 (S_l - 1) \mathbb{C}ov(\hat{R}_{b,ij,l},\hat{R}_{b,ik,l}%
    |S_l^c)\biggr].
    \end{align*}
   Since $|I_l|$ is of the same order as $\binom{n}{2}$, $n \binom{n}{2}^{-2} |I_l| \to 0$. $n \binom{n}{2}^{-2} |I_l|  2 (|S_l| - 1)$ is of the same order as $ n |S_l|/|I_l|$ and $n |S_l|/|I_l| \to c < \infty$ by Lemma \ref{lemma_CF} in the Online Appendix \ref{app_CF}. Hence%

    \begin{align*}
    \mathbb{E}\biggl[\biggl(\sqrt{n}\binom{n}{2}^{-1}\sum_{(i,j)\in I_{l}}%
    (\hat{R}_{b,ij,l}-\mathbb{E}(\hat{R}_{b,ij,l}|S_l^c)\biggr)^{2}%
    \biggr|S_l^c\biggr] &\leq c \cdot \mathbb{C}ov(\hat{R}_{b,ij,l},\hat{R}_{b,ik,l}|S_l^c)+o_{P}(1)\\
    &\leq c \cdot\sqrt{\mathbb{E}(\hat{R}_{b,ij,l}^{2}|S_l^c)\mathbb{E}(\hat
    {R}_{b,ik,l}^{2}|S_l^c)}+o_{P}(1)\rightarrow_{p}0,
    \end{align*}
    where convergence for $b = 1,2,3$ follows from Assumption
    \ref{Ass_AT_equiv}. Then, by the conditional Markov inequality, triangle inequality, and DCT
    \[
    \sqrt{n}\binom{n}{2}^{-1}\sum_{(i,j)\in I_{l}}\left(\hat{A}_{ij,l} - \mathbb{E}[\hat{A}_{ij,l} | S_l^c]\right)\rightarrow_{p}0.
    \]
    Note that $\mathbb{E}[\hat{R}_{1,ij,l}+\hat
    {R}_{2,ij,l}|S_l^c]=\mathbb{E}[\psi(W_{i},W_{j},\hat{\gamma}_{l}%
    ,\alpha_{0},\theta_{0}) | S_l^c]$. Therefore, by Assumption
    \ref{Ass_AT_GRsmallbias}%
    \[
    \biggl|\sqrt{n}\binom{n}{2}^{-1}\sum_{(i,j)\in I_{l}}\mathbb{E}(\hat
    {R}_{1,ij,l}+\hat{R}_{2,ij,l}|S_l^c)\biggr|\leq 2\sqrt{n}|\bar{\psi}(\hat{\gamma}%
    _{l},\alpha_{0},\theta_{0})|\rightarrow_{p}0,
    \]
    where $\bar \psi(\gamma,\alpha, \theta) = \int \int \psi(w_i,w_j,\gamma,\alpha,\theta) F_0(dw_i) F_0(dw_j)$. By Assumption \ref{Ass_AT_GRsmallbias} (i), $\mathbb{E}[\hat{A}_{ij,l} | S_l^c] = \mathbb{E}(\hat
    {R}_{1,ij,l}+\hat{R}_{2,ij,l}|S_l^c)$. Hence, using Equation (\ref{2}) with $\tau = 0$ and the triangle inequality%
    \[
    \sqrt{n}\binom{n}{2}^{-1}\sum_{(i,j)\in I_{l}}(\hat{R}_{1,ij,l}+\hat
    {R}_{2,ij,l}+\hat{R}_{3,ij,l}) =  \sqrt{n}\binom{n}{2}^{-1}\sum_{(i,j)\in I_{l}} \left(\hat{A}_{ij,l} - \mathbb{E}[\hat{A}_{ij,l} | S_l^c]\right) + \mathbb{E}[\hat{A}_{ij,l} | S_l^c] \rightarrow_{p}0.
    \]
    Hence, by Assumption \ref{Ass_AT_interaction} and the triangle inequality
    \begin{align*}
    &  \sqrt{n}\binom{n}{2}^{-1}\sum_{(i,j)\in I_{l}}\psi(W_{i},W_{j},\hat{\gamma
    }_{l},\hat{\alpha}_{l},\theta_{0})-\psi(W_{i},W_{j},\gamma_{0},\alpha
    _{0},\theta_{0})\\
    &  =\sqrt{n}\binom{n}{2}^{-1}\sum_{(i,j)\in I_{l}}(\hat{R}_{1,ij,l}+\hat{R}_{2,ij,l}+\hat{R}%
    _{3,ij,l}+\hat{\xi}_{l}(W_{i},W_{j}))\rightarrow_{p}0.
    \end{align*}
    Convergence of the interaction term follows from Assumption \ref{Ass_AT_interaction} (ii). Assumption \ref{Ass_AT_interaction} (i) implies (ii) by adding and subtracting $\mathbb{E}[\hat{\xi}_{l}(W_{i},W_{j}) | S_l^c]$ and using the same arguments as we have used for $\hat{R}_{b,ij,l}$ for $b = 1,2,3$. The same remains true when we sum across finite folds.
    \end{proof}

    \begin{proof}
    [Proof of Theorem \ref{Thm_AT_AN}]By the Mean Value Theorem and Lemmas \ref{Lemma_AT_Jacobianhat} and \ref{Lemma_AT_equiv}, 
\begin{eqnarray*}
0 &=&\sqrt{n}\hat{\psi}(\theta _{0})+\frac{\partial\hat{\psi}(\bar{\theta})}{\partial\theta} \sqrt{n}(\hat{\theta}-\theta _{0})
\\
&=&\sqrt{n}\tilde{\psi}(\theta _{0})+B\sqrt{n}(\hat{\theta}-\theta
_{0})+o_{P}(1).
\end{eqnarray*}%
    The normality result follows from standard U-statistic theory (see Theorem
    12.3 in \cite{van2000asymptotic}) and Slutsky's lemma. Hence $\sqrt{n}%
    (\hat{\theta} - \theta_{0}) \rightarrow_{d} \mathcal{N}(0,V)$. 
    \end{proof}

\begin{proof}
    [Proof of Proposition \ref{prop_Sigmahat_consistency}] Let $\hat{\psi}_{ij} = \psi(W_i,W_j,\hat{\gamma},\hat{\alpha},\hat{\theta})$, $\psi_{ij} = \psi(W_i,W_j,\gamma_0,\alpha_0,\theta_0)$ and define $\hat{\chi}_i = (n-1)^{-1}\sum_{j \neq i} \hat{\psi}_{ij}$, $\tilde{\chi}_i = (n-1)^{-1}\sum_{j \neq i} \psi_{ij}$ and $\chi_i = \mathbb{E}[\psi_{ij}|W_i]$. Without loss of generality, we focus on the scalar case. Note that $\hat{\Sigma} = 4n^{-1}\sum_{i=1}^n \hat{\chi}_i^2$ and we want to show that $\hat{\Sigma} \to_p 4\mathbb{E}[\chi_i^2] = \Sigma$. To do this note that
    \[
    \frac{1}{n}\sum_{i=1}^n(\hat{\chi}_i - \tilde{\chi}_i + \tilde{\chi}_i)^2  = \frac{1}{n}\sum_{i=1}^n \tilde{\chi}_i^2 + \frac{1}{n}\sum_{i=1}^n (\hat{\chi}_i - \tilde{\chi}_i)^2 + \frac{2}{n}\sum_{i=1}^n (\hat{\chi}_i - \tilde{\chi}_i)\tilde{\chi}_i.
    \]
    The first term in the right hand side goes in probability to $\mathbb{E}[\chi_i^2]$ by the strong LLN for U-statistics (p.190 in \cite{serfling1980approximation}). By Cauchy-Schwartz, 
    \[
    \frac{2}{n}\sum_{i=1}^n (\hat{\chi}_i - \tilde{\chi}_i)\tilde{\chi}_i \leq 2\biggl(\frac{1}{n}\sum_{i=1}^n (\hat{\chi}_i - \tilde{\chi}_i)^2 \biggr)^{\frac{1}{2}}\biggl(\frac{1}{n}\sum_{i=1}^n \tilde{\chi}_i^2 \biggr)^{\frac{1}{2}}.
    \]
    Hence, if we show that $n^{-1}\sum_{i=1}^n |\hat{\chi}_i - \tilde{\chi}_i|^2 \to_p 0$ we have that $n^{-1}\sum_{i=1}^n \hat{\chi}_i^2 = \mathbb{E}[\chi_i^2] + o_p(1)$. It follows by the $C_r$ inequality and the assumptions in the lemma that
    \begin{align*}
        \frac{1}{n}\sum_{i=1}^n |\hat{\chi}_i - \tilde{\chi}_i|^2 &=  \frac{1}{n}\sum_{i=1}^n \biggl|\frac{1}{n-1} \sum_{j\neq i}[\hat{g}_{ij} - g_{ij} + \hat{\phi}_{ij} - \phi_{ij}]\biggr|^2 \\
        &\leq \frac{2}{n}\sum_{i=1}^n |\hat{g}_{-i} - g_{-i}|^2 + \frac{2}{n}\sum_{i=1}^n |\hat{\phi}_{-i} - \phi_{-i}|^2 \to_p 0.
    \end{align*}
    For convergence of $\hat B = \binom{n}{2}^{-1} \sum_{i<j} \partial g(W_i,W_j,\hat \gamma,\hat \theta)/\partial\theta$ to $B$, the same proof as in Lemma \ref{Lemma_AT_Jacobianhat} applies in a simpler way since we do not have to deal with the terms related to $\phi$. We simply replace $\psi$ by $g$ and employ the unconditional Markov inequality. 
    \end{proof}

    \begin{proof}[Proof of Corollary \ref{corollary_AT_alphaknownuptogamma}] The result follows immediately from Theorems 1 and 2 and Proposition 1 after replacing the
    moment $g$ by the already orthogonal moment $\psi$. Since $\phi=0$, there is no auxiliary nuisance $\alpha$, no interaction remainder $\xi_l$, and all assumptions and proof terms involving $\phi$, $\alpha$, or $\xi_l$ vanish. 
    \end{proof}

\putbib[references]   
\end{bibunit}
\clearpage        
\pagenumbering{arabic}   
\setcounter{page}{1}     

\begin{center}
    \LARGE \textbf{Online Appendix}
\end{center}
\begin{bibunit}[ectabib] 
\addcontentsline{toc}{section}{Appendices} \renewcommand{\thesection}{\Alph{section}}

\section{Proofs of Inequality of Opportunity}

\label{app_pfsIOp}

\begin{proof}[Proof of Proposition \ref{prop_UFSIF_IOp}]
We use the short notation $\Delta_\tau \equiv \gamma_\tau(X_i)-\gamma_\tau(X_j)$ and $\Delta_0 \equiv \gamma_0(X_i)-\gamma_0(X_j)$. By $|\Delta_\tau|=sgn(\Delta_\tau)\Delta_\tau$ and the chain rule,
\begin{eqnarray*}
\frac{d}{d\tau}\mathbb{E}[g(W_i,W_j,\gamma(F_\tau),\theta)]
&=&
\frac{d}{d\tau}\mathbb{E}[\theta\{\gamma_\tau(X_i)+\gamma_\tau(X_j)\}]
-
\frac{d}{d\tau}\mathbb{E}[|\Delta_\tau|]
\\
&=&
\frac{d}{d\tau}\mathbb{E}[\theta\{\gamma_\tau(X_i)+\gamma_\tau(X_j)\}]
-
\frac{d}{d\tau}\mathbb{E}[sgn(\Delta_0)\Delta_\tau]
-
\frac{d}{d\tau}\mathbb{E}[sgn(\Delta_\tau)\Delta_0],
\end{eqnarray*}
where all derivatives are evaluated at $\tau=0$. The first two summands are already in the linear form required by
(\ref{linearization}). Hence, Lemma \ref{qorth2} yields the corresponding
adjustment terms
\begin{align*}
\phi_1(w_i,w_j,\gamma_0,\theta)
&=
\theta (y_i+y_j-\gamma_0(x_i)-\gamma_0(x_j)), \\
\phi_2(w_i,w_j,\gamma_0,\theta)
&=
-
sgn(\Delta_0)
(y_i-y_j-\gamma_0(x_i)+\gamma_0(x_j)).
\end{align*}
It remains to show that
\[
\left.
\frac{d}{d\tau}
\mathbb{E}[sgn(\Delta_\tau)\Delta_0]
\right|_{\tau=0}
=
0.
\]
We establish this directly. Consider the path $F_\tau=(1-\tau)F_0+\tau H$. For the conditional mean first step $\gamma_\tau(x)=\mathbb E_{F_\tau}[Y\mid X=x]$. Writing $f_{0,X}$ and $h_X$ for the marginal densities of $X$ under $F_0$
and $H$, respectively, and writing
\[
\gamma_H(x)=\mathbb E_H[Y\mid X=x],
\]
we have, for $0\leq \tau\leq \bar \tau$,
\[
\gamma_\tau(x)
=
\frac{(1-\tau)\gamma_0(x)f_{0,X}(x)+\tau\gamma_H(x)h_X(x)}
{(1-\tau)f_{0,X}(x)+\tau h_X(x)}.
\]
Therefore,
\[
\gamma_\tau(x)-\gamma_0(x)
=
\tau
\frac{
h_X(x)\{\gamma_H(x)-\gamma_0(x)\}
}
{
(1-\tau)f_{0,X}(x)+\tau h_X(x)
}.
\]
We restrict attention to regular paths for which the denominator is positive
for $0\leq \tau\leq\bar\tau$ and
\[
C_H(x)
:=
\sup_{0\leq s\leq\bar\tau}
\left|
\frac{
h_X(x)\{\gamma_H(x)-\gamma_0(x)\}
}
{
(1-s)f_{0,X}(x)+s h_X(x)
}
\right|
\]
is well defined with $\mathbb E[C_H(X_i)]<\infty$. This restriction is only on
the class of paths used to verify the pathwise derivative. Such regularity restrictions are standard in
pathwise differentiability arguments and still allow for a sufficiently rich class of perturbations. It follows that, for all $0\leq\tau\leq\bar\tau$,
\[
|\gamma_\tau(x)-\gamma_0(x)|
\leq
\tau C_H(x).
\]
Hence,
\[
|\Delta_\tau-\Delta_0|
\leq
|\gamma_\tau(X_i)-\gamma_0(X_i)|
+
|\gamma_\tau(X_j)-\gamma_0(X_j)|
\leq
\tau V_H,
\]
where
\[
V_H
=
V_H(X_i,X_j)
=
C_H(X_i)+C_H(X_j).
\]
Consequently, if $|\Delta_0|>\tau V_H$, then $\Delta_\tau$ cannot cross zero
relative to $\Delta_0$, and therefore $sgn(\Delta_\tau)=sgn(\Delta_0)$. Thus,
\[
\{sgn(\Delta_\tau)\neq sgn(\Delta_0)\}
\subseteq
\{|\Delta_0|\leq \tau V_H\}.
\]
Since $sgn(\cdot)$ takes values in $\{-1,0,1\}$,
\[
|sgn(\Delta_\tau)-sgn(\Delta_0)|
\leq
2\cdot 1\{|\Delta_0|\leq \tau V_H\}.
\]
Therefore,
\begin{eqnarray*}
\left|
\mathbb E\left[
\Delta_0
\left\{
\frac{sgn(\Delta_\tau)-sgn(\Delta_0)}{\tau}
\right\}
\right]
\right|
&\leq&
2\mathbb E\left[
|\Delta_0|
\frac{1\{|\Delta_0|\leq \tau V_H\}}{\tau}
\right]
\\
&\leq&
2\mathbb E\left[
V_H1\{0<|\Delta_0|\leq \tau V_H\}
\right]
\\
&\rightarrow&0,
\end{eqnarray*}
as $\tau\downarrow0$. The convergence follows from continuity from above of
the finite positive measure
\[
\mu(A)
=
\int_A
V_H(x_1,x_2)
F_0(dx_1)F_0(dx_2),
\]
because $\mathbb E[V_H]<\infty$ and the sets $A_\tau = \{(x_1,x_2):0<|\Delta_0(x_1,x_2)|\leq \tau V_H(x_1,x_2)\}$ decrease to the empty set as $\tau\downarrow0$. Hence,
\[
\left.
\frac{d}{d\tau}
\mathbb E[sgn(\Delta_\tau)\Delta_0]
\right|_{\tau=0}
=
0,
\]
and the desired adjustment term is $\phi=\phi_1+\phi_2$.
\end{proof}
\bigbreak

\begin{proof}[Proof of Proposition \ref{prop_sgn_consistency_smoothness}]
    We start with the smoothness part and leave the consistency of the sign for the end of the proof. Since $|sgn(\Delta_{\gamma}) - sgn(\Delta_0)| \leq 2 \cdot 1_A$ a.s. where $A = \{(x_i,x_j): sgn(\gamma(x_i) - \gamma(x_j)) \neq sgn(\gamma_0(x_i) - \gamma_0(x_j)) \}$, we have that
    \begin{align*}
        |\mathbb{E}[(sgn(\Delta_{\gamma}) - sgn(\Delta_0))\Delta_0] | \leq 2 \mathbb{E}[|\Delta_0| 1_A].
    \end{align*}
    For some $t > 0$, we can decompose into a small margin and a large margin term
    \[
     \mathbb{E}[|\Delta_0| 1_A 1(0 < |\Delta_0| \leq t)] + \mathbb{E}[|\Delta_0| 1_A 1(|\Delta_0| > t)].
    \]
    Let us now deal with the small margin term. We restrict the analysis to the set $A$ where $\Delta_{\gamma} \cdot \Delta_0 < 0$ or $\Delta_{\gamma} \cdot \Delta_0 = 0$. 
    If $\Delta_0 > 0$, then $\Delta_{\gamma} < 0$ meaning that
    \begin{align*}
        \Delta_{\gamma} &= \Delta_0 - (\gamma_0(X_i) - \gamma(X_i)) + \gamma_0(X_j) - \gamma(X_j) < 0\\
        \implies &- (\gamma_0(X_i) - \gamma(X_i)) + \gamma_0(X_j) - \gamma(X_j) < - \Delta_0 \\
         \implies &|\Delta_0| < |\gamma_0(X_i) - \gamma(X_i)| + |\gamma_0(X_j) - \gamma(X_j)|.
    \end{align*}
When $\Delta_0 < 0$, then $\Delta_{\gamma} > 0$ so 
\begin{align*}
        &- (\gamma_0(X_i) - \gamma(X_i)) + \gamma_0(X_j) - \gamma(X_j) > - \Delta_0 \\
         \implies &|-(\gamma_0(X_i) - \gamma(X_i)) + \gamma_0(X_j) - \gamma(X_j)| > - \Delta_0 = |\Delta_0| \\
         \implies &|\Delta_0| < |\gamma_0(X_i) - \gamma(X_i)| + |\gamma_0(X_j) - \gamma(X_j)|. 
\end{align*}
If $\Delta_0 = 0$ then $\mathbb{E}[|\Delta_0| 1_A] = 0$ and if $\Delta_\gamma = 0$
\[
|\Delta_0| = |\Delta_0 - \Delta_\gamma| = | \gamma_0(X_i) - \gamma(X_i) - (\gamma_0(X_j) - \gamma(X_j)) | \leq  | \gamma_0(X_i) - \gamma(X_i)| + | \gamma_0(X_j) - \gamma(X_j) |.
\]
Hence,
\begin{align*}
    \mathbb{E}[|\Delta_0| 1_A 1( 0 < |\Delta_0| \leq t)] &\leq 2\mathbb{E}[|\gamma(X_i) - \gamma_0(X_i)| 1(0<|\Delta_0| \leq t)].
\end{align*}
Now, using Hölder's inequality with $q' = q/(q-1)$ we get
\begin{align*}
    \mathbb{E}[|\Delta_0| 1_A 1(0 < |\Delta_0| \leq t)] &\leq 2 ||\gamma - \gamma_0||_q P(0 < |\Delta_0| \leq t)^{1/q'},
\end{align*}
so by definition of $q'$ and Assumption \ref{Ass_AT_iop} (iii) we get
\[
\mathbb{E}[|\Delta_0| 1_A 1(0 < |\Delta_0| \leq t)] \leq 2C ||\gamma - \gamma_0||_q t^{\frac{\beta(q-1)}{q}}.
\]
For the large margin, let $B_i = |\gamma_0(X_i) - \gamma(X_i)|$ for notational brevity. Using again that on $A$: $|\Delta_0| < |\gamma_0(X_i) - \gamma(X_i)| + |\gamma_0(X_j) - \gamma(X_j)|$ (also inside the indicator) we show that
\begin{align*}
    \mathbb{E}[|\Delta_0| 1_A 1(|\Delta_0| > t)] &\leq \mathbb{E}[(B_i + B_j) 1(B_i + B_j > t)] \\
    &\leq 2 \mathbb{E}[B_i 1(B_i > t/2)] + 2 \mathbb{E}[B_i 1(B_j > t/2)] \\
    &\leq 4C_{\Delta} ||\gamma - \gamma_0||_q P(B_i > t/2)^{\frac{q-1}{q}} \\
    &\leq 4 \cdot 2^{q-1}\cdot  C_{\Delta} \frac{||\gamma - \gamma_0||_q^q}{t^{q-1}},
\end{align*}
where we have used that $a+b > t$ implies $\{a > t/2\} \cup \{b > t/2\}$ and monotonicity of measures, then Hölder and Markov inequalities. Putting all together (for $q \geq 1$) we have
\begin{align*}
    |\mathbb{E}[(sgn(\Delta_{\gamma}) - sgn(\Delta_0))\Delta_0] | &\leq 4 \cdot 2^{q-1}\cdot  C \left( ||\gamma - \gamma_0||_q t^{\frac{\beta (q-1)}{q}} + \frac{||\gamma - \gamma_0||_q^q}{t^{q-1}} \right).
\end{align*}
Optimizing $t$ and plugging in the optimal, we get
\[
|\mathbb{E}[(sgn(\Delta_{\gamma}) - sgn(\Delta_0))\Delta_0] \leq C(\beta,q) ||\gamma - \gamma_0||_q^{\frac{q(1+\beta)}{q + \beta}}.
\]
Now, for the supremum norm, notice that by the derivations above, we have that $A$ implies
\begin{align*}
    |\Delta_0| \leq 2 ||\gamma - \gamma_0||_\infty,
\end{align*}
and hence also $1_A \leq 1(|\Delta_0| \leq 2 ||\gamma - \gamma_0||_\infty)$ a.s. Then
\begin{align*}
    |\mathbb{E}[(sgn(\Delta_{\gamma}) - sgn(\Delta_0))\Delta_0] &= |\mathbb{E}[(sgn(\Delta_{\gamma}) - sgn(\Delta_0))\Delta_0 1(|\Delta_0| > 0)]| \\
    &\leq 2 \mathbb{E}[|\Delta_0| 1_A 1(0<|\Delta_0|)] \\
    &\leq 2 ||\gamma - \gamma_0||_\infty P(0 < |\Delta_0| \leq 2 ||\gamma- \gamma_0||_\infty).
\end{align*}
Hence, by Assumption \ref{Ass_AT_iop} (iii)
\[
|\mathbb{E}[(sgn(\Delta_{\gamma}) - sgn(\Delta_0))\Delta_0] \leq C ||\gamma - \gamma_0||_\infty^{1 + \beta}.
\]
For the consistency of the sign difference, note that
\begin{align*}
    \mathbb{E}[|sgn(\Delta_{\gamma}) - sgn(\Delta_0)|] &= 2\mathbb{E}[1_A 1(X_i \neq X_j)] \\
    &= 2 \mathbb{E}[1_A 1(X_i \neq X_j) 1(|\Delta_0| \leq t)] + 2 \mathbb{E}[1_A 1(|\Delta_0| > t)1(X_i \neq X_j)].
\end{align*}
Let now $t \equiv t_n \downarrow0$, then the first term goes to zero by Assumption \ref{Ass_AT_iop} (ii). For the second term, we use the same arguments as above to show that
\begin{align*}
    \mathbb{E}[1_A 1(|\Delta_0| > t_n)1(X_i \neq X_j)] &\leq \mathbb{E}[1(B_i + B_j > t_n)] \\
    &\leq 4 \frac{||\gamma - \gamma_0||_1}{t_n}.
\end{align*}
Hence, using $\hat{\gamma}_l$ instead of $\gamma$ and conditioning on $S_l^c$, the consistency of the sign difference follows by selecting a sequence $t_n \downarrow0$ sufficiently slow.

\end{proof}

\begin{proof}[Proof of Proposition \ref{Prop_AT_IOp}]
By Corollary \ref{corollary_AT_alphaknownuptogamma} we can just check the needed high-level assumptions in the general asymptotic theory replacing $g$ by our $\psi$ and dropping all terms related to $\phi$. For Theorem \ref{Thm_AT_Consist} (i), by Cauchy-Schwartz
\begin{align*}
    \iint &|\psi(w_i,w_j,\hat{\gamma}_l,\theta) - \psi(w_i,w_j,\gamma_0,\theta)| F_0(dw_i) F_0(dw_j)  \\
    &\leq \sqrt{\iint |sgn(\Delta_{\hat \gamma}) - sgn(\Delta_0)|^2 F_0(dw_i) F_0(dw_j)} \sqrt{\mathbb{E}[(Y_i - Y_j)^2]}  \to_p 0,
\end{align*}
where the convergence follows from Proposition \ref{prop_sgn_consistency_smoothness} and finite second moment of $Y_i$. Since $\Theta = [-1,1]$ and
\[
|\psi(W_i,W_j,\gamma,\tilde \theta) - \psi(W_i,W_j,\gamma, \theta)| = |Y_i + Y_j| \cdot |\tilde \theta - \theta|,
\]
so the Lipschitz condition for consistency holds for any $\gamma$ by finite first moment of $Y_i$. Hence, by Theorem \ref{Thm_AT_Consist}, we have $\hat \theta \to_p \theta_0$. For Assumption \ref{Ass_AT_equiv}, if conditional second moments of $Y_i$ given $X_i$ are bounded
\begin{align*}
    \mathbb{E}[\hat{A}_{ij,l}^2 | S_l^c] &= \mathbb{E}\left(  \left[sgn(\Delta_0) - sgn(\Delta_{\hat{\gamma}_l})\right]^2(Y_i-Y_j)^2 | S_l^c \right)  \\
    &\leq C \mathbb{E}\left( |sgn(\Delta_0) - sgn(\Delta_{\hat{\gamma}_l})| \, | S_l^c \right) \to_p 0,
\end{align*}
or if $2+\delta$ moments of $Y_i - Y_j$ exist
\begin{align*}
    \mathbb{E}[\hat{A}_{ij,l}^2 | S_l^c] &= \mathbb{E}\left(  \left[sgn(\Delta_0) - sgn(\Delta_{\hat{\gamma}_l})\right]^2(Y_i-Y_j)^2 | S_l^c \right)  \\
    &\leq C \mathbb{E}\left( |sgn(\Delta_0) - sgn(\Delta_{\hat{\gamma}_l})| \, | S_l^c \right)^{\frac{\delta}{2+\delta}} \mathbb{E}[|Y_i - Y_j|^{2+\delta}] \to_p 0,
\end{align*}
The convergence in probability follows by Proposition \ref{prop_sgn_consistency_smoothness}.  To check Assumption \ref{Ass_AT_GRsmallbias} (ii) we have that by Proposition \ref{prop_sgn_consistency_smoothness}
\begin{align*}
    \sqrt{n} |\mathbb{E}[(sgn(\Delta_0) - sgn(\Delta_{\hat{\gamma}_l}))\Delta_0 | S_l^c] | &\leq \begin{cases}
    \sqrt{n}C(\beta,q) ||\hat{\gamma}_l - \gamma_0||_q^{\frac{q(1+\beta)}{q+ \beta}} &\text{ if } q \in [1,\infty),\\
    \sqrt{n}2C_{\Delta} ||\hat{\gamma}_l - \gamma_0||_\infty^{1+\beta} &\text{ if } q = \infty.
\end{cases}
\end{align*}
Under Assumption \ref{ass_rates_iop}, for any choice of $q$ we get the result that 
\begin{align*}
    \sqrt{n} |\mathbb{E}[(sgn(\Delta_0) - sgn(\Delta_{\hat{\gamma}_l}))\Delta_0 | S_l^c] | &= o_p(1).
\end{align*}
Assumption \ref{Ass_AT_Jacobianhat} follows trivially by linearity of $\psi$ in $\theta$ and $0 < \mathbb{E}[Y_i] < \infty$. By Theorem \ref{Thm_AT_AN}, we have $\sqrt{n}(\hat \theta - \theta_0) \to_d \mathcal{N}(0,V)$. For Assumption \ref{ass_AT_Sigmahat}, by $C_r$ inequality, we need to control
\[
\frac{(\hat{\theta} - \theta_0)^2}{n}\sum_{i=1}^n \left(\frac{1}{n-1} \sum_{j \neq i} (Y_i + Y_j) \right)^2, \quad 
\frac{1}{n} \sum_{i= 1}^n \left(\frac{1}{n-1} \sum_{j \neq i} (sgn(\hat{\Delta}) - sgn(\Delta_0))(Y_i - Y_j) \right)^2.
\]
For the first, letting $\bar{Y}_n = (1/n)\sum_{i=1}^n Y_i$,
\[
 \frac{(\hat{\theta} - \theta_0)^2}{n}\sum_{i=1}^n \left(\frac{1}{n-1} \sum_{j \neq i} (Y_i + Y_j) \right)^2 \leq C(\hat{\theta} - \theta_0)^2\left(\bar{Y}_n + \frac{1}{n} \sum_{i=1}^n Y_i^2 \right) \to_p 0.
\]
For the other term
\begin{align*}
    &P\left(\frac{1}{n} \left| \sum_{i=1}^n \left( \frac{1}{n-1} \sum_{j \neq i} (sgn(\hat{\Delta}) - sgn(\Delta_0))(Y_i - Y_j) \right)^2 \right| > \varepsilon \right)\\
    &\leq \frac{1}{\varepsilon(n-1)^2}\mathbb{E} \left[ \left| \sum_{j \neq i} (sgn(\hat{\Delta}) - sgn(\Delta_0))(Y_i - Y_j) \right|^2 \right] \\
    &= \frac{1}{\varepsilon(n-1)^2}\left((n-1)\mathbb{E} \left[ \hat{A}_{ij}^2  \right] + 
    (n-1)(n-2) \mathbb{E} \left[  \hat{A}_{ij} \hat{A}_{ik} \right] \right),
\end{align*}
where $\hat{A}_{ij} \equiv (sgn(\hat{\gamma}(X_i) - \hat{\gamma}(X_j)) - sgn(\gamma_0(X_i) - \gamma_0(X_j)))(Y_i - Y_j)$. The first term goes to zero since it is divided by $n-1$ and by the boundedness of the sign and boundedness of the conditional variance of $Y_i$. The second term goes to zero by applying Cauchy Schwarz and noting that by the finiteness of the conditional variance of $Y_i$
\[
\sqrt{\mathbb{E}[\hat{A}_{ij}^2]} \leq C \sqrt{\mathbb{E}\left( |sgn(\hat{\Delta} - \Delta_0)| \right)} \to 0,
\]
by the unconditional analogue of the sign-consistency argument in Proposition \ref{prop_sgn_consistency_smoothness}, using Assumption \ref{ass_Vhatcons_IOp} in place of the conditional $L_1$ consistency of $\hat \gamma_l$. The Jacobian estimator is consistent by the LLN so $\hat V \to_p V$ by the Continuous Mapping Theorem.

\end{proof}

\section{Proofs of AUC}

\label{app_pfsAUC}

\begin{proof}[Proof of Proposition \ref{AUCrep}]
Using
\[
1(a>b)+\frac12 1(a=b)=\frac{1+sgn(a-b)}{2}, \quad 1(a<b)+\frac12 1(a=b)=\frac{1-sgn(a-b)}{2},
\]
the symmetrized representation of the AUC gives
\[
\theta_0
=
\frac12
+
\frac{
\mathbb{E}\left[
sgn(\gamma_0(X_i)-\gamma_0(X_j))
\{Y_i(1-Y_j)-Y_j(1-Y_i)\}
\right]
}
{4p_0(1-p_0)}.
\]
Since
\[
Y_i(1-Y_j)-Y_j(1-Y_i)=Y_i-Y_j,
\]
the stated representation follows. Orthogonality follows since $A_0$
is the same pairwise sign functional appearing in the debiased Gini
representation, while $p_0$ does not depend on $\gamma_0$.
\end{proof}

\begin{proof}[Proof of Corollary \ref{cor_AUC_AN}]
Let
\[
h_A(W_i)
=
\mathbb E\!\left[
sgn(\gamma_0(X_i)-\gamma_0(X_j))(Y_i-Y_j)
\mid W_i
\right]
-
A_0.
\]
By Proposition \ref{Prop_AT_IOp},
\[
\sqrt n(\hat A-A_0)
=
\frac1{\sqrt n}
\sum_{i=1}^n
2h_A(W_i)
+
o_p(1).
\]
Also,
\[
\sqrt n(\hat p-p_0)
=
\frac1{\sqrt n}
\sum_{i=1}^n
(Y_i-p_0)
+
o_p(1).
\]
Define
\[
r(A,p)=\frac12+\frac{A}{4p(1-p)}.
\]
Since $\hat\theta_{AUC}=r(\hat A,\hat p)$ and
$\theta_0=r(A_0,p_0)$, the delta method gives
\begin{align*}
\sqrt n(\hat\theta_{AUC}-\theta_0)
&=
\frac{1}{4p_0(1-p_0)}
\sqrt n(\hat A-A_0)
\\
&\quad
-
\frac{A_0(1-2p_0)}{4p_0^2(1-p_0)^2}
\sqrt n(\hat p-p_0)
+
o_p(1).
\end{align*}
Substituting the two asymptotic linear representations,
\[
\sqrt n(\hat\theta_{AUC}-\theta_0)
=
\frac1{\sqrt n}
\sum_{i=1}^n
\left[
\frac{h_A(W_i)}{2p_0(1-p_0)}
-
\frac{A_0(1-2p_0)}{4p_0^2(1-p_0)^2}
(Y_i-p_0)
\right]
+
o_p(1).
\]
The result follows by the central limit theorem.
\end{proof}

\section{Proofs of Kernel-DML Estimators for Conditional Moment Restrictions}
\label{app_Kernel_DML_proofs}

\begin{proof}[Proof of Proposition \ref{prop:CMR_kernel_id}]
By iterated expectations and independence of $(W_i,W_j)$,
\[
Q(\gamma_0,\theta)
=
\mathbb{E}\!\left[m(Z_i,\gamma_0,\theta)m(Z_j,\gamma_0,\theta)K(Z_i,Z_j)\right].
\]
Define the finite signed measure $\mu_\theta$ on $(\mathcal{Z},\mathcal{B}(\mathcal{Z}))$ by
\[
\mu_\theta(A):=\int_A m(z,\gamma_0,\theta)\,dP_Z(z), \quad A\in\mathcal{B}(\mathcal{Z}).
\]
Since $m_\theta\in L_2(P_Z)$, $\mu_\theta$ is finite and absolutely continuous with respect to $P_Z$. Then
\[
Q(\gamma_0,\theta)
=
\iint K(z,z')\,d\mu_\theta(z)\,d\mu_\theta(z').
\]
Nonnegativity follows from positive definiteness of $K$. If $Q(\gamma_0,\theta)=0$, the ISPD property implies $\mu_\theta=0$, hence $m(Z,\gamma_0,\theta)=0$ almost surely. By \eqref{CMR}, $\theta=\theta_0$. The converse implication is immediate.
\end{proof}

\begin{proof}[Proof of Proposition \ref{prop_phi_KCMR}]
Let $K_{ij}=K(Z_i,Z_j)$. By Assumption \ref{ass_KCMR}, the chain rule and differentiation under the expectation give
\begin{align*}
\frac{d}{d\tau}
\mathbb E[g(W_i,W_j,\gamma_\tau,\theta)]
&=
\frac{d}{d\tau}
\mathbb E\!\left[
\varepsilon_\gamma(W_i,\gamma_0,\theta)
\varepsilon_\theta(W_j,\gamma_0,\theta)
K_{ij}\gamma_\tau(X_i)
\right]
\\
&\quad+
\frac{d}{d\tau}
\mathbb E\!\left[
\varepsilon_{\theta\gamma}(W_i,\gamma_0,\theta)
\varepsilon(W_j,\gamma_0,\theta)
K_{ij}\gamma_\tau(X_i)
\right]
\\
&\quad+
\frac{d}{d\tau}
\mathbb E\!\left[
\varepsilon_\theta(W_i,\gamma_0,\theta)
\varepsilon_\gamma(W_j,\gamma_0,\theta)
K_{ij}\gamma_\tau(X_j)
\right]
\\
&\quad+
\frac{d}{d\tau}
\mathbb E\!\left[
\varepsilon(W_i,\gamma_0,\theta)
\varepsilon_{\theta\gamma}(W_j,\gamma_0,\theta)
K_{ij}\gamma_\tau(X_j)
\right].
\end{align*}
By symmetry of $K_{ij}$ and exchangeability of $(W_i,W_j)$, the two terms involving
$\gamma_\tau(X_i)$ and the two terms involving $\gamma_\tau(X_j)$ are represented by the same function
\[
\alpha_0(x,\theta)
=
\mathbb E\!\left[
\Big(
\varepsilon_\gamma(W_i,\gamma_0,\theta)\varepsilon_\theta(W_j,\gamma_0,\theta)
+
\varepsilon_{\theta\gamma}(W_i,\gamma_0,\theta)\varepsilon(W_j,\gamma_0,\theta)
\Big)
K_{ij}
\mid X_i=x
\right].
\]
Hence,
\[
\frac{d}{d\tau}
\mathbb E[g(W_i,W_j,\gamma_\tau,\theta)]
=
\frac{d}{d\tau}
\mathbb E[\alpha_0(X_i,\theta)\gamma_\tau(X_i)]
+
\frac{d}{d\tau}
\mathbb E[\alpha_0(X_j,\theta)\gamma_\tau(X_j)].
\]
Applying Lemma \ref{qorth2} with $M=2$, $c_{11}=1$, $c_{21}=0$, $c_{12}=0$, $c_{22}=1$,
$\alpha_{01}(X_i,X_j)=\alpha_0(X_i)$, and
$\alpha_{02}(X_i,X_j)=\alpha_0(X_j)$ gives
\[
\phi(W_i,W_j,\gamma_0,\alpha_0(\theta),\theta)
=
\alpha_0(X_i,\theta)(Y_i-\gamma_0(X_i))
+
\alpha_0(X_j,\theta)(Y_j-\gamma_0(X_j)).
\]
\end{proof}

\subsection{Identification in Semiparametric Production Functions with Proxy Variables}
\label{app_ident_prod_fcns}

Let $\theta=(\beta',\rho)'\in\Theta:=\mathcal B\times\mathcal R\subset\mathbb R^{d_\beta+1}$ and define $m(Z_i,\gamma_0,\theta) = A_\beta(Z_i)-\rho B_\beta(X_i)$, where
\[
A_\beta(Z_i)
=
\mathbb E[Y_{i2}-F(K_{i2},L_{i2},\beta)\mid Z_i], \qquad B_\beta(X_i)
=
\gamma_0(X_i)-F(K_{i1},L_{i1},\beta).
\]
Since $X_i\subseteq Z_i$, $B_\beta(X_i)$ is $Z_i$-measurable. Moreover, $B_{\beta_0}(X_i)=\omega_0(X_i)$.

\begin{assumption}
\label{ass:prod_global_id}
(i) $\theta_0=(\beta_0',\rho_0)'$ satisfies $m(Z_i,\gamma_0,\theta_0)=0$ a.s. (ii) $\mathbb P(B_{\beta_0}(X_i)\neq0)>0$; (iii) for every $\beta\neq\beta_0$ and every $\rho\in\mathcal R$, $\mathbb P\!\left(
A_\beta(Z_i)\neq \rho B_\beta(X_i)
\right)>0$.
\end{assumption}

\begin{proposition}
\label{prop:prod_global_id}
Under Assumption \ref{ass:prod_global_id}, $m(Z_i,\gamma_0,\theta)=0$ a.s. $\iff
\theta=\theta_0$.
\end{proposition}

\begin{proof}[Proof of Proposition \ref{prop:prod_global_id}]
Suppose $m(Z_i,\gamma_0,\theta)=0$ a.s. Then $A_\beta(Z_i)=\rho B_\beta(X_i)$ a.s. By Assumption \ref{ass:prod_global_id} (iii), this implies
$\beta=\beta_0$. Hence, $(\rho-\rho_0)B_{\beta_0}(X_i)=0$ a.s. Since $\mathbb P(B_{\beta_0}(X_i)\neq0)>0$, Assumption \ref{ass:prod_global_id} (ii) yields
$\rho=\rho_0$.
\end{proof}

The kernel-DML estimator studied in the main text is based on the
first-order conditions of the kernel criterion rather than direct
minimization. Proposition \ref{prop:prod_global_id} establishes
identification of the underlying conditional moment restriction,
while Proposition \ref{prop:CMR_kernel_id} in the main text shows
that the kernel criterion preserves this identification under ISPD
kernels. Identification of the corresponding first-order conditions
is a separate issue. Local identification follows from
nonsingularity of the Hessian of the kernel criterion at
$\theta_0$, while global identification of the estimating equations
may be obtained under stronger curvature conditions. The calculations
below provide sufficient conditions for these properties in the
production-function setting.

\bigskip

\noindent
For identification from the first-order conditions, note that
\[
\varepsilon_\theta(W_i,\gamma,\theta)
=
\frac{\partial}{\partial\theta}
\varepsilon(W_i,\gamma,\theta)
=
\begin{pmatrix}
-\dfrac{\partial}{\partial\beta}
F(K_{i2},L_{i2},\beta)
+
\rho
\dfrac{\partial}{\partial\beta}
F(K_{i1},L_{i1},\beta)
\\[0.2cm]
-\gamma(X_i)+F(K_{i1},L_{i1},\beta)
\end{pmatrix},
\]
and
\[
\begin{aligned}
\varepsilon_{\theta\theta}(W_i,\gamma,\theta)
&=
\frac{\partial^2}{\partial\theta\,\partial\theta'}
\varepsilon(W_i,\gamma,\theta) \\
&=
\begin{pmatrix}
-\dfrac{\partial^2}{\partial\beta\,\partial\beta'}
F(K_{i2},L_{i2},\beta)
+
\rho
\dfrac{\partial^2}{\partial\beta\,\partial\beta'}
F(K_{i1},L_{i1},\beta)
&
\dfrac{\partial}{\partial\beta}
F(K_{i1},L_{i1},\beta)
\\[0.2cm]
\dfrac{\partial}{\partial\beta'}
F(K_{i1},L_{i1},\beta)
&
0
\end{pmatrix}.
\end{aligned}
\]
The Hessian of the kernel criterion is $H(\theta) = 2\{H_1(\theta)+H_2(\theta)\}$, where
\begin{align*}
H_1(\theta)
&=
\mathbb E\!\left[
\varepsilon_\theta(W_i,\gamma_0,\theta)
\varepsilon_\theta(W_j,\gamma_0,\theta)'
K(Z_i,Z_j)
\right],
\\
H_2(\theta)
&=
\frac12
\mathbb E\!\left[
\left(
\varepsilon_{\theta\theta}(W_i,\gamma_0,\theta)
\varepsilon(W_j,\gamma_0,\theta)
+
\varepsilon(W_i,\gamma_0,\theta)
\varepsilon_{\theta\theta}(W_j,\gamma_0,\theta)
\right)
K(Z_i,Z_j)
\right].
\end{align*}
Since $m(Z_i,\gamma_0,\theta_0)=0$ a.s., we have $H(\theta_0)=2H_1(\theta_0)$,
so non-singularity of $H_1(\theta_0)$ suffices for local identification. A sufficient condition for global identification from the FOCs is that $H(\theta)$ is positive definite for all $\theta\in\Theta$. For the linear specification $F(Q_{it},\beta)=\beta'Q_{it}$, we obtain
\begin{align*}
H_1(\theta)
&=
\mathbb E\!\left[
\begin{pmatrix}
(Q_{i2}-\rho Q_{i1})(Q_{j2}-\rho Q_{j1})'
&
(Q_{i2}-\rho Q_{i1})(\gamma_0(X_j)-\beta'Q_{j1})
\\
(\gamma_0(X_i)-\beta'Q_{i1})(Q_{j2}-\rho Q_{j1})'
&
(\gamma_0(X_i)-\beta'Q_{i1})
(\gamma_0(X_j)-\beta'Q_{j1})
\end{pmatrix}
K(Z_i,Z_j)
\right],
\\
H_2(\theta)
&=
\begin{pmatrix}
0
&
\dfrac12
\mathbb E\!\left[
(Q_{i1}\varepsilon_j(\theta)
+
Q_{j1}\varepsilon_i(\theta))
K(Z_i,Z_j)
\right]
\\[0.3cm]
\dfrac12
\mathbb E\!\left[
(Q_{i1}'\varepsilon_j(\theta)
+
Q_{j1}'\varepsilon_i(\theta))
K(Z_i,Z_j)
\right]
&
0
\end{pmatrix},
\end{align*}
where $\varepsilon_i(\theta) = \varepsilon(W_i,\gamma_0,\theta)$.
\begin{proposition}
\label{prop:prod_linear_id}
Suppose Assumption \ref{ass:prod_global_id}(i)-(ii) hold and
\[
F(Q_{it},\beta)=\beta'Q_{it}.
\]
If
\[
\mathbb P\!\left(
a'(Q_{i2}-\rho_0Q_{i1})
+b\,\omega_0(X_i)=0
\right)<1
\]
for every $(a',b)'\neq0$,
then
\[
m(Z_i,\gamma_0,\theta)=0 \text{ a.s.}
\iff
\theta=\theta_0.
\]
Consequently, under the conditions of Proposition
\ref{prop:CMR_kernel_id},
\[
\mathbb E[g(W_i,W_j,\gamma_0,\theta)]=0
\iff
\theta=\theta_0.
\]
\end{proposition}

\begin{proof}
For the linear specification,
\[
m(Z_i,\gamma_0,\theta)
=
(\beta_0-\beta)'Q_{i2}
+\rho\,\beta'Q_{i1}
-\rho_0\beta_0'Q_{i1}
-(\rho-\rho_0)\omega_0(X_i).
\]
Hence
\[
m(Z_i,\gamma_0,\theta)=0
\]
a.s. implies
\[
(\beta_0-\beta)'Q_{i2}
+
(\rho\beta-\rho_0\beta_0)'Q_{i1}
-
(\rho-\rho_0)\omega_0(X_i)
=
0
\]
a.s.
The stated rank condition therefore implies
\[
\beta=\beta_0,
\qquad
\rho=\rho_0.
\]
The converse is immediate.

The second claim follows from Proposition
\ref{prop:CMR_kernel_id}, which establishes that the kernel criterion
preserves the identifying content of the conditional moment restriction.
\end{proof}

The following example shows how identification from finitely many instruments may fail.

\begin{example}[Weak identification from finite instruments]
Suppose
\[
Y_{it}=K_{it}+\omega_{it}+\varepsilon_{it},
\qquad
I_{i1}=K_{i1}e^{\omega_{i1}},
\]
so that $\omega_0(X_i)=\log(I_{i1}/K_{i1})$. Then $m(X_i,\rho)
=
(\rho_0-\rho)\omega_0(X_i)$, and finite-dimensional instruments $h(X_i)$ identify $\rho_0$ only through
\[
\mathbb E[m(X_i,\rho)h(X_i)]
=
(\rho_0-\rho)\mathbb E[\omega_0(X_i)h(X_i)].
\]
Hence, identification may become weak when $\omega_0(X_i)$ is weakly correlated with the chosen instruments. Figure \ref{fig:prod_fcn_R2_F} reports Monte Carlo evidence for $K_{i1}=e^{V_i}$, $V_i\sim N(0,\sigma_V^2)$ and $\omega_{i1}\sim N(0,1)$; using polynomial instruments in $(K_{i1},I_{i1})$. As $\sigma_V$ increases, first-stage $R^2$ and F-statistics deteriorate sharply.

\begin{figure}[H]
\centering
\includegraphics[width=\textwidth]{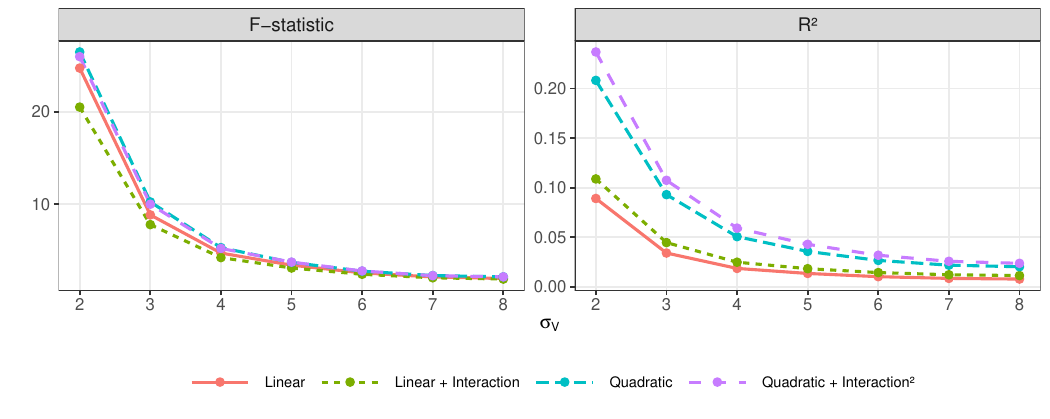}
\caption{Average first-stage $R^2$ and F-statistic across 500 simulations with $n=500$.}
\label{fig:prod_fcn_R2_F}
\end{figure}

$\square$
\end{example}

 \subsection{Proofs for Semiparametric Production Functions with a Proxy Variable}

\begin{proof}[Proof of Proposition \ref{prop:prod_phi}]
Let
\[
F_{it}(\beta)=F(K_{it},L_{it},\beta),
\qquad
F_{\beta,it}(\beta)
=
\frac{\partial}{\partial\beta}
F(K_{it},L_{it},\beta).
\]
For
\[
\varepsilon(W_i,\gamma_0,\theta)
=
Y_{i2}-F_{i2}(\beta)-\rho\{\gamma_0(X_i)-F_{i1}(\beta)\},
\]
we have
\[
\varepsilon_\theta(W_i,\gamma,\theta)
=
\begin{pmatrix}
-F_{\beta,i2}(\beta)+\rho F_{\beta,i1}(\beta)
\\
-\gamma(X_i)+F_{i1}(\beta)
\end{pmatrix},
\]
together with
\[
\varepsilon_\gamma(W_i,\gamma,\theta)
=
-\rho,
\qquad
\varepsilon_{\theta\gamma}(W_i,\gamma,\theta)
=
\begin{pmatrix}
0\\
-1
\end{pmatrix}.
\]
Hence, by Proposition \ref{prop_phi_KCMR},
\begin{align*}
\alpha_0(x)
&=
\mathbb E\!\left[
\left\{
\varepsilon_\gamma(W_i,\gamma_0,\theta_0)
\varepsilon_\theta(W_j,\gamma_0,\theta_0)
+
\varepsilon_{\theta\gamma}(W_i,\gamma_0,\theta_0)
\varepsilon(W_j,\gamma_0,\theta_0)
\right\}
K(x,X_j)
\right]
\\
&=
\begin{pmatrix}
\rho_0
\mathbb E\!\left[
\{F_{\beta,j2}(\beta_0)-\rho_0F_{\beta,j1}(\beta_0)\}
K(x,X_j)
\right]
\\[0.3cm]
\rho_0
\mathbb E\!\left[
\{\gamma_0(X_j)-F_{j1}(\beta_0)\}
K(x,X_j)
\right]
-
\mathbb E\!\left[
\varepsilon(W_j,\gamma_0,\theta_0)K(x,X_j)
\right]
\end{pmatrix}.
\end{align*}
Since
\[
\mathbb E[\varepsilon(W_j,\gamma_0,\theta_0)\mid X_j]=0,
\]
iterated expectations imply
\[
\mathbb E[
\varepsilon(W_j,\gamma_0,\theta_0)K(x,X_j)
]
=0.
\]
Moreover,
\[
\gamma_0(X_j)-F_{j1}(\beta_0)
=
\mathbb E[Y_{j1}-F_{j1}(\beta_0)\mid X_j].
\]
Therefore,
\[
\alpha_0(X_i)
=
\left(
(\rho_0\eta_{02}(X_i)-\rho_0^2\eta_{01}(X_i))',
\,
\rho_0\eta_{03}(X_i)
\right)'.
\]
\end{proof}

\begin{proof}[Proof of Proposition \ref{prop_AN_prod_fcns}]
For Theorem \ref{Thm_AT_Consist} (i) we have that by bounded $K$, compact $\Theta$ and finite second moments of $Y_{i2}$, $Q_{i1}$, $Q_{i2}$ and $\gamma_0(X_i)$ (Assumptions \ref{ass_asympt_prod_fcns} (i)-(iii))
\[
\iint |g(w_i,w_j,\hat \gamma_l,\theta) - g(w_i,w_j,\gamma_0,\theta)| F_0(dw_i) F_0(dw_j) \leq C \|\hat \gamma_l - \gamma_0 \| + C\|\hat \gamma_l - \gamma_0 \|^2 = o_p(1),
\]
where convergence follows from Assumption \ref{ass_asympt_prod_fcns} (iv). Also, again by finite second moments and $\Theta$ compact
\[
\|\hat{\alpha}_l - \alpha_0\|_\mathcal{A}^2 \leq C\left(\|\hat \eta_{1l} - \eta_{01}\|^2 + \|\hat \eta_{2l} - \eta_{02}\|^2 + \|\hat \eta_{3l} - \eta_{03}\|^2 \right) \to_p 0,
\]
with convergence by the fact that $\hat \eta$'s are just averages of observed quantities. By $\Theta$ compact, Theorem \ref{Thm_AT_Consist} (iv) holds with $a = 1$ and 
\begin{align*}
d_g(W_i,W_j,\gamma)
&=
C |K(X_i,X_j)| M_i(\gamma) M_j(\gamma),
\\
M_i(\gamma)
&=
1 + |Y_{i2}| + |Q_{i2}| + |Q_{i1}| + |\gamma(X_i)|,
\\[0.2cm]
d_\phi(W_i,W_j,\gamma,\alpha)
&=
C\Big[
|Y_{i1}-\gamma(X_i)|A(X_i)
+
|Y_{j1}-\gamma(X_j)|A(X_j)
\Big],
\\
A(X_i)
&=
|\eta_{1}(X_i)|
+
|\eta_{2}(X_i)|
+
|\eta_{3}(X_i)|.
\end{align*}
which have finite expectation by $K$ bounded and existence of first moments. Hence, $\hat \theta \to_p \theta_0$. For Assumption \ref{Ass_AT_equiv} we have finite second moment of $\psi$ by finite second moments of all variables and
\begin{align*}
&\iint |g(w_i,w_j,\hat \gamma_l,\theta_0) - g(w_i,w_j,\gamma_0,\theta_0)|^2 F_0(dw_i) F_0(dw_j) \leq C \|\hat \gamma_l - \gamma_0 \|^2 + C\|\hat \gamma_l - \gamma_0 \|^4 = o_p(1), \\
&\iint |\phi(w_i,w_j,\hat \gamma_l,\alpha_0,\theta_0) - \phi(w_i,w_j,\gamma_0,\alpha_0,\theta_0)|^2 F_0(dw_i) F_0(dw_j) \leq C \|\hat \gamma_l - \gamma_0 \|^2 = o_p(1), \\
&\iint |\phi(w_i,w_j, \gamma_0,\hat \alpha(\theta_0),\theta_0) - \phi(w_i,w_j,\gamma_0,\alpha_0,\theta_0)|^2 F_0(dw_i) F_0(dw_j) \leq C \|\hat \alpha_l (\theta_0)- \alpha_0(\theta_0) \|^2 = o_p(1),
\end{align*}
by Assumption \ref{ass_asympt_prod_fcns} and convergence of sample averages. For Assumption \ref{Ass_AT_interaction}
\[
\sqrt{n} \iint |\hat \xi_l(w_i,w_j)| F_0(dw_i) F_0(dw_j) \leq \sqrt{n} \|\hat{\alpha}_l(\theta_0) - \alpha_0\| \cdot \|\hat \gamma_l - \gamma_0\| = \sqrt{n}O_p(n^{-1/2})o_p(n^{-1/4}) =  o_p(1),
\]
by Assumption \ref{ass_asympt_prod_fcns} (iv) and sample averages being $O_p(n^{-1/2})$. For Assumption \ref{Ass_AT_GRsmallbias} by $\Theta$ compact, $K$ bounded and Assumption \ref{ass_asympt_prod_fcns} (iv)
\[
\sqrt{n}| \bar \psi(\hat{\gamma}_l,\alpha_0, \theta_0) \leq C \sqrt{n} \|\hat \gamma_l - \gamma_0\|^2 = o_p(1).
\]
For Assumption \ref{Ass_AT_Jacobianhat}, from linearity everything is differentiable, the Lipschitz conditions hold with same $d_g$ and $d_\phi$ as before up to the constant $C$ and with $a = 1$. Since same bounds apply, the converge of the derivatives follow in the same way as we did to check Assumptions \ref{Ass_AT_equiv} and \ref{Ass_AT_interaction}. Hence, $\sqrt{n}(\hat \theta - \theta_0) \to_d \mathcal{N}(0,V)$.
\end{proof}

\section{Asymptotics for Alternative Implementations}
\label{Alter_implem}
\begin{corollary}
\label{corollary_preliminary_estimator}
Let $\tilde\theta_l\to_p\theta_0$ be a preliminary cross-fitted consistent estimator and define
\[
\begin{aligned}
\hat\psi^p(\theta)
&=\binom n2^{-1}\sum_{l=1}^L\sum_{(i,j)\in I_l}
g(W_i,W_j,\hat\gamma_l,\theta)+\hat\phi^p,\\
\hat\phi^p
&=
\binom n2^{-1}\sum_{l=1}^L\sum_{(i,j)\in I_l}
\phi(W_i,W_j,\hat\gamma_l,\hat\alpha_l(\tilde\theta_l),\tilde\theta_l),\\
\hat\xi_l^p(w_i,w_j)
&=
\phi_{ij}(\hat\gamma_l,\hat\alpha_l(\tilde\theta_l),\tilde\theta_l)
-\phi_{ij}(\gamma_0,\hat\alpha_l(\tilde\theta_l),\tilde\theta_l)-\phi_{ij}(\hat\gamma_l,\alpha_0,\theta_0)
+\phi_{ij}(\gamma_0,\alpha_0,\theta_0).
\end{aligned}
\]
Then the conclusions of Theorems 2 and 3 and
Proposition 2 continue to hold after replacing $\hat\xi_l$ by $\hat\xi_l^p$, replacing $\phi(w_i,w_j,\gamma_0,\hat\alpha_l(\theta_0),\theta_0)$ by $\phi(w_i,w_j,\gamma_0,\hat\alpha_l(\tilde\theta_l),\tilde\theta_l)$, and imposing all Lipschitz conditions stated above only on $g$. For consistency, the
uniform convergence conditions involving $\phi$ are replaced by $\hat\phi^p=o_p(1)$.
\end{corollary}

    \begin{proof}[Proof of Corollary \ref{corollary_preliminary_estimator}]
    Since $\hat\phi^p$ does not depend on $\theta$, the consistency proof is the proof of
    Theorem 2 with all uniform convergence and Lipschitz arguments involving $\phi$ replaced
    by the single requirement $\hat\phi^p=o_p(1)$. The expansion in Lemma 4 is unchanged
    after replacing $\hat\xi_l$ by $\hat\xi_l^p$ and replacing
    $\phi(\cdot,\gamma_0,\hat\alpha_l(\theta_0),\theta_0)$ by
    $\phi(\cdot,\gamma_0,\hat\alpha_l(\tilde\theta_l),\tilde\theta_l)$, because the same
    conditional Markov and product-rate arguments apply term by term. Since the correction
    term is fixed in $\theta$, the Jacobian of $\hat\psi^p(\theta)$ is simply the Jacobian of
    $U_n g(\cdot,\hat\gamma,\theta)$, whose consistency follows from the modified Jacobian
    condition. The conclusion of Theorem 3 then follows from the same mean-value expansion
    and U-statistic CLT. Finally, under the corresponding version of Assumption 7, the proof of
    Proposition 2 is unchanged, because the leave-one-out consistency conditions are imposed
    directly and the Jacobian estimator only involves $g$.
    \end{proof}

\section{Computations for the Gaussian First-Step Error Model}
\label{app_computations_GFSM}
In this section, we provide some calculations for the Gaussian First-Step
Model that are used in the main text. In this DGP, $X_i \sim \mathcal{N}(0,\sigma_X^2)$ and
\[
Y_i = \gamma_0(X_i) + \varepsilon_i, \quad \gamma_0(x) = 5 + x^2,
\]
where $\varepsilon_i \sim \mathcal{N}(0, \sigma_\varepsilon^2)$, $\sigma_\varepsilon^2 = \mathbb{V}ar(\gamma_0(X_i))/StN$, and $StN$ is the Signal to Noise Ratio which is fixed by us. This DGP is called a Gaussian First-Step Error Model because we generate
\[
\hat{\gamma}(X_i)
=
\gamma_0(X_i)+b_n(X_i)+\sigma_nU_i,
\qquad
U_i\sim N(0,1),
\]
where
\[
b_n(x)
=
c_1\bigl(\mathbb E[\gamma_0(X_i)]-\gamma_0(x)\bigr)n^{-\rho_1},
\qquad
\sigma_n^2
=
\frac{c_2}{STN}n^{-\rho_2},
\]
and $c_1,c_2,\rho_1,\rho_2$ are constants. The variables $(X_i,\varepsilon_i,U_i)$ are independent. We first show that

\begin{equation*}
\sqrt{n}\frac{\partial \theta (\gamma _{0})}{\partial \gamma }\left[ \hat{%
\gamma}-\gamma _{0}\right] \sim \mathcal{N}\left( C_{1}n^{0.5-\rho
_{1}},\sigma _{dn}^{2}\right) ,
\end{equation*}%
where $C_{1}=\theta _{0}c_{1}$, $\sigma _{dn}^{2}=C_{2}n^{1-\rho _{2}}$ and $%
C_{2}=c_{2}(\mathbb{E}[2Y_i])^{-2}\theta _{0}^{2}/StN$. Therefore, the bias
coming from this derivative term diverges when $1/4<\rho _{1}<1/2$. This
illustrates the bias problem and the lack of root-$n$ consistency of plug-in
estimators based on non-locally robust moment functions. From our derivative calculations in the proof of Proposition 1, 
\begin{equation*}
\frac{\partial \theta (\gamma _{0})}{\partial \gamma }\left[ h\right]
=-\left( \mathbb{E}[2Y_i]\right) ^{-1}\left\{ \theta _{0}\mathbb{E}%
[h(X_{i})+h(X_{j})]-\mathbb{E}[sgn\left( X_{i}^{2}-X_{j}^{2}\right)
(h(X_{i})-h(X_{j}))]\right\} .
\end{equation*}%
Hence, $\frac{\partial \theta (\gamma _{0})}{\partial \gamma }\left[ \hat{%
\gamma}-\gamma _{0}\right] $ follows a normal distribution with mean 
\begin{equation*}
\left( \mathbb{E}[2Y_i]\right) ^{-1}\mathbb{E}[\left\vert
X_{i}^{2}-X_{j}^{2}\right\vert ]c_{1}n^{-\rho _{1}}=\theta _{0}c_{1}n^{-\rho
_{1}},
\end{equation*}%
and variance $\left( \mathbb{E}[2Y]\right) ^{-2}\sigma _{n}^{2}\theta _{0}^{2}$, where we have used that $\mathbb{E}[sgn\left( X_{i}^{2}-X_{j}^{2}\right) ]=0$. Next, we verify the conditions for the asymptotic theory (i.e. Assumption
7). Straightforward calculations show that the random variable $\Delta
_{0}=\gamma _{0}(X_{i})-\gamma _{0}(X_{j})=X_{i}^{2}-X_{j}^{2}$ is
absolutely continuous with a density $f_{0}(y)=(1/2\sigma _{X}^{2})\exp(-\frac{\left\vert y\right\vert }{%
2\sigma _{X}^{2}})K_{0}\left( (\left\vert y\right\vert/2\sigma
_{X}^{2})\right)$,
where $K_{0}$ is the modified Bessel function of the second kind of order 0.
It is known that as $z\downarrow 0$ $K_{0}\left( z\right) \sim -\ln z-\xi ,$
where $\xi $ is Euler's constant, hence Assumption 7(ii) and (iii) holds
with any $\beta <1.$ This means that we require $||\hat{\gamma}-\gamma
_{0}||=o_{p}(n^{-\rho _{2\beta }}),$ with  $\rho _{2\beta }\geq \frac{2+\beta }{4(1+\beta )}.$
This holds for any $\beta <1,$ so it must be that $\rho _{2\beta } > 3/8$ and $||\hat{\gamma}_l - \gamma_0||^2 = o_p(n^{-3/4})$. Note 
\begin{align*}
    ||\hat{\gamma}_l - \gamma_0||_2^2 &= E[|\mu_n(X_i) + \sigma^2_{n}U_i|^2 | U_i]], \\
    P\left(E[n^{3/4}|\mu_n(X_i) + \sigma^2_{n}U_i|^2 | U_i]  > \varepsilon\right) &\leq n^{3/4} \mathbb{E}[b_n(X_i)^2]/\varepsilon \\
    &\leq  n^{3/4}C(n^{-\rho_2} + n^{-2\rho_1})/\varepsilon,
\end{align*}
so we need $\min(2\rho_1, \rho_2) > 3/4$ which implies $\rho _{1}>3/8=0.375$
and $\rho _{2}>3/4$. Note that these are sufficient conditions. In the
simulations, we consider values below and above the required to
evaluate the robustness of the finite sample performance to our sufficient
conditions.

\section{Empirical Appendix}
\label{app_emp}

This appendix reports additional empirical results for the Inequality of Opportunity (IOp) application. The main text focuses on Conditional Inference Forest (CIF), which is among the most widely used machine-learning methods in the IOp literature. Here we compare results across all ML methods considered in the paper and further investigate the sources of the debiasing correction.

Figure \ref{fig:iop_all} reports country-level relative IOp estimates obtained using all ML methods considered in the analysis. For each method, we report both the plug-in and debiased estimators together with the corresponding confidence intervals. Consistent with the findings in the main text, plug-in estimates exhibit substantial sensitivity to the choice of learner, whereas debiased estimates are considerably more stable across ML methods.

Figure \ref{fig:iop_decomp_all} presents the decomposition of the debiasing correction across ML methods. Recall that the correction consists of two components: an orthogonalization term that removes the first-order effect of estimating the nuisance function and an overfitting correction induced by cross-fitting. The figure illustrates the relative importance of these components across alternative ML learners.

Finally, Figure \ref{fig:iop_bias_n} examines the behavior of the plug-in bias as a function of sample size. The results show that the magnitude of the bias can remain substantial even in large samples and varies considerably across ML methods, reflecting differences in regularization and model-selection behavior. In contrast, the debiased estimator exhibits much greater stability across learners.

\begin{figure}[h!]
    \centering
    \includegraphics[width=\textwidth,height=10cm]{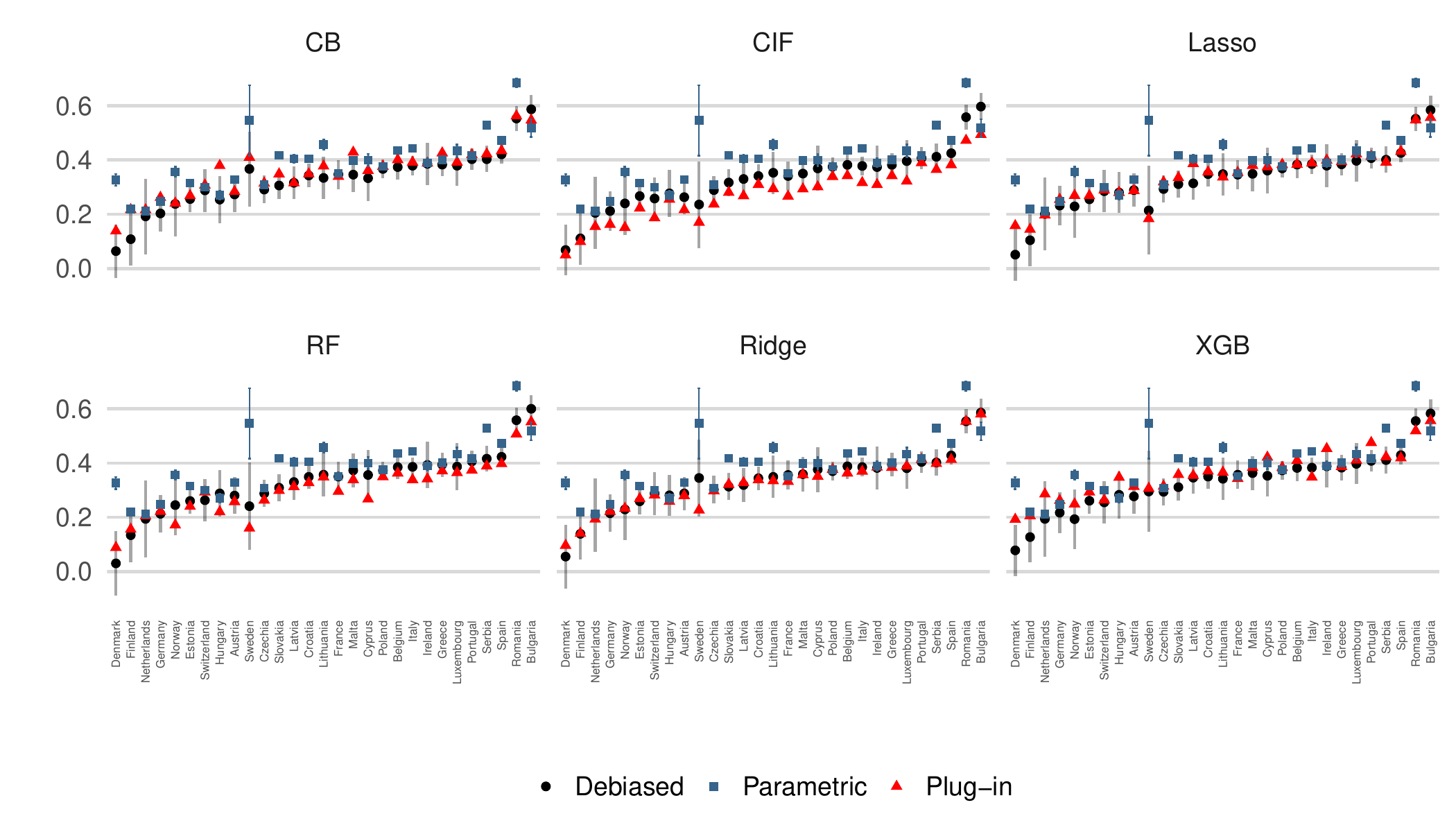}
    \caption{Relative IOp estimates across ML methods.}
    \label{fig:iop_all}
\end{figure}

\begin{figure}[h!]
    \centering
    \includegraphics[width=\textwidth,height=10cm]{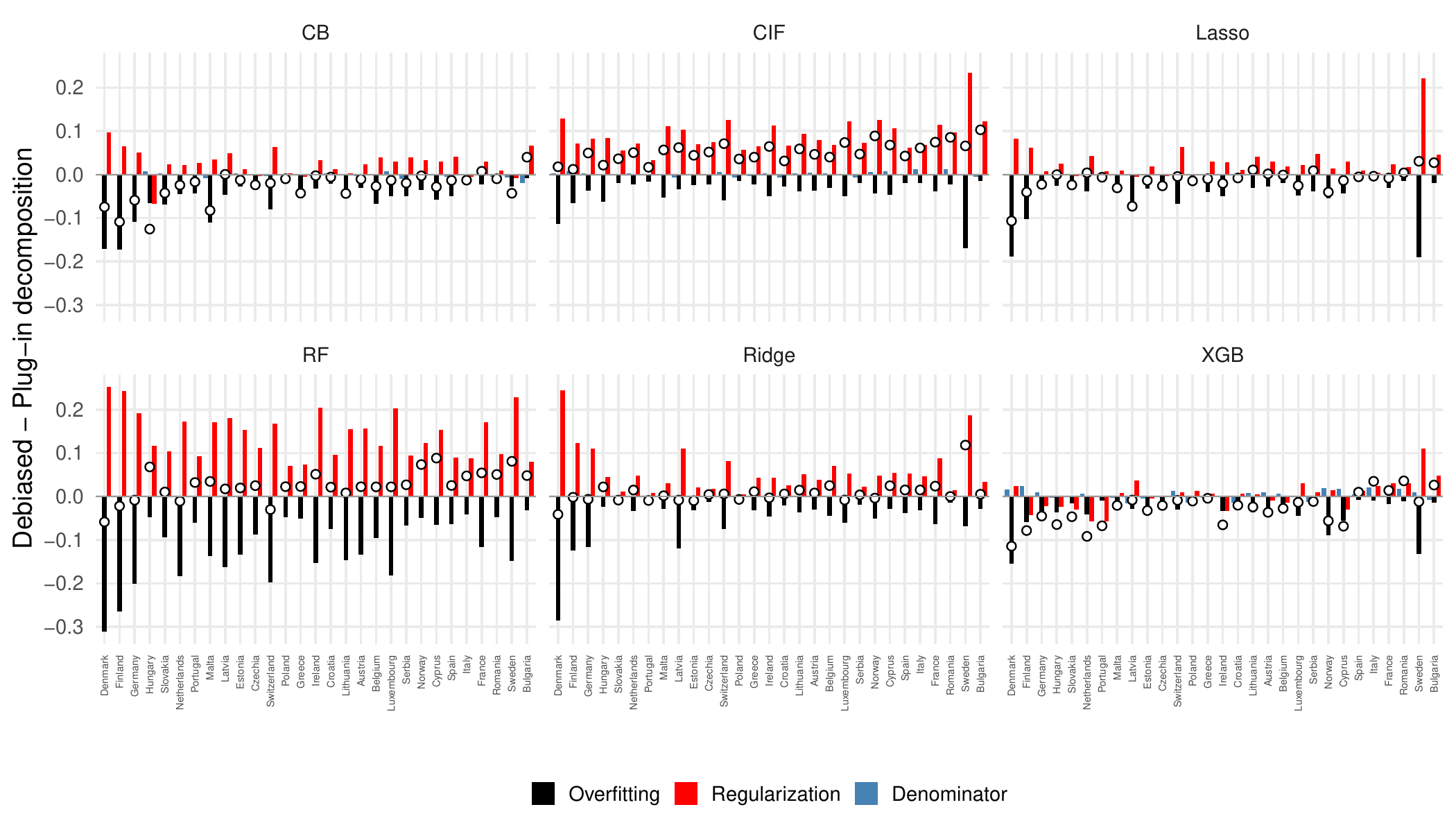}
    \caption{Decomposition of the debiasing correction across ML methods.}
    \label{fig:iop_decomp_all}
\end{figure}

\begin{figure}[h!]
    \centering
    \includegraphics[width=\textwidth]{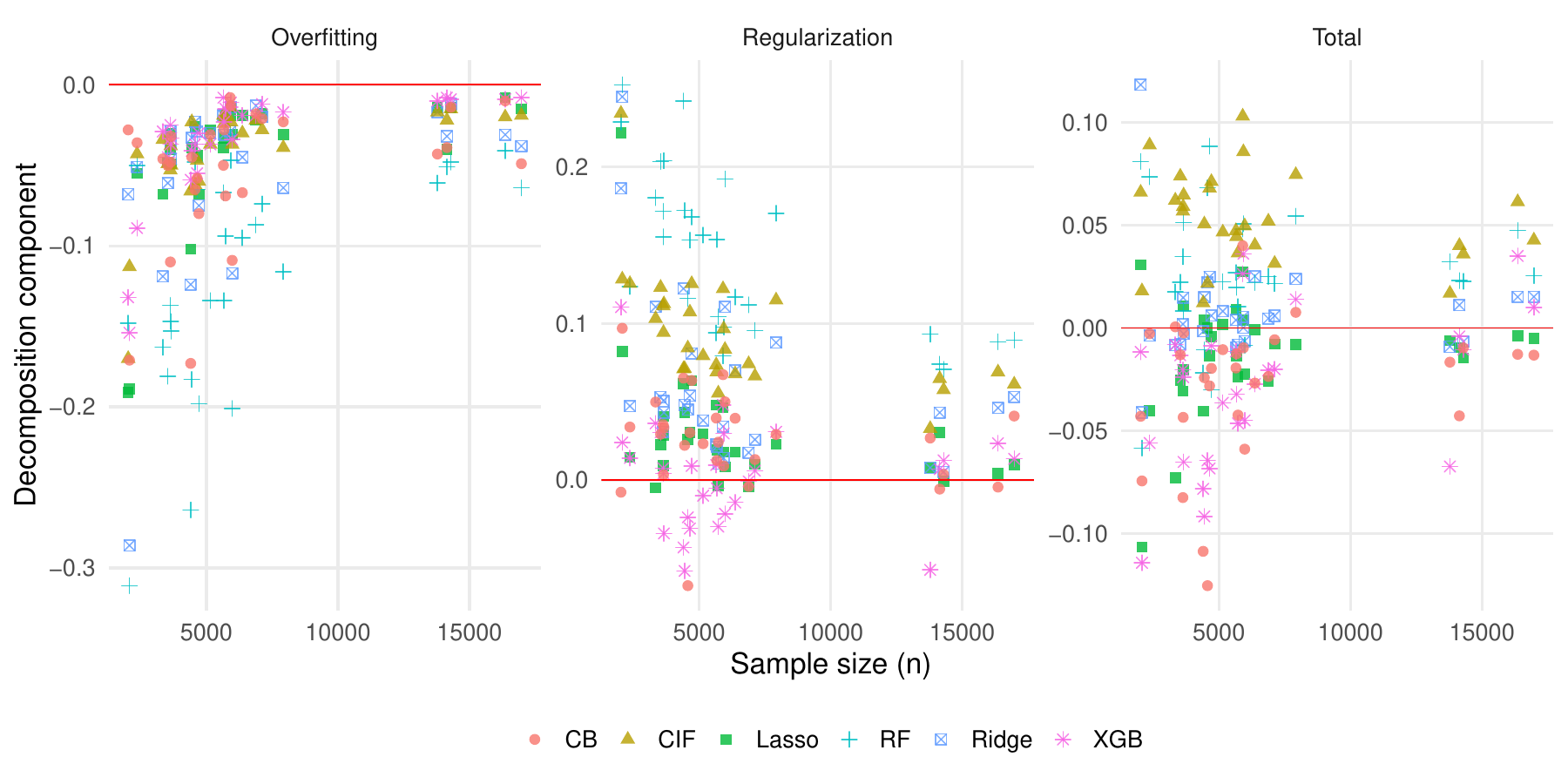}
    \caption{Bias of the plug-in estimator as a function of sample size across ML methods.}
    \label{fig:iop_bias_n}
\end{figure}

\section{Further Applications}
\label{app_FurtherApp}

\subsection{Gini with missing data}

Suppose we have income $Y_i$ and other observables $X_i$. We observe $X_i$ for $i = 1,...,n$ but we only observe $Y_i$ for $i \in A \subset \{1,...,n\}$. For all $i = 1,...,n$, we have a binary variable $O_i$ which is $1$ if $Y_i$ is observed and $0$ otherwise. We let $W_i = (Y_i, O_i, X_i')'$. We are interested in the Gini of $Y_i$
\[
\theta_0 = \frac{\mathbb{E}[|Y_i - Y_j|]}{\mathbb{E}[Y_i + Y_j]}.
\]
If income observations are not randomly missing, then estimating the Gini only with the observed incomes is not justified. What we can identify is $\mathbb{E}[|Y_i - Y_j| \, | O_i = 1, O_j = 1]$ and $\mathbb{E}[Y_i | O_i = 1]$ not $\theta_0$. To circumvent this, we assume that income is Missing at Random. For identification, we also assume a typical common support assumption. We let $p_0(x) = P(O_i = 1 | X_i = x)$.
\begin{assumption}[MAR and Common support]
    \begin{align*}
        &(i) \, Y_i \perp O_i \mid X_i, \\
        &(ii) \, \varepsilon \leq p_0(X_i) \leq 1- \varepsilon \text{ for some } \varepsilon > 0.
    \end{align*}
\end{assumption}
MAR can be relaxed to just mean independence. Under MAR the Gini is identified as
\[
\theta_0 = \frac{\theta_{01}}{\theta_{02}}, \text{ where } \theta_{01} = \mathbb{E}\biggl[|Y_i - Y_j| \frac{O_i O_j}{p_0(X_i)p_0(X_j)}\biggr] \text{ and } \theta_{02} = 2\mathbb{E}\biggl[Y_i \frac{O_i}{p_0(X_i)}\biggr],
\]
and $p_0(x) = P(O_i = 1 | X_i = x)$. Let $\delta_0(X_i,X_j) = \mathbb{E}[|Y_i - Y_j| \, | X_i,X_j]$, note that it is identified since under MAR: $\mathbb{E}[|Y_i - Y_j| \, | X_i,X_j, O_i = 1, O_j = 1] = \delta_0(X_i,X_j)$. As before $\gamma_0(X_i) = \mathbb{E}[Y_i | X_i]$. $\theta_0$ is a ratio of semiparametric U-statistics so it falls into our framework. The next result provides an expression for $\theta_0$ as a ratio of locally robust moments.\footnote{Proofs available upon request.} 
\begin{proposition}
    Under MAR and Common Support, we have the following locally robust functions for the numerator and denominator
    \begin{align*}
        \psi_1(W_i,\delta_0,\alpha_{01},\theta_1) &= \delta_0(X_i,X_j) + \alpha_{01}(X_i,X_j,O_i,O_j)(|Y_i - Y_j| - \delta_0(X_i,X_j)) - \theta_1, \\
        \psi_2(W_i,\gamma_0,\alpha_{02},\theta_2) &= \gamma_0(X_i) + \alpha_{02}(X_i,O_i)(Y_i - \gamma_0(X_i)) - \theta_2,
    \end{align*}
    where $\alpha_{01}(X_i,X_j,O_i,O_j) = O_iO_j/(p_0(X_i)p_0(X_j))$ and $\alpha_{02}(X_i,O_i) = O_i/p_0(X_i)$.
\end{proposition}
As in the case of the Augmented Inverse Propensity weighting (AIPW) we recover a double robustness property that makes us robust to misspecification of either nuisance parameter.
\begin{proposition}[Double Robustness]
\label{prop_MD_DR}
\begin{align*}
    \psi_1(W_i, \delta, \alpha_1,\theta_1) &= \theta_{01} - \theta_1 + \mathbb{E}[(\delta(X_i,X_j) - \delta_0(X_i,X_j))(\alpha_{01}(X_i,X_j) - \alpha_1(X_i,X_j))], \\
    \psi_2(W_i, \gamma, \alpha_2,\theta_2) &= \theta_{02} - \theta_2 + \mathbb{E}[(\gamma(X_i) - \gamma_0(X_i))(\alpha_{02}(X_i) - \alpha_2(X_i))].
\end{align*}
\end{proposition}
Hence, a locally robust estimator is given by 
\[
\hat{\theta} = \frac{\binom{n}{2}^{-1} \sum_{l=1}^L \sum_{(i,j) \in I_l}\biggl[\hat{\delta}_l(X_i,X_j) + \frac{O_i O_j}{\hat{p}_l(X_i)\hat{p}_l(X_j)}(|Y_i - Y_j| - \hat{\delta}_l(X_i,X_j))\biggr]}{ \frac{2}{n} \sum_{k  =1}^K \sum_{i \in C_k} \biggl[\hat{\gamma}_k(X_i) + \frac{O_i}{\hat{p}_k(X_i)}(Y_i - \hat{\gamma}_k(X_i)) \biggr]},
\]
where for the numerator a traditional cross-fitting with the sample split in $C_1,...,C_K$ is used. As with the usual AIPW approach for missing data and causal inference, it is easy to check (particularly using Proposition \ref{prop_MD_DR}) that all the general asymptotic assumptions hold if we have mean-square consistent estimators $\hat{\delta}$, $\hat{p}$ and $\hat{\gamma}$ that converge at at least $n^{-1/4}$.

\begin{remark}
    Although we have focused on the Gini, we can apply this procedure to any U-statistic by replacing $|Y_i-Y_j|$ by some known function $g(W_i,W_j)$. For example, if we let $g(W_i,W_j) = sgn(Y_i - Y_j)sgn(X_i - X_j)$, where $Y_i$ is income and $X_i$ parental income, we get a popular measure of Intergenerational Mobility measured by the Kendall-$\tau$.
\end{remark}

\begin{remark}
    In the next section we document how our general theory applies to causal U-statistics. In particular, how to estimate the counterfactual coefficient, say Gini or Kendall-$\tau$ coefficient, under a given treatment allocation. For example one could answer: What would be the Gini coefficient if a treatment is given only to females with no university degree? The estimators suggested for general causal U-statistics are also doubly robust.
\end{remark}

\subsection{Causal U-statistics}

\label{app_causal_ustats}
Our theory generalizes causal estimands proposed in \cite{mao2018causal} and \cite{yin2024highly}. Let $W_{i} = (Y_{i}, D_{i}, X_{i})$ and $Z_{i} = (Y_{i}(1), Y_{i}(0))$ where
$Y_{i} = D_{i} Y_{i}(1) + (1-D_{i})Y_{i}(0)$. Consider a family of causal estimands where $g$ is a known function and $\gamma^{(1)}$ and $\gamma^{(0)}$ are potential nuisance parameters (e.g. $\mathbb{E}[Y_i(1)| X_i]$ and $\mathbb{E}[Y_i(0)| X_i]$)
\[
\theta_0(\pi) = \mathbb{E}\left[ \sum_{(a,b) \in \{0,1\}^2} g(Y_i(a),X_i,Y_j(b),X_j,\gamma^{(a)},\gamma^{(b)})\pi_{ab}(X_i,X_j)\right] \text{ with } \pi_{ab}: \mathcal{X} \times \mathcal{X} \to \{0,1\}.
\]
For a given $\pi_{ab}$ we interpret $\theta(\pi)$ as a causal estimand. A trivial example is $g(y_i(1),y_i(0)) = y_i(1) - y_i(0)$ which gives the Average Treatment Effect (ATE). Our theory is better motivated for nonlinear $g$ where typical methods cannot be used, for example, the Wilcoxon statistic is based on $g(y_i(a),y_j(b)) = 1(y_i(a) \geq  y_j(b))$. This is important when most treated fare worse but the ATE is positive due to a few treated that fare much better.

Interesting causal estimands arise depending on the choice of $\pi_{ab}$. We would recommend to let $\pi: \mathcal{X} \to \{0,1\}$ and $\pi_{ab}(X_i,X_j) = 1(\pi(X_i) = a, \pi(X_j) = b)$. For instance, suppose that $\pi(X_i) = 1$ only for female high school dropouts and zero otherwise. Then we can estimate the counterfactual value whenever only female high school dropouts are treated. If we are interested in the Gini Mean Difference (the Gini coefficient follows directly), we have
\[
\theta_0(\pi) = \mathbb{E}\left[ \sum_{(a,b) \in \{0,1\}^2} |Y_i(a) - Y_j(b)|1(\pi(X_i) = a, \pi(X_j) = b)\right],
\]
hence we can estimate the counterfactual inequality whenever certain groups are treated. Likewise, for intergenerational mobility measured by the Kendall-$\tau$ coefficient we have
\[
    g(y_i(a),x_i,y_j(b),x_j,\gamma^{(a)},\gamma^{(b)}) = sgn(Y_i(a) - Y_j(b))sgn(X_{1i} - X_{1j}),
\]
where $Y_i(j)$ for $j = 0,1$ are potential permanent incomes and $X_{1i}$ is the permanent income of the parent. For IOp let $\gamma^{(k)}(x_i) = \mathbb{E}[Y_i(k)|X_i = x_i]$ for $k = 0,1$ and set
\[
    g(y_i(a),x_i,y_j(b),x_j,\gamma^{(a)},\gamma^{(b)}) = |\gamma^{(a)}(x_i) - \gamma^{(b)}(x_j)|.
\]
Under selection on observables and common support assumptions, \cite{terschuur2023} provides doubly robust moments for such a family of causal U-statistics. While the methodological innovation in \cite{terschuur2023} is in the estimation of the optimal allocation rule $\pi$ and its statistical properties, the results on estimating the causal estimands for a given $\pi_{ab}$ rely on the general theory in this paper.

\section{Cross-Fitting}
\label{app_CF}

\begin{proof}[Proof of Proposition \ref{prop_UCF}]
    For (i) notice that $\hat{\theta}$ solves $\sum_{l=1}^L
\sum_{(i,j)\in I_l}
\psi(W_i,W_j,\hat{\gamma}_l,\hat{\alpha}_l,\theta)=0$. Hence, all blocks $I_l$ with $l = 1,..,L$ appear once in the sum. Also, the sum for a fixed $l$, sums once across all pairs in $I_l$. Hence, if $I_1,...,I_L$ forms a partition of $\mathcal{P}$, then each pair in $\mathcal{P}$ appears exactly once in the cross-fitted sum. The construction of $I_l$, $l =1,...,L$ ensures that it is a partition of $\mathcal{P}$. (ii) holds by the definition of $(\hat{\gamma}_l, \hat{\alpha}_l)$. (iii) For triangles, $|S_l| = |G_u|$ since a triangle takes the form $\{(i,j) \in \mathcal{P}: i < j, i \in G_u, j \in G_u\}$. $|G_u| = n/T$, so a triangle leaves $n - n/T = n(T-1)/T$ observations for nuisance estimation. For the rectangles, we use $|G_u|/2$ observations for one element of the pair and $|G_v|/2$ for the other element, hence $|S_l| = |G_u|/2 + |G_v|/2 = n/2T + n/2T = n/T$ and hence $n(T-1)/T$ observations are left for nuisance estimation.(iv) if we set $T = n/2$, then each pair in $\mathcal{P}$ form its own block and leaves $n-2$ observations for estimation.
\end{proof}

The main task is to partition the set of $n(n-1)/2$ pairs $\{(i,j): i < j, i,j = 1,...,n\}$ into $L$ sets where the number of observations left out in each set is the same or almost the same. Here we provide a general version of Algorithm \ref{alg:cf_u} which achieves this up to a few observations for general $n$. Algorithm \ref{alg:cf_u_general_min} leaves approximately $n - \lfloor n/T \rfloor$ observations for nuisance parameter estimation in each block (exactly if $n$ is divisible by $2 \cdot T$). 

\begin{algorithm}[H]
\small
\caption{Cross-Fitting for U-Statistics (General Case)}
\label{alg:cf_u_general_min}
\begin{algorithmic}[1]

\Require Sample $\{W_i\}_{i=1}^n$, integer $T$ with $2 \leq T \le n/2$
\Ensure Cross-fitted estimator $\hat\theta$

\State Compute $\texttt{base} \gets \lfloor n/T \rfloor$ and $\texttt{rem} \gets n \bmod T$.
\State Partition $\{1,\dots,n\}$ into $T$ parts $G_1,\dots,G_T$ such that
\Statex \hspace{1.5em} $|G_t|=\texttt{base}+1$ for $t\le \texttt{rem}$ and $|G_t|=\texttt{base}$ for $t>\texttt{rem}$.

\State Initialize block index $l \gets 0$.

\For{$t=1$ to $T$}  \Comment{Within-fold blocks}
    \State $l \gets l+1$
    \State $I_l \gets \{(i,j): i<j,\ i,j \in G_t\}$
\EndFor

\For{$1 \le u < v \le T$}  \Comment{Between-fold blocks}
    \State Split $G_u$ into two (almost) equal subsets using the fixed order:
    \Statex \hspace{1.5em} $A_{uv} \gets \text{first }\lfloor |G_u|/2\rfloor \text{ elements of }G_u$, \;
    $A'_{uv} \gets G_u \setminus A_{uv}$.
    \State Split $G_v$ analogously:
    \Statex \hspace{1.5em} $B_{uv} \gets \text{first }\lfloor |G_v|/2\rfloor \text{ elements of }G_v$, \;
    $B'_{uv} \gets G_v \setminus B_{uv}$.

    \For{each combination $(A,B) \in \{A_{uv},A'_{uv}\} \times \{B_{uv},B'_{uv}\}$}
        \State $l \gets l+1$
        \State $I_l \gets \{(i,j): i \in A,\ j \in B\}$
    \EndFor
\EndFor

\State Let $L \gets l$.

\For{$l=1$ to $L$}
    \State Let $S_l$ be the set of unique indices appearing in $I_l$.
    \State Estimate $(\hat\gamma_l,\hat\alpha_l)$ using $\{W_i : i \notin S_l\}$.
\EndFor

\State Return $\hat\theta$ solving $\sum_{l=1}^L
\sum_{(i,j)\in I_l}
\psi(W_i,W_j,\hat{\gamma}_l,\hat{\alpha}_l,\theta)=0.$
\end{algorithmic}
\end{algorithm}

\noindent  In Figure \ref{fig_CF2}, we show an example with $n$ not divisible by $2\cdot T$. Now we prove a Lemma about the cross-fitting algorithm that is used in the proofs of the asymptotic theory. 

\begin{figure}[h]
\centering
\includegraphics[scale = 0.4]{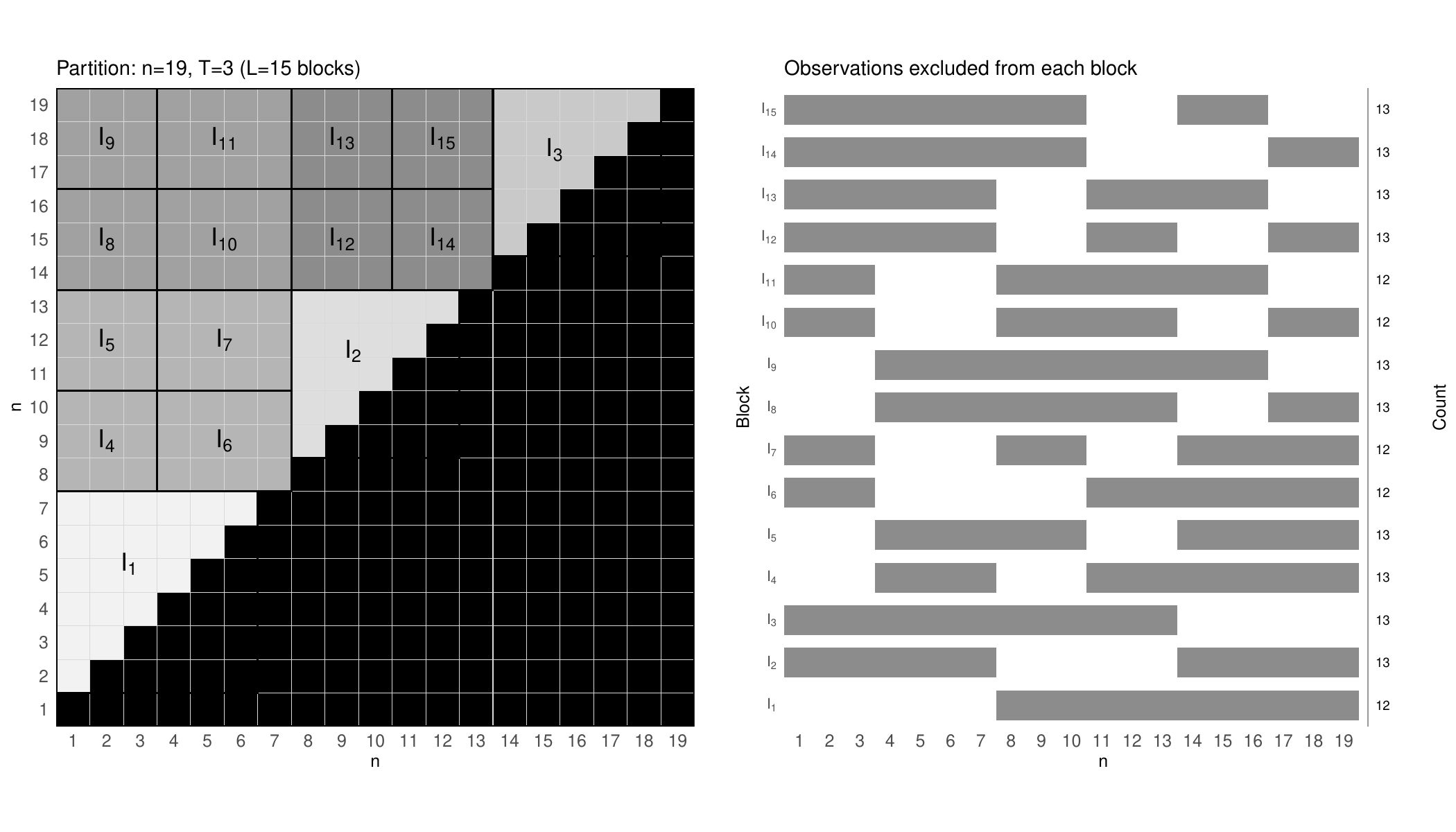}
\caption{\scriptsize Cross-fitting partition}%
\label{fig_CF2}%
\end{figure}

\begin{lemma}
\label{lemma_CF}
    For Algorithm \ref{alg:cf_u_general_min}, we have that for all $l = 1,...,L$, $|S_l|/|I_l| \to 0$ and $n \cdot |S_l|/|I_l| \to c < \infty$ as $n \to \infty$.
\end{lemma}
\begin{proof}
    Note that $T$ is fixed. If $l$ is such that $I_l$ is a triangle, for some $u = 1,...,T$ 
    \[
    |S_l|/ |I_l| = \frac{2 |S_l|}{|S_l| (|S_l| - 1)} = \frac{2}{(|S_l| - 1)} \to 0,
    \]
    and
    \[
    |S_l| = |G_u| = (\lfloor n/T \rfloor + 1(u \leq n \bmod T),
    \]
    so
    \[
    n|S_l|/ |I_l| = \frac{2n}{(\lfloor n/T \rfloor) - 1 + 1(u \leq n \bmod T) } \to 2T.
    \]
    If $l$ is such that $I_l$ is a rectangle, then $I_l=\{(i,j): i\in A,\ j\in B\}$ for some
    $A\subseteq G_u$ and $B\subseteq G_v$ with $u\neq v$, where $A$ and $B$ are one of the
    two halves constructed in Algorithm \ref{alg:cf_u_general_min}. Let
    $a:=|A|$ and $b:=|B|$. Then $|I_l| = ab$ and $|S_l| = a+b$. Therefore
    \[
    \frac{|S_l|}{|I_l|}=\frac{a+b}{ab}=\frac{1}{a}+\frac{1}{b}
    \le \frac{2}{\min(a,b)}\to 0,
    \]
    because $|G_u|\to\infty$, $|G_v|\to\infty$ as $n\to\infty$, and
    $\min(a,b) \geq \min( \lfloor |G_u|/2\rfloor,\lfloor |G_v|/2\rfloor) \to \infty$. Next,
    \[
    n\frac{|S_l|}{|I_l|}
    = n\Big(\frac{1}{a}+\frac{1}{b}\Big).
    \]
    Write $|G_u|=\lfloor n/T\rfloor + \mathbf 1(u\le n\bmod T)$ and similarly for $|G_v|$.
    Since $a = |G_u|/2 + O(1)$ and $b=|G_v|/2+O(1)$, we have
    \[
    \frac{n}{a}=\frac{n}{n/(2T)+O(1)}\to 2T,\qquad
    \frac{n}{b}=\frac{n}{n/(2T)+O(1)}\to 2T,
    \]
    hence $n|S_l|/|I_l|\to 4T.$ Hence, the lemma holds with $c=4T$.
\end{proof}

\section{Debiasing other inequality measures}
\label{app_Othermeasures}
In this section, we provide valid inference theory for a broad class of inequality measures applied to estimated objects. We consider this to be a contribution in its own right. We focus on the Atkinson index proposed in \cite{atkinson1970measurement} and on Generalized Entropy (see \cite{theil1967economics} and \cite{shorrocks1980class}).
\subsection{Atkinson index and Empirical Welfare}
The Atkinson index features an inequality aversion parameter $\varepsilon \geq 0$ 
\begin{equation}
    A(\varepsilon) =
    \begin{cases}
        1 -  \frac{\mathbb{E} \left[ Y_i ^{1-\varepsilon} \right]^{\frac{1}{1-\varepsilon}}}{\mathbb{E}[Y_i]} & \text{if } \varepsilon \neq 1,\infty \\
        1 - \frac{\exp(\mathbb{E}[\ln Y])}{\mathbb{E}[Y_i]} & \text{if } \varepsilon = 1, \\
        1 - \frac{\inf\{y:f(y) > 0\}}{\mathbb{E}[Y_i]} & \text{if } \varepsilon = \infty.
    \end{cases}
\end{equation}
$\varepsilon = 1$ and $\varepsilon = \infty$ are limiting cases, where $f(y)$ is the density of $Y_i$. We now focus on applying the Atkinson index to the estimated quantities $\hat{\gamma}(X_i)$ considered in the main text. Let us start with the case of a finite $\varepsilon$ different from 1. Our parameter of interest is:
\begin{align*}
    \theta_0(\varepsilon) &= 1 - \frac{\theta_{01}(\varepsilon)^{\frac{1}{1-\varepsilon}}}{\theta_{02}}, \quad \theta_{01}(\varepsilon) = \mathbb{E} [\gamma_0(X_i)^{1-\varepsilon}], \quad \theta_{02} = \mathbb{E}[\gamma_0(X_i)].
\end{align*}
The debiasing strategy involves constructing locally robust moments for both $\theta_{01}(\varepsilon)$ and $\theta_{02}$. For $\theta_{02}$, the locally robust moment is: $\mathbb{E}[Y_i - \theta_{02}]$, i.e. no need to use any ML, the average of $Y_i$ works. To debias $\theta_{01}(\varepsilon)$, following \cite{newey1994asymptotic}, we can write:
\begin{align*}
    \frac{d}{d\tau}\mathbb{E}[\gamma_\tau(X_i)^{1-\varepsilon}] &= \frac{d}{d\tau}\mathbb{E}[(1-\varepsilon)\gamma_0(X_i)^{-\varepsilon}\gamma_\tau(X_i)] = \int (1-\varepsilon)\gamma_0(x)^{-\varepsilon}(y - \gamma_0(x)) H(dw).
\end{align*}
Hence, an orthogonal moment for $\theta_{01}(\varepsilon)$ is:
\[
\psi(w,\gamma_0,\alpha_0,\theta,\varepsilon) = \gamma(x)^{1-\varepsilon} + \alpha_0(x,\varepsilon)(y- \gamma_0(x)) - \theta,
\]
where $\alpha_0(x,\varepsilon) = (1-\varepsilon)\gamma_0(x)^{-\varepsilon}$. For regular identification we require that $\mathbb{E}[\gamma_0(X_i)^{-2\varepsilon}] < \infty$, which is a condition on the tail near zero of the distribution of $\gamma_0(X_i)$ and becomes more stringent the higher the inequality aversion $\varepsilon$. A debiased estimator of $\theta_{01}$ is then given by:
\[
\hat{\theta}_1(\varepsilon) = \frac{1}{n}\sum_{l = 1}^L \sum_{i \in C_l} \left( \hat{\gamma}_l(X_i)^{1-\varepsilon} + \hat{\alpha}_l(X_i,\varepsilon)(Y_i - \hat{\gamma}_l(X_i)) \right), \quad \hat{\alpha}_l(X_i,\varepsilon) = (1-\varepsilon)\hat{\gamma}_l(X_i)^{-\varepsilon},
\]
where for the crossfitting we have divided the sample in $L$ folds: $C_1,...,C_L$. Then, letting $\hat{\theta}_2 = \frac{1}{n}\sum_{i=1}^n Y_i$, we can construct the debiased Atkinson index as $\hat{\theta}(\varepsilon) = 1 - \hat{\theta}_1(\varepsilon)^{\frac{1}{1-\varepsilon}}/\hat{\theta}_2$. Under regularity conditions and similar assumptions as in the main text and letting $\pi = 1/(1-\varepsilon)$, one can use the Delta Method to show that:
\[
\sqrt{n}(\hat{\theta}(\varepsilon) - \theta_0(\varepsilon)) \to_d N\left(0, \left(\frac{\pi \theta_{01}^{\pi - 1}}{\theta_{02}}\right)^2 \Sigma_{11} + \frac{\theta_{01}^{2\pi}}{\theta_{01}^4}\Sigma_{22} - 2\frac{\pi \theta_{01}^{2\pi - 1}}{\theta_{02}^3}\Sigma_{12}\right),
\]
where $\Sigma_{11} =\mathbb{E}\left[\psi(W_i,\gamma_0,\alpha_0,\theta_{01},\varepsilon)^2\right]$, $\Sigma_{22} = \mathbb{E}\left[(Y_i - \theta_{02})^2\right]$, $\Sigma_{12} = \mathbb{E}\left[\psi(W_i,\gamma_0,\alpha_0,\theta_{01},\varepsilon)(Y_i - \theta_{02})\right].$
For the case of $\varepsilon = 1$, we need to construct a locally robust estimator of $\theta_{01}(1) = \mathbb{E}[\ln \gamma_0(X_i)]$, again using the arguments in \cite{newey1994asymptotic} and the chain rule,
\begin{align*}
    \frac{d}{d\tau} \mathbb{E}\left[\ln \gamma_\tau(X_i)\right] &= \frac{d}{d\tau}\mathbb{E}\left[\frac{1}{\gamma_0(X_i)}\gamma_\tau(X_i)\right] \\
    &= \int \frac{1}{\gamma_0(x)}(y - \gamma_0(x)) H(dw),
\end{align*}
hence, an orthogonal moment for $\mathbb{E}[\ln \gamma_0(X_i)]$ is:
\[
\psi(w,\gamma_0,\alpha_0,\theta,1) = \ln \gamma_0(x) + \alpha_0(x)(y - \gamma_0(x)) - \theta, \quad \alpha_0(x) = \frac{1}{\gamma_0(x)},
\]
A debiased estimator of $\mathbb{E}[\ln \gamma_0(X_i)]$ is then given by:
\[
\hat{\theta}_1(1) = \frac{1}{n}\sum_{l = 1}^L \sum_{i \in C_l} \left( \ln \hat{\gamma}_l(X_i) + \hat{\alpha}_l(X_i)(Y_i - \hat{\gamma}_l(X_i)) \right), \quad \hat{\alpha}_l(X_i) = \frac{1}{\hat{\gamma}_l(X_i)}.
\]
Then, we can construct the debiased Atkinson index for $\varepsilon = 1$ as $\hat{\theta}(1) = 1 - \frac{\exp(\hat{\theta}_1(1))}{\hat{\theta}_2}.$ Under regularity conditions and assumptions in the main text, by Delta Method:
\[
\sqrt{n}(\hat{\theta}(1) - \theta_0(1)) \to_d N\left(0, \frac{\exp(2\theta_{01})}{\theta_{02}^2}\Sigma_{11} + \frac{\exp(2\theta_{01})}{\theta_{02}^4}\Sigma_{22} - 2\frac{\exp(2\theta_{01})}{\theta_{02}^3}\Sigma_{12}\right),
\]
where the variance terms are as before but with $\varepsilon=1$. The case of infinite inequality aversion $\varepsilon = \infty$ is hard to deal with even when applied to observed outcomes. In the case of estimated objects, the problem becomes harder due to the lack of differentiability of the $\inf$ functional. An important application is to the measurement of welfare based on Constant Relative Risk Aversion (CRRA) utility functions which is a cornerstone of welfare economics.
\subsection{Generalized Entropy}
The generalized entropy index, parametrized by a scalar $\beta$, is defined as follows:
\[
GE(\beta) =
\begin{cases}
    \frac{1}{\beta(\beta-1)}\left( \frac{\mathbb{E}[Y_i^\beta]}{\mathbb{E}[Y_i]^\beta} - 1\right) & \text{if } \beta \neq 0,1 \\
    \ln \mathbb{E}[Y_i] -  \mathbb{E}[\ln Y_i] & \text{if } \beta = 0, \\
    \frac{\mathbb{E}[Y_i \ln(Y_i)] - \ln\mathbb{E}[Y_i]}{\mathbb{E}[Y_i]} & \text{if } \beta = 1.
\end{cases}
\]
$\beta = 0$ and $\beta = 1$ are known as the Mean Logarithmic Deviation and the Theil index respectively. $\beta = 2$ is half the squared coefficient of variation. Let $\theta_0(\beta)$ be the generalized entropy index applied to our first-step estimands $\gamma_0(X_i)$. For $\beta \neq 0,1$, we need locally robust moments for $\theta_{01}(\beta) = \mathbb{E}[\gamma_0(X_i)^\beta]$ and $\theta_{02} = \mathbb{E}[\gamma_0(X_i)]$. Employing the same arguments as in the Atkinson case, we get a locally robust moment for $\theta_{01}(\beta)$ as:
\[
\psi(w,\gamma_0,\alpha_0,\theta,\beta) = \gamma_0(x)^\beta + \alpha_0(x,\beta)(y - \gamma_0(x)) - \theta, \quad \alpha_0(x,\beta) = \beta\gamma_0(x)^{\beta-1}.
\]
An estimator robust to first-stage bias is given by:
\[
\hat{\theta}_1(\beta) = \frac{1}{n}\sum_{l = 1}^L \sum_{i \in C_l} \left( \hat{\gamma}_l(X_i)^\beta + \hat{\alpha}_l(X_i,\beta)(Y_i - \hat{\gamma}_l(X_i)) \right), \quad \hat{\alpha}_l(X_i,\beta) = \beta\hat{\gamma}_l(X_i)^{\beta-1}.
\]
Then, we can construct the debiased generalized entropy index as:
\[
\hat{\theta}(\beta) = \frac{1}{\beta(\beta-1)}\left( \frac{\hat{\theta}_1(\beta)}{\hat{\theta}_2^\beta} - 1\right), \quad \hat{\theta}_2 = \frac{1}{n}\sum_{i=1}^n Y_i.
\]
Under regularity conditions and similar assumptions as in the main text, by Delta Method:
\[
\sqrt{n}(\hat{\theta}(\beta) - \theta_0(\beta)) \to_d N\left(0, \frac{1}{\beta^2 (1 - \beta)^2 \theta_{02}^{2\beta}} \Sigma_{11}
+ \frac{\theta_{01}^2}{(1 - \beta)^2 \theta_{02}^{2\beta + 2}} \Sigma_{22}
- \frac{2 \theta_{01}}{\beta (1 - \beta)^2 \theta_{02}^{2\beta + 1}} \Sigma_{12}\right),
\]
where $\Sigma_{11}=\mathbb{E}\left[\psi(W_i,\gamma_0,\alpha_0,\theta_{01},\beta)^2\right]$, $\Sigma_{22} = \mathbb{E}\left[(Y_i - \theta_{02})^2\right]$, $    \Sigma_{12} = \mathbb{E}\left[\psi(W_i,\gamma_0,\alpha_0,\theta_{01},\beta)(Y_i - \theta_{02})\right]$. For the case of $\beta = 0$, we need to construct a locally robust estimator of $\theta_{01}(0) = \mathbb{E}[\ln \gamma_0(X_i)]$. Again, using the arguments in \cite{newey1994asymptotic} and the chain rule, we can write:
\begin{align*}
    \frac{d}{d\tau} \mathbb{E}\left[\ln \gamma_\tau(X_i)\right] &= \frac{d}{d\tau}\mathbb{E}\left[\frac{1}{\gamma_0(X_i)}\gamma_\tau(X_i)\right] = \int \frac{1}{\gamma_0(x)}(y - \gamma_0(x)) H(dw),
\end{align*}
hence, an orthogonal moment for $\mathbb{E}[\ln \gamma_0(X_i)]$ is:
\[
\psi(w,\gamma_0,\alpha_0,\theta,0) = \ln \gamma_0(x) + \alpha_0(x)(y - \gamma_0(x)) - \theta, \quad \alpha_0(x) = \frac{1}{\gamma_0(x)}.
\]
A debiased estimator of $\mathbb{E}[\ln \gamma_0(X_i)]$ is then given by:
\[
\hat{\theta}_1(0) = \frac{1}{n}\sum_{l = 1}^L \sum_{i \in C_l} \left( \ln \hat{\gamma}_l(X_i) + \hat{\alpha}_l(X_i)(Y_i - \hat{\gamma}_l(X_i)) \right), \quad \hat{\alpha}_l(X_i) = \frac{1}{\hat{\gamma}_l(X_i)}.
\]
A debiased generalized entropy index for $\beta = 0$ is $\hat{\theta}(0) = \ln \hat{\theta}_2 - \hat{\theta}_1(0)$. Under regularity conditions and similar assumptions as in the main text, by Delta Method:
\[
\sqrt{n}(\hat{\theta}(0) - \theta_0(0)) \to_d N\left(0, \Sigma_{11} + \frac{\Sigma_{22}}{\theta_{02}^2} - 2\frac{\Sigma_{12}}{\theta_{02}}\right),
\]
where the variance terms are as before but with $\beta = 0$.
Finally, for the case of $\beta = 1$, we need to construct a locally robust estimator of $\theta_{01}(1) = \mathbb{E}[\gamma_0(X_i) \ln \gamma_0(X_i)]$. Again, using the arguments in \cite{newey1994asymptotic} and the chain rule, we can write:
\begin{align*}
    \frac{d}{d\tau} \mathbb{E}\left[\gamma_\tau(X_i) \ln \gamma_\tau(X_i)\right] &= \frac{d}{d\tau}\mathbb{E}\left[\gamma_0(X_i) \ln \gamma_\tau(X_i) \right] + \frac{d}{d\tau}\mathbb{E}\left[\ln \gamma_0(X_i) \gamma_\tau(X_i) \right] \\
    &= \frac{d}{d\tau} \mathbb{E}\left[\gamma_\tau(X_i) \right] + \frac{d}{d\tau}\mathbb{E}\left[\ln \gamma_0(X_i) \gamma_\tau(X_i) \right] \\
    &= \int (1+\ln\gamma_0(X_i))(y - \gamma_0(x)) H(dw).
\end{align*}
Hence, an orthogonal moment for $\mathbb{E}[\gamma_0(X_i) \ln \gamma_0(X_i)]$ is:
\[
\psi(w,\gamma_0,\alpha_0,\theta,1) = \gamma_0(x) \ln \gamma_0(x) + \alpha_0(x)(y - \gamma_0(x)) - \theta, \quad \alpha_0(x) = 1 + \ln \gamma_0(x).
\]
A debiased estimator of $\mathbb{E}[\gamma_0(X_i) \ln \gamma_0(X_i)]$ is then given by:
\[
\hat{\theta}_1(1) = \frac{1}{n}\sum_{l = 1}^L \sum_{i \in C_l} \left( \hat{\gamma}_l(X_i) \ln \hat{\gamma}_l(X_i) + \hat{\alpha}_l(X_i)(Y_i - \hat{\gamma}_l(X_i)) \right), \quad \hat{\alpha}_l(X_i) = 1 + \ln \hat{\gamma}_l(X_i).
\]
Hence, a debiased generalized entropy index estimator for $\beta = 1$ is $\hat{\theta}(1) = \frac{\hat{\theta}_1(1) - \ln \hat{\theta}_2}{\hat{\theta}_2}.$ Under regularity conditions and assumptions in the main text, by the Delta Method:
\[
\sqrt{n}(\hat{\theta}(1) - \theta_0(1)) \to_d N\left(0, \frac{1}{\theta_{02}^2} \Sigma_{11}
+ \frac{(\ln \theta_{02} - \theta_{01} - 1)^2}{\theta_{02}^4} \Sigma_{22}
+ \frac{2 (\ln \theta_{02} - \theta_{01} - 1)}{\theta_{02}^3} \Sigma_{12}\right),
\]
where the variance terms are as before but with $\beta = 1$.

\section{Extension to Higher-Order U-Statistics}
\label{app_r_order}

For simplicity, the paper focuses on second-order U-statistics. This section briefly discusses extensions to fixed-order $r$-sample U-statistics with $r\geq2$. The case $r=2$ corresponds to the pairwise setting studied in the main text: then
\[
\mathcal I_n^{(2)}
=
\{(i,j):1\le i<j\le n\}
=
\mathcal P,
\]
the kernel reduces to $g(W_i,W_j,\gamma_0,\theta)$, and the block partition below becomes the triangle/rectangle construction of Algorithm \ref{alg:cf_u}.

Let
\[
\theta_0
=
\mathbb E[g(W_1,\ldots,W_r,\gamma_0,\theta_0)],
\]
where $g$ is symmetric in $(W_1,\ldots,W_r)$ and $\gamma_0$ is an unknown nuisance function satisfying $\gamma_0(x)=\mathbb E[Y\mid X=x]$. Suppose the pathwise derivative admits the analogue of Assumption \ref{Ass_linearization}:
\[
\frac{d}{d\tau}
\mathbb E[g(W_1,\ldots,W_r,\gamma_\tau,\theta)]
=
\sum_{k=1}^{r}
\sum_{m=1}^{M}
c_{km}
\frac{d}{d\tau}
\mathbb E\!\left[
\delta_m(W_1,\ldots,W_r,\gamma_0,\theta)
\gamma_\tau(X_k)
\right],
\]
for square-integrable functions
$\delta_m(W_1,\ldots,W_r,\gamma_0,\theta)$, $m=1,\ldots,M$,
and constants $c_{km}$. Applying the projection argument of Lemma \ref{qorth2} term-by-term yields representers $\alpha_{0km}$ and the orthogonal correction
\[
\phi(W_1,\ldots,W_r,\gamma_0,\alpha_0,\theta)
=
\sum_{k=1}^{r}
\sum_{m=1}^{M}
c_{km}
\alpha_{0km}(X_1,\ldots,X_r,\theta)
\bigl(Y_k-\gamma_0(X_k)\bigr).
\]
Hence, $\psi=g+\phi$ is Neyman orthogonal. For $r=2$, this reduces to the construction in Lemma \ref{qorth2}.

\paragraph{Cross-fitting.}
The cross-fitting construction extends naturally to fixed $r>2$. Let
\[
\mathcal I_n^{(r)}
=
\{(i_1,\ldots,i_r):1\leq i_1<\cdots<i_r\leq n\}.
\]
Partition the observations into $T$ groups $G_1,\ldots,G_T$. Each block $I_l\subset\mathcal I_n^{(r)}$ is constructed from tuples whose indices lie in a prescribed collection of subgroups. For each block, define
\[
S_l
=
\{i:\ i \text{ appears in some } (i_1,\ldots,i_r)\in I_l\}.
\]
The nuisance estimators used for block $I_l$ are estimated using only observations in $S_l^c$. The cross-fitted estimator then solves
\[
\sum_{l=1}^{L}
\sum_{(i_1,\ldots,i_r)\in I_l}
\psi(W_{i_1},\ldots,W_{i_r},
\hat\gamma_l,\hat\alpha_l,\theta)
=
0.
\]

For $r=3$, the construction is the direct analogue of the triangle/rectangle decomposition for $r=2$. The set of triples is partitioned into three types of blocks:
\[
\{(i,j,k): i<j<k,\ i,j,k\in G_u\},
\]
corresponding to within-group triples;
\[
\{(i,j,k): i,j\in G_u,\ k\in G_v,\ u\neq v\},
\]
corresponding to two-one blocks; and
\[
\{(i,j,k): i\in G_u,\ j\in G_v,\ k\in G_w,\ u,v,w
\text{ distinct}\},
\]
corresponding to three-way rectangular blocks. As in Algorithm \ref{alg:cf_u}, the rectangular blocks can be subdivided into equal-size sub-blocks so that each block leaves the same number of observations for nuisance estimation. For example, a three-way block based on $G_u\times G_v\times G_w$ can be split by partitioning each group into two halves and taking the resulting $2^3$ sub-rectangles. Similarly, a two-one block can be split by dividing the group contributing two observations and the group contributing one observation so that all resulting blocks use the same number of observations.

This construction ensures the three key properties used in the second-order case: each $r$-tuple is used exactly once; the nuisance estimators for a block are computed from observations outside that block; and the number of observations left for nuisance estimation is balanced across blocks up to the same divisibility issues already present when $r=2$.

\paragraph{Asymptotic analysis.}
The asymptotic theory proceeds in two steps. First, one establishes that the cross-fitted orthogonal estimator is asymptotically equivalent to the infeasible oracle statistic with known nuisance functions. The arguments are analogous to those developed in the second-order case. In particular, Neyman orthogonality removes first-order nuisance estimation effects, while cross-fitting guarantees sufficient conditional independence between the estimating equations and the nuisance estimators. The main additional bookkeeping comes from the larger number of overlapping-index configurations among $r$-tuples, which can be handled using standard U-statistic arguments; see \citet{lee2019u}.

Second, once oracle equivalence is established, asymptotic normality follows from the Hoeffding decomposition for higher-order U-statistics. This second step is standard and relies on the first nondegenerate Hoeffding projection. Since the applications considered in this paper naturally lead to second-order structures, we focus on the practically most relevant case $r=2$.
\putbib[references]   
\end{bibunit}

\end{document}